%% file: tocl-article.tex
\documentclass[manuscript,screen,review=false]{acmart}

\usepackage{./sty/leosstylefile-acm}
\usepackage{./sty/tikzit}
\usetikzlibrary{positioning,intersections,quotes}
\input{./sty/layered-journal.tikzstyles}
\input{./sty/tikz-ct.tikzstyles}
\input{./sty/macros.tex}

\input{./sty/circuits.tikzdefs}
\input{./sty/circuits.tikzstyles}
\input{./sty/zx.tikzdefs}
\input{./sty/zx.tikzstyles}

\begin{document}

\title{Layered Monoidal Theories I: Diagrammatic Algebra and Applications}

\author{Leo Lobski}
\email{leo.lobski.21@ucl.ac.uk}
\orcid{0000-0002-0260-0240}
\author{Fabio Zanasi}
\email{f.zanasi@ucl.ac.uk}
\orcid{0000-0001-6457-1345}
\affiliation{%
  \institution{University College London}
  \city{London}
  \country{United Kingdom}
}

\begin{abstract}
We develop {\em layered monoidal theories} -- a generalisation of monoidal theories combining formal descriptions of a system at different levels of abstraction. Via their representation as {\em string diagrams}, monoidal theories provide a graphical formalism to reason algebraically about information flow in models across different fields of science. {\em Layered} monoidal theories allow mixing several monoidal theories (together with translations between them) within the same string diagram, while retaining mathematical precision and semantic interpretability. We develop the mathematical foundations of layered monoidal theories, as well as providing several instances of our approach, including digital and electrical circuits, quantum processes, chemical reactions, concurrent processes, and probability theory.
\end{abstract}




\maketitle

\section{Introduction}\label{sec:layered-introduction}
\input{./tex/tocl-article/intro}

\section{Monoidal theories}\label{sec:monoidal-theories}
\input{./tex/tocl-article/monoidal-theories-models}

\section{Layered theories}\label{sec:syntax-layered-theories}
\input{./tex/tocl-article/layered-theories}

\section{Functor boxes and coboxes}\label{sec:functor-boxes}
\input{./tex/tocl-article/functor-boxes}

\section{Case studies}\label{sec:layered-examples}
\input{./tex/tocl-article/layered-examples}

\section{Discussion and future work}\label{sec:layered-discussion}
\input{./tex/tocl-article/layered-discussion}

\begin{acks}
The authors thank the anonymous referees for their helpful suggestions, and acknowledge support from the ARIA Safeguarded AI Programme.
\end{acks}

\bibliographystyle{ACM-Reference-Format}
\bibliography{bibliography}

\appendix

\section{Proofs for Subsection~\ref{subsec:ccs}}\label{sec:ccs-proofs}
\input{./tex/tocl-article/appendix/ccs-proofs}

\section{Parametric monoidal theories}\label{sec:para-copara}
\input{./tex/tocl-article/appendix/para-copara}

\section{ZX-calculus and measurement based quantum computing}\label{sec:zx-mbqc}
\input{./tex/tocl-article/appendix/zx-mbqc}

\end{document}

%% file: sty/layered-journal.tikzstyles
\tikzstyle{node}=[fill=black, draw=black, shape=circle]
\tikzstyle{wnode}=[fill=white, draw=black, shape=circle]
\tikzstyle{textbox}=[inner sep=2pt, shape=rectangle, fill=none]
\tikzstyle{textnode}=[inner sep=0mm, shape=circle, fill=white]
\tikzstyle{gnode}=[inner sep=0mm, minimum size=1mm, fill={rgb,255: red,221; green,221; blue,221}, draw={rgb,255: red,221; green,221; blue,221}, shape=circle]
\tikzstyle{refine}=[fill=black, draw=black, shape=regular polygon, regular polygon sides=3, rotate=180, scale=0.5]
\tikzstyle{coarsen}=[fill=white, draw=black, shape=regular polygon, regular polygon sides=3, scale=0.5]
\tikzstyle{bdytextbox}=[fill=white, draw=black, shape=rectangle]
\tikzstyle{redbox}=[fill=white, draw=red, shape=rectangle, text=red]
\tikzstyle{bluecirc}=[inner sep=1mm, fill=white, draw={rgb,255: red,4; green,51; blue,255}, shape=circle, text={rgb,255: red,4; green,51; blue,255}]
\tikzstyle{rednode}=[fill=red, draw=red, shape=circle, scale=0.5]
\tikzstyle{new style 0}=[fill=white, draw=red, shape=circle, scale=0.5]
\tikzstyle{bluenode}=[fill={rgb,255: red,4; green,51; blue,255}, draw={rgb,255: red,4; green,51; blue,255}, shape=circle, scale=0.5]
\tikzstyle{yellownode}=[fill={rgb,255: red,255; green,210; blue,75}, draw={rgb,255: red,255; green,210; blue,75}, shape=circle, scale=0.5]
\tikzstyle{blacksq}=[fill=black, draw=black, shape=rectangle, scale=0.5]
\tikzstyle{bluetext}=[fill=none, draw=none, shape=rectangle, text={rgb,255: red,4; green,51; blue,255}]
\tikzstyle{reg}=[draw, fill=white, rounded rectangle, rounded rectangle left arc=none, minimum height=1em, minimum width=1em, node font={\scriptsize}, scale=1.25]
\tikzstyle{coreg}=[draw, fill=white, rounded rectangle, rounded rectangle right arc=none, minimum height=1em, minimum width=1em, node font={\scriptsize}, scale=1.25]
\tikzstyle{turquoisenode}=[fill={rgb,255: red,115; green,255; blue,239}, draw=black, shape=circle, scale=0.5]
\tikzstyle{resistor}=[R]
\tikzstyle{inductor}=[L]
\tikzstyle{capacitor}=[C]
\tikzstyle{voltage-source}=[american voltage source]
\tikzstyle{current-source}=[american current source]
\tikzstyle{togray}=[anchor=center]
\tikzstyle{snode}=[fill=black, draw=black, shape=circle, scale=0.5]
\tikzstyle{swnode}=[fill=white, draw=black, shape=circle, scale=0.5]

\tikzstyle{edge}=[-, draw=black]
\tikzstyle{diredge}=[->, draw=black]
\tikzstyle{dashed edge}=[-, draw=black]
\tikzstyle{loosedash}=[-, dash pattern=on 1.5pt off 3pt, draw=black]
\tikzstyle{dirdash}=[->, dashed, dash pattern=on 2pt off 0.5pt, draw=black]
\tikzstyle{mapsto}=[{|->}, draw=black]
\tikzstyle{gray diredge}=[draw={rgb,255: red,221; green,221; blue,221}, ->]
\tikzstyle{dark grey dirdash}=[->, dashed, dash pattern=on 2pt off 0.5pt, draw={rgb,255: red,81; green,81; blue,81}]
\tikzstyle{doubedge}=[-, draw=black, double=none, double distance=3pt, inner sep=0pt, thick]
\tikzstyle{thedge}=[-, line width=1.5pt, draw=black]
\tikzstyle{gray dashed}=[-, dashed, dash pattern=on 1pt off 1.5pt, draw={rgb,255: red,128; green,128; blue,128}]
\tikzstyle{rededge}=[-, draw=red]
\tikzstyle{gray edge}=[-, draw={rgb,255: red,128; green,128; blue,128}]
\tikzstyle{blthedge}=[-, thick, draw={rgb,255: red,4; green,51; blue,255}]
\tikzstyle{blthedge-extend}=[-, thick, draw={rgb,255: red,4; green,51; blue,255}, shorten >= -0.3pt, shorten <= -0.3pt]
\tikzstyle{blthdash}=[-, dashed, dash pattern=on 3pt off 1pt, thick, draw={rgb,255: red,4; green,51; blue,255}]
\tikzstyle{dirrededge}=[draw=red, ->]

%% file: sty/tikz-ct.tikzstyles
\tikzstyle{object}=[inner sep=0mm, shape=circle, fill=none]
\tikzstyle{bullet}=[fill=black, draw=black, shape=circle, scale=0.3]
\tikzstyle{circ}=[fill=white, draw=black, shape=circle, scale=0.3]
\tikzstyle{objectbox}=[inner sep=3pt, shape=rectangle, fill=none]
\tikzstyle{bdyobjectbox}=[fill=white, draw=black, shape=rectangle]

\tikzstyle{morphism}=[->, draw=black]
\tikzstyle{dash morphism}=[->, dashed, dash pattern=on 1.5pt off 1pt, draw=black]
\tikzstyle{mapsto}=[{|->}, draw=black]
\tikzstyle{nat transf}=[-implies, double, double distance=3pt, thick]
\tikzstyle{gray nat transf}=[-implies, draw=gray, double, double distance=3pt, thick]
\tikzstyle{equality}=[-, double, double distance=3pt]
\tikzstyle{squig morphism}=[->, draw=black, line join=round, decorate, decoration={zigzag, segment length=4, amplitude=.9,post=lineto, post length=2pt}]
\tikzstyle{hookarrow}=[right hook->, draw=black]
\tikzstyle{hookarrowmirror}=[left hook->, draw=black]
\tikzstyle{dotted line}=[-, dotted, draw=black]

%% file: sty/macros.tex
\definecolor{electroblue}{RGB}{4,51,255}
\usepackage{siunitx}
\usepackage[american,siunitx]{circuitikz}
\ctikzset{
    color/.initial=electroblue,
    color=electroblue,
    bipoles/thickness=2.5,
    bipoles/length=1.55cm,
    sources/scale=0.8,
    capacitors/scale=0.9,
}
\tikzstyle{vsource}=[rmeter, t={\textsf{\tiny -- +}}] 
\tikzstyle{ammeter}=[rmeter, t={\textsf{A}}] 
\tikzstyle{vmeter}=[rmeter, t={\textsf{V}}] 
\tikzstyle{elecdot}=[circle,fill,inner sep=0.85pt]

\newcommand{\includegraffle}[1]{
  {\lower10pt\hbox{$\includegraphics[height=1cm]{tikz/#1.pdf}$}}
}

%% file: sty/circuits.tikzstyles
\tikzstyle{none}=[anchor=center]
\tikzstyle{vertex}=[anchor=center, fill=black, draw=black, shape=circle, minimum size=2.5mm, tikzit category=hypergraph, inner sep=0]
\tikzstyle{edge subgraph}=[fill=neutral, draw=black, shape=rectangle, font={\boxsize}, tikzit category=hypergraph, minimum width=8mm, minimum height=8mm, rounded corners=3mm, very thick, dashed]
\tikzstyle{port}=[minimum size=2.5mm, fill=none, draw=none, shape=circle, tikzit draw={rgb,255: red,154; green,154; blue,154}, anchor=center, tikzit fill=neutral]
\tikzstyle{product}=[fill=neutral, draw=black, shape=circle, scale=0.66, tikzit category=string diagram]
\tikzstyle{type}=[fill=none, draw=none, shape=circle, font={\large}, tikzit fill={rgb,255: red,37; green,193; blue,141}, inner sep=0, anchor=center, tikzit category=string diagram]
\tikzstyle{tiny box white}=[{\corners}, font={\boxsize}, fill=neutral, draw=black, shape=rectangle, tikzit category=string diagram, minimum width={\tinywidth}, minimum height={\tinywidth}, anchor=center, line width={\stringwidth}, tikzit fill=neutral, inner sep={\innersep}]
\tikzstyle{tiny box comb}=[{\corners}, font={\boxsize}, fill=comb, draw=black, shape=rectangle, tikzit category=string diagram, minimum width={\tinywidth}, minimum height={\tinywidth}, anchor=center, line width={\stringwidth}, tikzit fill=comb, inner sep={\innersep}]
\tikzstyle{tiny box seq}=[{\corners}, font={\boxsize}, fill=seq, draw=black, shape=rectangle, tikzit category=string diagram, minimum width={\tinywidth}, minimum height={\tinywidth}, anchor=center, line width={\stringwidth}, inner sep={\innersep}]
\tikzstyle{tiny signal seq}=[{\corners}, font={\boxsize}, shape=signal, signal to=west, signal pointer angle=110, fill=seq, draw=black, tikzit category=string diagram, minimum width=6mm, minimum height=5mm, anchor=center, line width={\stringwidth}, inner sep={\innersep}]
\tikzstyle{small box white}=[{\corners}, font={\boxsize}, fill=neutral, draw=black, shape=rectangle, minimum height={\smallwidth}, minimum width={\tinywidth}, tikzit category=string diagram, anchor=center, line width={\stringwidth}, inner sep={\innersep}]
\tikzstyle{small square box white}=[{\corners}, font={\boxsize}, fill=neutral, draw=black, shape=rectangle, minimum height={\smallwidth}, minimum width={\smallwidth}, tikzit category=string diagram, anchor=center, line width={\stringwidth}, inner sep={\innersep}]
\tikzstyle{medium box}=[{\corners}, font={\boxsize}, fill=neutral, draw=black, shape=rectangle, tikzit category=string diagram, minimum height={\mediumwidth}, minimum width={\smallwidth}, anchor=center, line width={\stringwidth}, tikzit fill=neutral, inner sep={\innersep}]
\tikzstyle{medium box white}=[{\corners}, font={\boxsize}, fill=neutral, draw=black, shape=rectangle, tikzit category=string diagram, minimum height={\mediumwidth}, minimum width={\smallwidth}, anchor=center, line width={\stringwidth}, tikzit fill=comb, inner sep={\innersep}]
\tikzstyle{medium box comb}=[{\corners}, font={\boxsize}, fill=comb, draw=black, shape=rectangle, tikzit category=string diagram, minimum height={\mediumwidth}, minimum width={\smallwidth}, anchor=center, line width={\stringwidth}, tikzit fill=comb, inner sep={\innersep}]
\tikzstyle{medium square box white}=[{\corners}, font={\boxsize}, fill=neutral, draw=black, shape=rectangle, tikzit category=string diagram, minimum height={\mediumwidth}, minimum width={\mediumwidth}, line width={\stringwidth}, tikzit fill=neutral, inner sep={\innersep}]
\tikzstyle{medium square box comb}=[{\corners}, font={\boxsize}, fill=comb, draw=black, shape=rectangle, tikzit category=string diagram, minimum height={\mediumwidth}, minimum width={\mediumwidth}, line width={\stringwidth}, tikzit fill=comb, inner sep={\innersep}]
\tikzstyle{medium square box seq}=[{\corners}, draw=black, font={\boxsize}, fill=seq, shape=rectangle, tikzit category=string diagram, minimum height={\mediumwidth}, minimum width={\mediumwidth}, line width={\stringwidth}, tikzit fill=seq, inner sep={\innersep}]
\tikzstyle{medium square rounded box white}=[rounded corners, font={\boxsize}, fill=neutral, draw=black, shape=rectangle, tikzit category=string diagram, minimum height={\mediumwidth}, minimum width={\mediumwidth}, line width={\stringwidth}, tikzit fill=neutral, inner sep={\innersep}]
\tikzstyle{medium square rounded box comb}=[rounded corners, font={\boxsize}, fill=comb, draw=black, shape=rectangle, tikzit category=string diagram, minimum height={\mediumwidth}, minimum width={\mediumwidth}, line width={\stringwidth}, tikzit fill=comb, inner sep={\innersep}]
\tikzstyle{medium square rounded box seq}=[rounded corners, draw=black, font={\boxsize}, fill=seq, shape=rectangle, tikzit category=string diagram, minimum height={\mediumwidth}, minimum width={\mediumwidth}, line width={\stringwidth}, tikzit fill=seq, inner sep={\innersep}]
\tikzstyle{large box}=[{\corners}, font={\boxsize}, fill=neutral, draw=black, shape=rectangle, tikzit category=string diagram, minimum height=15mm, minimum width=10mm, anchor=center, line width={\stringwidth}, tikzit fill=neutral, inner sep={\innersep}]
\tikzstyle{large square box white}=[{\corners}, font={\boxsize}, fill=neutral, draw=black, shape=rectangle, tikzit category=string diagram, minimum width=15mm, minimum height=15mm, line width={\stringwidth}, tikzit fill=neutral, inner sep={\innersep}]
\tikzstyle{large square box comb}=[{\corners}, font={\boxsize}, fill=comb, draw=black, shape=rectangle, tikzit category=string diagram, minimum width=12mm, minimum height=12mm, line width={\stringwidth}, tikzit fill=comb, inner sep={\innersep}]
\tikzstyle{huge box}=[{\corners}, fill=neutral, draw=black, shape=rectangle, minimum height=40mm, minimum width=15mm, anchor=center, tikzit category=string diagram, line width={\stringwidth}, tikzit fill=neutral, inner sep={\innersep}]
\tikzstyle{output-node}=[fill=neutral, draw=black, shape=circle, anchor=west, inner sep=0.5]
\tikzstyle{state}=[fill=neutral, draw=none, shape=rectangle, minimum height=10mm, minimum width=10mm]
\tikzstyle{delay}=[fill=neutral, draw=black, font={\boxsize}, shape=signal, tikzit category=string diagram, line width={\stringwidth}, minimum height={\tinywidth}, minimum width={\tinywidth}, inner sep=0.5mm, outer sep=0mm, minimum height=1em, minimum width=1em]
\tikzstyle{register}=[fill=seq, draw=black, font={\boxsize}, shape=signal, tikzit category=string diagram, line width={\stringwidth}, minimum height={2em}, minimum width={1.5em}, align=center, anchor=center, inner xsep=1mm, inner ysep=-1mm, outer xsep=0mm]
\tikzstyle{waveform}=[fill=seq, draw=black, font={\boxsize}, shape=signal, signal from=west, signal to=east, tikzit category=string diagram, line width={\stringwidth}, minimum height={2em}, minimum width={1.5em}, align=center, anchor=center, inner xsep=1mm, inner ysep=-1mm, outer xsep=0mm]
\tikzstyle{bproduct}=[fill=black, draw=black, shape=circle, scale=0.5, tikzit category=string diagram]
\tikzstyle{wproduct}=[fill=neutral, draw=black, shape=circle, scale=0.5, line width=0.75, tikzit category=string diagram]
\tikzstyle{gproduct}=[fill=unit, draw=unit, shape=circle, scale=0.5, tikzit category=string diagram]
\tikzstyle{bport}=[minimum size=2.5mm, fill=none, draw=none, shape=circle, tikzit draw={rgb,255: red,154; green,154; blue,154}, anchor=center, tikzit fill=neutral, tikzit category=string diagram]
\tikzstyle{mux}=[fill=neutral, draw=black, shape=trapezium, tikzit category=circuits, rotate=-90, minimum height=1em, line width={\stringwidth}, scale=1.25]
\tikzstyle{fork}=[fill=black, draw=black, shape=circle, tikzit category=circuits, scale=0.25]
\tikzstyle{box}=[fill=neutral, draw=black, shape=rectangle, tikzit category=circuits]
\tikzstyle{dangling}=[fill=none, draw=none, shape=circle, anchor=east, scale=0.01, tikzit category=string diagram, tikzit fill={rgb,255: red,162; green,76; blue,77}]
\tikzstyle{label}=[fill=none, draw=none, shape=circle, align=center, inner sep=0, outer sep=0]
\tikzstyle{small label}=[scale=0.75, fill=none, draw=none, shape=circle, outer sep=0, inner sep=0, anchor=center]
\tikzstyle{wire label left}=[scale=1.25, fill=none, draw=none, shape=rectangle, outer sep=0, inner sep=0, anchor=east]
\tikzstyle{wire label right}=[scale=1.25, fill=none, draw=none, shape=rectangle, outer sep=0, inner sep=0, anchor=west]
\tikzstyle{wire label mid}=[scale=1.25, fill=\bgcolour, shape=rectangle, outer sep=0, inner sep=0, shape=circle, minimum size=0]
\tikzstyle{tile}=[fill=neutral, draw=black, shape=rectangle, tikzit category=string diagram, {\corners}, minimum height=5mm, minimum width=5mm]
\tikzstyle{commuting label}=[fill=none, draw=none, shape=circle, scale=0.5]
\tikzstyle{and}=[fill=neutral, draw=black, shape=and gate US, line width={\stringwidth}, and gate, tikzit category=circuits]
\tikzstyle{or}=[fill=neutral, draw=black, or gate US, line width={\stringwidth}, tikzit category=circuits]
\tikzstyle{not}=[fill=neutral, draw=black, not gate US, line width={\stringwidth}, tikzit category=circuits]
\tikzstyle{nor}=[fill=neutral, draw=black, nor gate, scale=1.75, line width={\stringwidth}, tikzit category=circuits]
\tikzstyle{nand}=[fill=neutral, draw=black, nand gate, scale=1.75, line width={\stringwidth}, tikzit category=circuits]
\tikzstyle{xor}=[fill=neutral, draw=black, xor gate, scale=1.75, line width={\stringwidth}, tikzit category=circuits]
\tikzstyle{xnor}=[fill=neutral, draw=black, xnor gate, scale=1.75, line width={\stringwidth}, tikzit category=circuits]
\tikzstyle{west}=[fill=none, draw=none, shape=circle, anchor=west]
\tikzstyle{bundler}=[fill=neutral, draw=black, shape=rounded rectangle, minimum width=3em, rounded rectangle arc length=180, rotate=90, line width={\stringwidth}]
\tikzstyle{long bundler}=[fill=neutral, draw=black, shape=rounded rectangle, minimum width=5em, rounded rectangle arc length=180, rotate=90, line width={\stringwidth}]

\tikzstyle{interface}=[-, fill={rgb,255: red,238; green,238; blue,255}, dashed, draw={rgb,255: red,170; green,170; blue,225}, ultra thick]
\tikzstyle{graph}=[-, fill={rgb,255: red,238; green,238; blue,238}, draw={rgb,255: red,191; green,191; blue,191}, dashed, ultra thick, anchor=center]
\tikzstyle{tentacle}=[-, very thick]
\tikzstyle{wire}=[-, tikzit category=string diagram, line width={\stringwidth}]
\tikzstyle{empty}=[-, densely dashed, {\corners}, dash pattern=on 0pt off 1.25pt on 1.25pt, line width={\stringwidth}]
\tikzstyle{transition}=[->]
\tikzstyle{output}=[->, decorate, decoration=zigzag]
\tikzstyle{gate}=[-, fill=neutral]
\tikzstyle{thicc}=[-, line width=1]
\tikzstyle{strikethrough}=[-, decoration={markings, mark=at position 0.5 with {
        \draw [blue, thick, - ] (-0.15,-0.15);
    }}, postaction=decorate]
\tikzstyle{traced}=[-, densely dashed, draw=gray]
\tikzstyle{arrow}=[<-, line width={\stringwidth}]
\tikzstyle{arrow up}=[->, line width={\stringwidth}]
\tikzstyle{dashed arrow}=[<-, dashed]
\tikzstyle{unit wire}=[-, fill=none, draw=unit, line width={\stringwidth}]
\tikzstyle{interfacearrow}=[->, very thick, dashed]
\tikzstyle{tile none}=[-, draw=none, {\corners}, fill=none, line width={\stringwidth}]
\tikzstyle{tile white}=[-, {\corners}, fill=neutral, line width={\stringwidth}]
\tikzstyle{tile comb}=[-, {\corners}, fill=comb, line width={\stringwidth}]
\tikzstyle{tile seq}=[-, {\corners}, fill=seq, line width={\stringwidth}]
\tikzstyle{commute}=[->]
\tikzstyle{curved rectangle}=[-, rounded corners]
\tikzstyle{hasse}=[-]
\tikzstyle{wiredash}=[-, line width={\stringwidth}]
\tikzstyle{rewrite}=[-, dashed, line width={\stringwidth}]
\tikzstyle{juxtaposition}=[-, draw={rgb,255: red,128; green,128; blue,128}, densely dashed, line width={\stringwidth}]
\tikzstyle{functor box}=[-, thick, draw={rgb,255: red,83; green,83; blue,83}]
\tikzstyle{boundary box}=[-, draw=none, tikzit draw={rgb,255: red,255; green,0; blue,4}]

%% file: sty/zx.tikzstyles
\tikzstyle{box}=[shape=rectangle, text height=1.5ex, text depth=0.25ex, yshift=0.5mm, fill=white, draw=black, minimum height=5mm, yshift=-0.5mm, minimum width=5mm, font={\small}]
\tikzstyle{Z dot}=[inner sep=0mm, minimum size=2mm, shape=circle, draw=black, fill={rgb,255: red,221; green,255; blue,221}]
\tikzstyle{Z phase dot}=[minimum size=1.2em, font={\footnotesize\boldmath}, shape=rectangle, rounded corners=0.5em, inner sep=0.2em, outer sep=-0.2em, scale=0.8, tikzit shape=circle, draw=black, fill={rgb,255: red,221; green,255; blue,221}, tikzit draw=blue]
\tikzstyle{X dot}=[Z dot, shape=circle, draw=black, fill={rgb,255: red,255; green,136; blue,136}]
\tikzstyle{X phase dot}=[Z phase dot, tikzit shape=circle, tikzit draw=blue, fill={rgb,255: red,255; green,136; blue,136}, font={\footnotesize\boldmath}]
\tikzstyle{hadamard}=[fill=yellow, draw=black, shape=rectangle, inner sep=0.6mm, minimum height=1.5mm, minimum width=1.5mm]
\tikzstyle{vertex}=[inner sep=0mm, minimum size=1mm, shape=circle, draw=black, fill=black]
\tikzstyle{vertex set}=[inner sep=0mm, minimum size=1mm, shape=circle, draw=black, fill=white, font={\footnotesize\boldmath}]
\tikzstyle{target}=[inner sep=0mm, minimum size=3mm, shape=circle, draw=black]

\tikzstyle{hadamard edge}=[-, dashed, dash pattern=on 2pt off 1.5pt, thick, draw={rgb,255: red,68; green,136; blue,255}]
\tikzstyle{brace edge}=[-, tikzit draw=blue, decorate, decoration={brace,amplitude=1mm,raise=-1mm}]
\tikzstyle{diredge}=[->]
\tikzstyle{dashed edge}=[-, dashed, dash pattern=on 2pt off 0.5pt, draw=black]

%% file: tex/tocl-article/intro.tex
Categorical modelling is a widespread approach for the formal description of systems across the natural sciences and engineering, with a particular emphasis on {\em compositionality} -- i.e.~ability to systematically deduce properties of the system from those of its constituent parts. In particular, {\em monoidal categories} (and their associated algebraic specification, {\em monoidal theories}) provide a well-developed and intuitive graphical syntax for many scientific theories via their representation as {\em string diagrams}~\cite{piedeleu-zanasi,selinger}. Monoidal theories have found applications, for instance, in quantum computing~\cite{zx-for2020,heunen-vicary-book}, linear algebra~\cite{zanasi-thesis,graphical-affine-algebra}, signal flow theory~\cite{survey-signal-flow}, electrical circuit theory~\cite{electrical-circuits}, digital circuit theory~\cite{kaye-thesis}, linguistics~\cite{discocat}, probability theory~\cite{fritz-markov20,jacobs-spr,JacobsZanasi-essentials} and causal reasoning~\cite{JacobsKZ21,kissinger19,lorenz-tull-causal}.

Categorical models often involve a translation between theories, for instance via an adjunction interpreting one theory in another, or a functor refining a more abstract theory by translating it to a more detailed one. Such translations are typically studied on a case-by-case basis, without a uniform mechanism allowing different theories to interact directly, e.g.~by describing them within a single category. However, such translations arise naturally, since systems can be described at different {\em levels of abstraction}. The levels correspond to distinct perspectives on the system, emphasising different details or features. It is common to regard some levels as {\em coarser} and others as more {\em fine-grained}, with translations  corresponding intuitively to ``zooming in'' or ``abstracting away'' the details. For example, reaction networks treat chemical compounds as placeholders with no internal structure, whereas the molecular structure is considered when modelling the compounds as graphs~\cite{chem-trans-motifs}. Similarly, higher level descriptions in computer architecture  must ultimately be implemented as microelectronic circuits~\cite{kaye-thesis}. The left side of Figure~\ref{fig:layers-of-abstraction} illustrates a translation from a coarse layer, consisting of a restricted set of English names for molecules, to a finer layer of molecular graphs equipped with rewriting rules. These examples demonstrate the ubiquity of mixed levels of abstraction across sciences: our aim is, therefore, to provide a general formalism for mathematical reasoning across different layers.

\begin{figure}
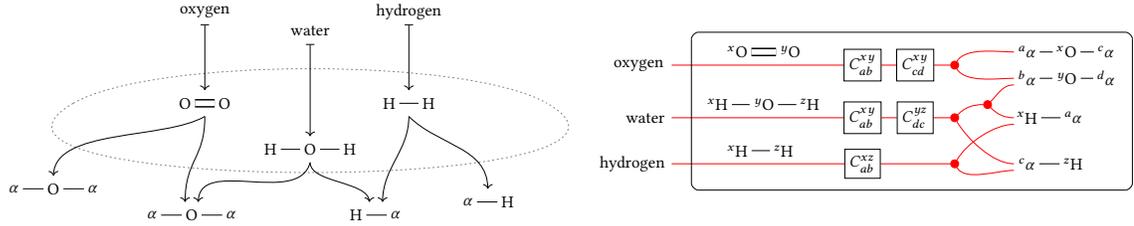

\begin{minipage}{0.49\textwidth}
\hspace{-35pt}\scalebox{.7}{\input{journal-figures/layers-of-abstraction.tikz}}
\end{minipage}\hfill
\begin{minipage}{0.49\textwidth}
\hfill\scalebox{.7}{\input{journal-figures/layers-of-abstraction-formally.tikz}}
\end{minipage}
\caption{Left: an informal translation from a coarser language (English names) to a finer language (molecular graphs). Right: A formalisation of said translation as a term in a layered monoidal theory.\label{fig:layers-of-abstraction}}
\end{figure}

The main contribution of this work is the introduction of {\em layered monoidal theories}, a mathematical framework where layers of abstractions for a wide range of phenomena may be studied algebraically within the same string diagrammatic formalism. Intuitively, a layered monoidal theory can be thought of as a ``glueing'' of different monoidal theories. Somewhat more formally, the data of a layered monoidal theory is equivalent to a diagram of (strict) monoidal categories and (strict) monoidal functors between them.

To establish some preliminary intuition, we show a typical term (morphism) in a layered monoidal theory in Figure~\ref{subfig:typical-term}. Here $x$ is a morphism in the category $\omega$ (drawn as a box on the left) and $y$ is a morphism in the category $\tau$ (drawn as a box on the right). The dotted lines indicate that $x$ and $y$ sit above $\omega$ and $\tau$, and that the ``functor boundary'' $\refine_f$ sits above $f$ -- these are formally not part of the term and are included for illustrative purposes only. The morphism $x$ can be, moreover, pushed through the functor boundary, as shown in Figure~\ref{subfig:functor-boundary}. Moreover, we allow functor boundaries $\coarsen_f$ in the opposite direction, yielding diagrams of the shape shown in Figure~\ref{subfig:window}. We call such diagrams {\em windows} (see Section~\ref{sec:functor-boxes}). Windows allow one to ``peek in'' at the semantics of the functor $f$ -- without performing the full translation. Furthermore, in many scenarios (e.g.~Sections~\ref{sec:functor-boxes} and~\ref{subsec:zx-extraction}), one encounters transformations between processes that are non-trivial, in the sense that they are not identities, thus genuinely modifying the process. Within layered monoidal theories, such transformations are modelled as {\em 2-cells}, drawn in Figure~\ref{subfig:twocell}.

\begin{figure}
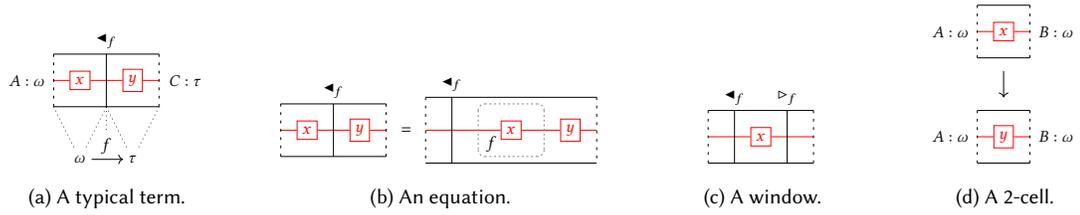

\centering
\begin{subfigure}{0.22\textwidth}
\centering
\scalebox{.7}{\input{journal-figures/term-layered-theory.tikz}}
\caption{A typical term.\label{subfig:typical-term}}
\end{subfigure}\hfill
\begin{subfigure}{0.34\textwidth}
\centering
\scalebox{.7}{\input{journal-figures/term-layered-pushing.tikz}}
\caption{An equation.\label{subfig:functor-boundary}}
\end{subfigure}\hfill
\begin{subfigure}{0.2\textwidth}
\centering
\scalebox{.7}{\input{journal-figures/window.tikz}}
\caption{A window.\label{subfig:window}}
\end{subfigure}\hfill
\begin{subfigure}{0.2\textwidth}
\centering
\scalebox{.7}{\input{journal-figures/twocell-int-xy.tikz}}
\caption{A 2-cell.\label{subfig:twocell}}
\end{subfigure}
\caption{Various features of layered monoidal theories.}
\end{figure}

Equipped with such formalism, we will study layered levels of abstractions in several different domains. In particular, the intuitive picture on the left of Figure~\ref{fig:layers-of-abstraction} will be translated into the layered string diagrams shown on the right in Figure~\ref{fig:layers-of-abstraction}, as discussed in Subsection~\ref{subsec:glucose}.

In summary, layered monoidal theories enjoy the following desirable properties:
\begin{itemize}
\item they provide a precise specification of multiple (potentially infinitely many) monoidal theories and functorial translations between them via finite, easily manageable syntax: see, for example, Figure~\ref{fig:dig-circ-eqns}, which gives the full recursive definition of the infinite layered monoidal theory of {\em simple arithmetic circuits}, expressing logical and arithmetic operations on $n$-bit wires in terms of single bit wires,
\item the construction is {\em modular}, in the sense that the internal terms contained between any two functor boundaries can always be interpreted within a single monoidal theory; this gives the possibility of {\em partial translations} of terms, which do no affect the other parts of the diagram,
\item the special case with two layers and one translation between them recovers standard monoidal {\em functorial semantics}; this functor is faithful (corresponding to completeness) if a certain single equation in the layered monoidal theory holds (Equation~\eqref{eq:cobox-pres-id}),
\item the existing notation for {\em functor boxes} naturally arises within the formalism, together with its dual notion, which we dub a {\em functor cobox}, hitherto appearing in the literature only informally,
\item our framework encompasses a variety of existing systems: variable amount of bits in digital circuits, impedance boxes in electrical circuits, rewriting for quantum circuit extraction, reaction mechanisms in chemistry, derivations in concurrency theory, conditionals in probability theory; whereas in the literature these systems have been treated as flat (single layer) monoidal theories, we use our framework to analyse their multilayered structure.
\end{itemize}

The study of the levels of abstraction is motivated by hierarchies of scientific theories. For instance, chemistry provides the vocabulary to talk about interactions between molecules, whereas particle physics lacks the language to describe interactions of such complexity. Similarly, biology can describe the function of a cell within an organism, while chemistry is restricted to describing reactions that are necessary to fulfil this function. Choosing the right level of abstraction is important for successfully interpreting one theory in another one. Our aim is to provide a mathematical formalism for doing so. At the same time, we explicitly place questions of philosophical commitment to reductionism or emergentism outside of the scope of the current work. Very roughly, reductionism holds that any higher level theory is ultimately reducible to some lower level (usually taken as a fundamental theory of physics), whereas emergentism maintains that  there are emergent, higher level phenomena which are lost under such reductions. Our formalism, insofar as possible, remains neutral and may be used to investigate either position.

\paragraph{Structure of the article} This work separates into four parts: (1) {\em Preliminaries} --- we recall monoidal theories in Section~\ref{sec:monoidal-theories}; (2) {\em Syntax}  --- we introduce three closely related classes of layered monoidal theories: {\em opfibrational}, {\em fibrational} and {\em deflational} theories in Section~\ref{sec:syntax-layered-theories}; (3) {\em Reasoning within layered monoidal theories} --- we define and study {\em functor boxes} and {\em coboxes} in Section~\ref{sec:functor-boxes}; (4) {\em Applications} --- we provide six extended case studies in Section~\ref{sec:layered-examples}. For a reader who wishes to see concrete examples first, it is certainly possible to follow the general ideas of the case studies in Section~\ref{sec:layered-examples} without reading the preceding theory. The syntax developed in Section~\ref{sec:syntax-layered-theories} can then be taken as a mathematical formalisation of such ideas.

Note that this paper contains no explicit discussion of semantic models of layered monoidal theories -- this is the topic of Part~II~\cite{lmt-part2}. This paper is based on the material presented in Chapters~3--6 of the first author's PhD thesis~\cite{lobski-phdthesis}, with additional material drawn from other chapters. Part of this work is based on a conference paper~\cite{lobski-zanasi}, which has been greatly expanded both in the theoretical contributions and in the applications.

%% file: tex/tocl-article/monoidal-theories-models.tex
Here we briefly define monoidal signatures and theories. Everything in this section is well-known, and is included in order to establish terminology and notation that will be built on in the subsequent sections.

A monoidal signature consists of a set of generating objects called {\em colours}, and for each pair of lists of colours, a set of generating morphisms called {\em monoidal generators}. Let us denote by $(-)^*:\Set\rightarrow\Mon$ the free monoid functor.
\begin{definition}[Monoidal signature]
A {\em monoidal signature} is a pair $(C,\Sigma)$ of a set $C$ and a function $\Sigma : C^*\times C^*\rightarrow\Set$.
\end{definition}
Given a monoidal signature $(C,\Sigma)$, the elements of $C$ are called {\em colours}, and for $a,b\in C^*$, the elements in $\Sigma(a,b)$ are the {\em monoidal generators} with arity $a$ and coarity $b$.

\begin{definition}[Morphism of monoidal signatures]
A {\em morphism} between two monoidal signatures $(C,\Sigma)\rightarrow (D,\Gamma)$ consists of a function $f:C\rightarrow D$, and for each pair $(a,b)\in C^*\times C^*$, a function $f_{a,b}:\Sigma(a,b)\rightarrow\Gamma\left(f^*a,f^*b\right)$.
\end{definition}
We denote the resulting category of monoidal signatures by $\MSgn$. We often denote a morphism in $\MSgn$ simply by $f:(C,\Sigma)\rightarrow (D,\Gamma)$ ranging over the whole family of functions: $f:C\rightarrow D$ (without subscripts) is the function on the colours, and $f_{a,b}:\Sigma(a,b)\rightarrow\Gamma(f^*a,f^*b)$ are the functions on the monoidal generators.

Note that there are two ways of ``forgetting'' the structure of $\MSgn$. There is the projection to the first component in the tuple defining the monoidal signature $\MSgn\rightarrow\Set$, and there is the functor into the category of families of sets $\MSgn\rightarrow\Fam(\Set)$ defined by $(C,\Sigma)\mapsto (C^*\times C^*,\Sigma)$, which ``forgets'' that morphisms need to preserve (co)arities. These functors arise in the following situation:

\noindent\begin{minipage}{0.6\textwidth}
\begin{proposition}\label{prop:msgn-pullback}
The square~\eqref{eq:msgn-pullback} is a pullback, where the right vertical map is the family fibration.
\end{proposition}
\end{minipage}
\begin{minipage}{0.4\textwidth}
\begin{equation}\label{eq:msgn-pullback}
\input{journal-figures/msgn-pullback.tikz}
\end{equation}
\end{minipage}
\begin{proof}
Let $F:\cat X\rightarrow\Set$ and $G:\cat X\rightarrow\Fam(\Set)$ be functors such that the resulting square commutes. Given an object $X\in\Ob(\cat X)$ and a morphism $u:X\rightarrow Y$ in $\cat X$, let us write $G(X) = (I_X,\Sigma_X)$ and $G(u) = (f_u,\{f^u_i\}_{i\in I_X})$. With this notation, commutativity means that $I_X = F(X)^*\times F(X)^*$ and $f_u = F(u)^*\times F(u)^*$ for all objects $X$ and morphisms $u$ of $\cat X$. We thus define the functor $H:\cat X\rightarrow\MSgn$ by $X\mapsto (F(X),\Sigma_X)$ on objects and by $u\mapsto (F(u),\{f^u_i\}_{i\in I_X})$ on morphisms. It follows by construction that $H$ composed with the forgetful functors gives $F$ and $G$. Uniqueness is immediate, as composition with $\MSgn\rightarrow\Set$ determines the colours and functions between them, whereas composition with $\MSgn\rightarrow\Fam(\Set)$ determines the generators and functions between them.
\end{proof}
\begin{corollary}\label{cor:forgetful-msgn-fibration}
The forgetful functor $\MSgn\rightarrow\Set$ is a fibration.
\end{corollary}
\begin{proof}
The family fibration $\Fam(\Set)\rightarrow\Set$ is a fibration, and fibrations are stable under pullbacks.
\end{proof}
We note that, due to uniqueness of pullbacks, we could take Proposition~\ref{prop:msgn-pullback} as the definition of $\MSgn$.

In order to equip monoidal signatures with non-trivial equations, we need a way to build more complicated expressions from the generators. The following defines the {\em terms} of a monoidal signature, which are used to define monoidal theories as well as free models. Upon quotienting by the structural identities (Definition~\ref{def:str-id}), the terms become what are known as {\em string diagrams}.

\begin{definition}[Terms of a monoidal signature]\label{def:terms-monoidal}
Given a monoidal signature $(C,\Sigma)$, the set of {\em sorts} is given by $C^*\times C^*$. The {\em terms} are generated by the recursive sorting procedure below:
\begin{center}
  \centering\scriptsize\noindent
  \begin{bprooftree}
    \AxiomC{$\sigma\in\Sigma(a,b)$}
    \RightLabel{\;\;}
    \UnaryInfC{$\scalebox{.8}{\input{journal-figures/sigmadiag.tikz}} : (a,b)$}
    \DisplayProof
    \AxiomC{$a\in C$}
    \RightLabel{\;\;}
    \UnaryInfC{$\scalebox{.8}{\input{journal-figures/iddiag-dashed.tikz}} : (a,a)$}
    \DisplayProof
    \AxiomC{\phantom{$a\in C$}}
    \RightLabel{\;\;}
    \UnaryInfC{$\scalebox{.8}{\input{journal-figures/emptydiag.tikz}} : (\varepsilon,\varepsilon)$}
    \DisplayProof
    \AxiomC{$\scalebox{.8}{\input{journal-figures/t-box.tikz}} : (a,b)$}
    \AxiomC{$\scalebox{.8}{\input{journal-figures/s-box.tikz}} : (b,c)$}
    \RightLabel{\;\;}
    \BinaryInfC{$\scalebox{.8}{\input{journal-figures/t-s-compose.tikz}} : (a,c)$}
    \DisplayProof
    \AxiomC{$\scalebox{.8}{\input{journal-figures/t-box.tikz}} : (a,b)$}
    \AxiomC{$\scalebox{.8}{\input{journal-figures/s-box.tikz}} : (c,d)$}
    \RightLabel{.}
    \BinaryInfC{$\scalebox{.8}{\input{journal-figures/t-s-tensor.tikz}} : (ac,bd)$}
  \end{bprooftree}
\end{center}
When using the linear notation, we denote the terms generated by the above rules by $\sigma$, $\id_a$, $\id_{\varepsilon}$, $(t;s)$ and $(t\otimes s)$, respectively. We denote the set of terms by $\Term_{C,\Sigma}$, and by $S:\Term_{C,\Sigma}\rightarrow C^*\times C^*$ the ``underlying sort function" mapping $t:(a,b)\mapsto (a,b)$.
\end{definition}

A pair of terms $(t,s)\in\Term_{C,\Sigma}\times\Term_{C,\Sigma}$ is {\em parallel} if $S(t)=S(s)$. The set of parallel pairs of terms is denoted by $P_{C,\Sigma}$.

Given a morphism $f:(C,\Sigma)\rightarrow (D,\Gamma)$ of monoidal signatures, it immediately extends to a function on terms, also denoted by $f:\Term_{C,\Sigma}\rightarrow\Term_{D,\Gamma}$, by recursively defining:\label{p:extend-morphism-signatures-terms}
\begin{center}
\begin{tabular}{c c c}
$\sigma : (a,b) \mapsto f_{a,b}(\sigma) : (f^*a,f^*b)$, & $\id_a : (a,a) \mapsto \id_{fa} : (fa,fa)$, & $\id_{\varepsilon} : (\varepsilon,\varepsilon) \mapsto \id_{\varepsilon} : (\varepsilon,\varepsilon)$, \\
$(t;s) : (a,c) \mapsto \left(f(t);f(s)\right) : (f^*a,f^*c)$, & \multicolumn{2}{c}{$(t\otimes s) : (ac,bd) \mapsto \left(f(t)\otimes f(s)\right) : (f^*(ac),f^*(bd))$.}
\end{tabular}
\end{center}
Observe that $f:\Term_{C,\Sigma}\rightarrow\Term_{D,\Gamma}$ preserves parallel terms: if $(t,s)\in P_{C,\Sigma}$, then $(ft,fs)\in P_{D,\Gamma}$.

\begin{definition}[Monoidal theory]\label{def:monoidal-theory}
A {\em monoidal theory} $\mathcal T$ is a triple $(C,\Sigma,E)$, where $(C,\Sigma)$ is a monoidal signature, and $E\sse P_{C,\Sigma}$ is a set of parallel terms. We refer to $E$ as the {\em equations} of $\mathcal T$.
\end{definition}

A morphism of monoidal theories $f:(C,\Sigma,E)\rightarrow (D,\Gamma,F)$ is given by a morphism of monoidal signatures $f:(C,\Sigma)\rightarrow (D,\Gamma)$ such that for all $(s,t)\in E$, we have $(fs,ft)\in F$. We denote the category of monoidal theories by $\MTh$.

\begin{definition}[Structural identities]\label{def:str-id}
Given a monoidal signature $(C,\Sigma)$, the set of {\em structural identities} $S$ is given by the following equations, where $s$, $s_i$ and $t_i$ range over the terms of the appropriate type:
\begin{equation*}
\scalebox{.6}{\input{journal-figures/structural-identities.tikz}}.
\end{equation*}
\end{definition}

\begin{definition}[Term congruence]
Given a monoidal theory $(C,\Sigma,E)$, the {\em term congruence} $\ \Eeq\ $ is the smallest equivalence relation on $\Term_{C,\Sigma}$ generated by $E\cup S$, which is a congruence with respect to $;$ and $\otimes$, i.e.~if $t_1\Eeq t_2$ and $s_1\Eeq s_2$, then $t_1;s_1\Eeq t_2;s_2$
and $t_1\otimes s_1\Eeq t_2\otimes s_2$, whenever the composition is defined.
\end{definition}

Note that the structural equations justify dropping all the dashed boxes when drawing the terms up to a term congruence: each diagram uniquely determines an equivalence class of terms. Such equivalence classes of terms are known as {\em string diagrams}. We utilise this in the following convention: if $w\in C^*$ with $w=a_1\cdots a_n$, we often abbreviate the identity term on $w$ as \scalebox{.6}{\input{journal-figures/const-id-terms.tikz}}; the structural identities guarantee that this expression unambiguously identifies an equivalence class of terms, no matter how the diagram is parsed.

\subsection{Examples of monoidal theories}\label{subsec:examples-monoidal-theories}

We give examples of monoidal theories that will be referenced later. Most of these are very well-known and commonly used. Monoidal theories with more involved generators and equations appear in Section~\ref{sec:layered-examples}, where we discuss several domain specific theories. An example of obtaining a monoidal theory from an existing one appears in Appendix~\ref{sec:para-copara}, where we present a construction that allows morphisms to be {\em parameterised}.

\begin{definition}[Symmetric monoidal theory]\label{def:symm-mon-thy}
A {\em symmetric monoidal theory} is a monoidal theory $(C,\Sigma,\mathcal S)$ such that for all $a,b\in C$, the set $\Sigma(ab,ba)$ contains the special generator, denoted by $\scalebox{.4}{\input{journal-figures/symmetry.tikz}}$, called the {\em symmetry}, and the set $\mathcal S$ contains the following equations, where $s$ ranges over all terms of the appropriate type, and the symmetry on the non-generator wires is defined recursively: \scalebox{.6}{\input{journal-figures/symmetry-eqns.tikz}}.
\end{definition}
Symmetric monoidal theories allow permuting the wires in any order -- working with a symmetric monoidal theory is in a sense equivalent to treating the (co)arities as multisets rather than lists. We remark that most of the concrete examples of monoidal theories we will see in this article will be symmetric.

\begin{definition}[Theory of monoids]\label{def:thy-mon}
The {\em theory of monoids} is the monoidal theory $(\{\bullet\},\Sigma,\mathcal M)$ with generators \scalebox{.5}{\input{journal-figures/monoid.tikz}} and equations \scalebox{.5}{\input{journal-figures/monoid-equations.tikz}}.
\end{definition}
We note that the models of the theory of monoids are monoidal categories with a chosen object with a monoid structure. In the monoidal category $(\Set,\times,1)$ of sets with the cartesian product, the models are precisely monoids in the sense of universal algebra.

\begin{definition}[Theory of comonoids]\label{def:thy-comon}
The {\em theory of comonoids} is the monoidal theory $(\{\bullet\},\Sigma,\mathcal C)$ with generators \scalebox{.5}{\input{journal-figures/comonoid.tikz}} and equations \scalebox{.5}{\input{journal-figures/comonoid-equations.tikz}}.
\end{definition}
Note that the theory of comonoids is the precise mirror symmetry of the theory of monoids: the generators are reflected horizontally by swapping the arities with the coarities. This makes the theory {\em coalgebraic} rather than algebraic.

\begin{definition}[Theory with uniform comonoids]\label{def:thy-univ-comonoids}
We say that a symmetric monoidal theory $(C,\Sigma,\mathcal U)$ has {\em uniform comonoids} if every $a\in C$ has a comonoid structure (i.e.~the sets $\Sigma(a,aa)$ and $\Sigma(a,\varepsilon)$ contain the comonoid generators and $\mathcal U$ contains the comonoid equations from Definition~\ref{def:thy-comon}), and upon extending the comonoid structure to all $w\in C^*$ by the following recursion:
\begin{equation*}
\scalebox{.6}{\input{journal-figures/comon-extend.tikz}},
\end{equation*}
the set $\mathcal U$ contains the following equations, where $s$ ranges over all terms with appropriate type:
\begin{equation*}
\scalebox{.6}{\input{journal-figures/comon-nat-eqns.tikz}}.
\end{equation*}
\end{definition}
It is well-known that a monoidal category is a model of the theory with uniform comonoids if and only if its monoidal product is {\em cartesian} -- this is known as {\em Fox's theorem}. We refer the reader e.g.~to~\cite[Section~6.4]{mellies-cat-semantics-linear-logic} for the details.

%% file: figures/journal-figures/sigmadiag.tikz
\begin{tikzpicture}
	\begin{pgfonlayer}{nodelayer}
		\node [style=none] (4) at (-0.75, 0) {};
		\node [style=none] (5) at (0.75, 0) {};
		\node [style=bdytextbox] (6) at (0, 0) {$\sigma$};
	\end{pgfonlayer}
	\begin{pgfonlayer}{edgelayer}
		\draw [style=edge] (4.center) to (6);
		\draw [style=edge] (6) to (5.center);
	\end{pgfonlayer}
\end{tikzpicture}

%% file: figures/journal-figures/t-box.tikz
\begin{tikzpicture}
	\begin{pgfonlayer}{nodelayer}
		\node [style=none] (0) at (-0.75, 0) {};
		\node [style=none] (1) at (0.75, 0) {};
		\node [style=bdytextbox] (2) at (0, 0) {$t$};
	\end{pgfonlayer}
	\begin{pgfonlayer}{edgelayer}
		\draw [style=edge] (0.center) to (2);
		\draw [style=edge] (2) to (1.center);
	\end{pgfonlayer}
\end{tikzpicture}

%% file: figures/journal-figures/s-box.tikz
\begin{tikzpicture}
	\begin{pgfonlayer}{nodelayer}
		\node [style=none] (0) at (-0.75, 0) {};
		\node [style=none] (1) at (0.75, 0) {};
		\node [style=bdytextbox] (2) at (0, 0) {$s$};
	\end{pgfonlayer}
	\begin{pgfonlayer}{edgelayer}
		\draw [style=edge] (0.center) to (2);
		\draw [style=edge] (2) to (1.center);
	\end{pgfonlayer}
\end{tikzpicture}

%% file: tex/tocl-article/layered-theories.tex
This section contains the main syntactic constructions of the article. We define {\em layered signatures} in Subsection~\ref{subsec:layered-signatures}, their types and terms in Subsection~\ref{subsec:types-terms} and finally give the full general definition of a layered monoidal theory in Subsection~\ref{subsec:equations-theories}. We then define {\em opfibrational}, {\em fibrational} and {\em deflational} theories in Subsection~\ref{subsec:opfib-defl-theories}. We try to mimic as close as possible the structure of the presentation in Section~\ref{sec:monoidal-theories} on monoidal theories, with additions where necessary, such as the 2-terms. We remark that the terms {\em layered theory} and {\em layered monoidal theory} are used interchangeably; ditto the terms {\em layered signature} and {\em layered monoidal signature}.

\subsection{Layered signatures}\label{subsec:layered-signatures}

A layered monoidal signature can be thought of as an indexing of monoidal signatures.

\begin{definition}[Layered signature]\label{def:layered-signature}
A {\em layered signature} is a tuple $\left(\Omega,\mathcal F,\left\{\M_{\omega}\right\}_{\omega\in\Omega}\right)$,
where $\Omega$ is a set, $\mathcal F:\Omega\times\Omega\rightarrow\Set$ is a function, and $\M_{\omega}$ is a monoidal signature for each $\omega\in\Omega$.
\end{definition}
Given a layered signature with monoidal signatures $\M_{\omega}$, we write $\M_{\omega}=(C_{\omega},\Sigma_{\omega})$. We often abbreviate a layered signature to $(\Omega,\mathcal F,\M_{\omega})$, where the index $\omega$ is implicitly assumed to range over $\Omega$. We refer to the elements of $\Omega$ as {\em layers} and to the elements of $\mathcal F(\omega,\tau)$ as {\em generators} with domain layer $\omega$ and codomain layer $\tau$. Semantically, the layers $\Omega$ correspond to monoidal categories presented by $\M_{\omega}$, and the generators in $\mathcal F(\omega,\tau)$ correspond to monoidal functors $\omega\rightarrow\tau$; we refer the reader to~\cite{lmt-part2} for the details.

A {\em morphism of layered signatures} $\left(\Omega,\mathcal F,\M_{\omega}\right)\rightarrow\left(\Psi,\mathcal G,\M_{\psi}\right)$ is given by: (1) a function $F:\Omega\rightarrow\Psi$, (2) for every pair $(\omega,\tau)\in\Omega\times\Omega$, a function $F_{\omega,\tau}:\mathcal F(\omega,\tau)\rightarrow\mathcal G(F(\omega), F(\tau))$, and (3) for every $\omega\in\Omega$, a morphism of monoidal signatures $F^{\omega}:\M_{\omega}\rightarrow\M_{F(\omega)}$.

We denote the resulting category of layered signatures by $\LSgn$. As for the category of monoidal signatures $\MSgn$, we denote a morphism in $\LSgn$ simply by $F:(\Omega,\mathcal F,\M_{\omega})\rightarrow (\Psi,\mathcal G,\M_{\psi})$ as follows: $F:\Omega\rightarrow\Psi$ (without subscripts or superscripts) denotes the function on layers, $F_{\omega,\tau}:\mathcal F(\omega,\tau)\rightarrow\mathcal G(F(\omega),F(\tau))$ the functions on the generators, and $F^{\omega}:\M_{\omega}\rightarrow\M_{F(\omega)}$ the morphisms of monoidal signatures.

As with monoidal signatures, there are two ways of ``forgetting'' the structure of $\LSgn$. First, one can ``forget'' the monoidal signatures, obtaining a functor $\LSgn\rightarrow\mathtt{Fam}(\Set)$ given by $(\Omega,\mathcal F,\M_{\omega})\mapsto (\Omega\times\Omega,\mathcal F)$. Second, one can ``forget'' the generators between the layers, obtaining a functor $\LSgn\rightarrow\mathtt{Fam}(\MSgn)$ simply given by $(\Omega,\mathcal F,\M_{\omega})\mapsto (\Omega,\M_{\omega})$. These functors arise as part of a pullback of categories (cf.~Proposition~\ref{prop:msgn-pullback}).

\noindent\begin{minipage}{0.5\textwidth}
\begin{proposition}\label{prop:lsgn-pullback}
The square~\eqref{eq:lsgn-pullback} is a pullback, where the right vertical map is the family fibration, and the bottom left horizontal map is the forgetful functor from families of monoidal signatures to sets.
\end{proposition}
\end{minipage}
\begin{minipage}{0.5\textwidth}
\begin{equation}\label{eq:lsgn-pullback}
\input{journal-figures/lsgn-pullback.tikz}
\end{equation}
\end{minipage}
\begin{proof}
The proof is very similar to that of Proposition~\ref{prop:msgn-pullback}: choosing a family of sets $(I,\mathcal F)$ and a family of monoidal signatures $(\Omega,\M_{\omega})$ such that the resulting square commutes amounts to the equality $I=\Omega\times\Omega$, and is thus equivalent to choosing a layered signature.
\end{proof}
As for Proposition~\ref{prop:msgn-pullback}, we observe that Proposition~\ref{prop:lsgn-pullback} could be taken as the definition of layered signatures, using the uniqueness of pullbacks.
\begin{corollary}
The forgetful functor $\LSgn\rightarrow\Fam(\MSgn)$ is a fibration.
\end{corollary}
\begin{proof}
As in Corollary~\ref{cor:forgetful-msgn-fibration}, the functor $\LSgn\rightarrow\Fam(\MSgn)$ is a pullback of the family fibration $\Fam(\Set)\rightarrow\Set$, and fibrations are stable under pullbacks.
\end{proof}

When defining a layered signature (especially one that has finitely many layers and generators), it will be useful to simply draw the graph (quiver) with the layers and the generators. We formalise this as the notion of {\em shape}.
\begin{definition}[Shape of a layered signature]\label{def:shape-layered-signature}
Given a layered signature $(\Omega,\mathcal F,\M_{\omega})$, its {\em shape} is the quiver
$$\left(\Omega,\coprod_{\omega,\tau\in\Omega}\mathcal F(\omega,\tau),s,t\right)$$
with the source and the target maps defined by $s(f)\coloneq\omega$ and $t(f)\coloneq\tau$ if $f\in\mathcal F(\omega,\tau)$. Conversely, given a quiver $Q=(V,E,s,t)$, a layered signature with shape $Q$ has the vertices $V$ as the layers, and for the vertices $u,v\in V$ the generators from $u$ to $v$ are given by $s^{-1}(u)\cap t^{-1}(v)$.
\end{definition}

\begin{example}\label{ex:layered-sgn}
A layered signature with the shape $\bullet$ (a quiver with one vertex and no edges) is an ordinary monoidal signature. A layered signature with the shape $\bullet\longrightarrow\bullet$ consists of two monoidal signatures and a single generator between them.
\end{example}

\subsection{Types and terms}\label{subsec:types-terms}

We define types, terms and 2-terms generated by a layered signature. All three levels may come with equations, ultimately giving rise to the notion of a {\em layered theory} (Definition~\ref{def:layered-theory}). We summarise the three levels of generated objects in the following table:
\begin{center}
\renewcommand{\arraystretch}{1.5}
\begin{tabular}{ c | c | c | c }
level & name & meaning & equations \\
\hhline{=|=|=|=}
$\Type_{\mathcal L}$ & types & objects / 0-cells & 0-equations (Definition~\ref{def:0equations}) \\
\hline
$\Term^1_{\mathcal L}$ & terms & morphisms / 1-cells & 1-equations (Definition~\ref{def:1equations}) \\
\hline
$\Term^2_{\mathcal L}$ & 2-terms & 2-morphisms / 2-cells & 2-equations (Definition~\ref{def:2equations})
\end{tabular}
\renewcommand{\arraystretch}{1}
\end{center}

A {\em type} is a list of pairs $A : \omega$, where $\omega$ is a layer and $A$ is a list of colours from $C_{\omega}$. Thinking of this syntax as generating a category of elements (or a Grothendieck construction), $\omega$ corresponds to a (monoidal) category and $A$ is an object in $\omega$; the list then represents an object in the cartesian product of the categories. Syntactically, we make this formal below.
\begin{definition}[Types of a layered signature]\label{def:types-layered-signature}
Given a layered signature $(\Omega,\mathcal F,\M_{\omega})$ with $\M_{\omega}=(C_{\omega},\Sigma_{\omega})$, the set of {\em types} $\Type_{\mathcal L}$ is recursively generated by the rules:
\begin{center}
  \centering\small\noindent
  \begin{bprooftree}
    \AxiomC{\phantom{$\omega\in\Omega$}}
    \RightLabel{\scriptsize\customlabel{type:e-u}{(e-u)}\;}
    \UnaryInfC{$\varepsilon : \varepsilon$}
    \DisplayProof
    \AxiomC{$\omega\in\Omega$}
    \RightLabel{\scriptsize\customlabel{type:i-u}{(i-u)}\;}
    \UnaryInfC{$\varepsilon : \omega$}
    \DisplayProof
    \AxiomC{$\omega\in\Omega$}
    \AxiomC{$a\in C_{\omega}$}
    \RightLabel{\scriptsize\customlabel{type:i-c}{(i-c)}\;}
    \BinaryInfC{$a : \omega$}
  \end{bprooftree}
  \begin{prooftree}
    \AxiomC{$A : \omega$}
    \AxiomC{$f\in\mathcal F(\omega,\tau)$}
    \RightLabel{\scriptsize\customlabel{type:i-g}{(i-g)}\;}
    \BinaryInfC{$f(A) : \tau$}
    \DisplayProof
    \AxiomC{$A : \omega$}
    \AxiomC{$B : \omega$}
    \RightLabel{\scriptsize\customlabel{type:i-m}{(i-m)}\;}
    \BinaryInfC{$AB : \omega$}
    \DisplayProof
    \AxiomC{$T$}
    \AxiomC{$S$}
    \RightLabel{\scriptsize\customlabel{type:e-m}{(e-m)},}
    \BinaryInfC{$T,S$}
  \end{prooftree}
\end{center}
subject to the condition that the rules~\ref{type:i-g} and~\ref{type:i-m} apply only to {\em internal} types, defined as the types generated by the rules whose label begins with i (i.e.~without applying the rules~\ref{type:e-u} and~\ref{type:e-m}). We further require the types formed with the rule~\ref{type:e-m} to define a monoid with the unit $\varepsilon : \varepsilon$, and for each $\omega\in\Omega$, the types formed with the rule~\ref{type:i-m} to define a monoid with the unit $\varepsilon : \omega$ (i.e.~we quotient by associativity and unitality).
\end{definition}
We note that the internal types are of the form $A:\omega$ for some $\omega\in\Omega$ and $A\in C_{\omega}^*$. The subset of internal types is denoted by $\IntType_{\mathcal L}$.

Given a morphism of layered signatures $F:\mathcal L\rightarrow\mathcal K$ (with $\mathcal L = (\Omega,\mathcal F,\mathcal M_{\omega})$), it induces a function $F:\Type_{\mathcal L}\rightarrow\Type_{\mathcal K}$ recursively defined as follows:
\begin{align*}
\varepsilon : \varepsilon &\mapsto \varepsilon : \varepsilon &
\varepsilon : \omega &\mapsto \varepsilon : F(\omega) &
a : \omega &\mapsto F^{\omega}(a) : F(\omega) \\
f(A) : \tau &\mapsto F_{\omega,\tau}(f)(F(A)) : F(\tau) &
AB : \omega &\mapsto F(A)F(B) : F(\omega) &
T,S &\mapsto F(T),F(S).
\end{align*}

Let us define the binary relation $P^0_{\mathcal L}\sse\IntType_{\mathcal L}\times\IntType_{\mathcal L}$ on the internal types as containing those terms which are in the same layer: $(A:\omega,B:\tau)\in P^0_{\mathcal L}$ if and only if $\omega = \tau$. We note that this relation is preserved by any function $F:\Type_{\mathcal L}\rightarrow\Type_{\mathcal K}$ induced by a morphism of layered signatures: if $(T,S)\in P^0_{\mathcal L}$, then $(F(T),F(S))\in P^0_{\mathcal K}$.

\begin{definition}[0-equations]\label{def:0equations}
Given a layered signature $\mathcal L$, a set of {\em 0-equations} is a subset $E^0\sse P^0_{\mathcal L}$.
\end{definition}
Given a set of 0-equations $E^0$, we extend it to the {\em type congruence} $\ \Zeq\ $, i.e.~to the smallest equivalence relation on $\Type_{\mathcal L}$ satisfying $E^0\sse {\Zeq}$ as well as the recursive clauses:
\begin{center}
  \centering\small\noindent
  \begin{bprooftree}
    \AxiomC{$A:\omega\Zeq B:\omega$}
    \AxiomC{$f\in\mathcal F(\omega,\tau)$}
    \RightLabel{\;}
    \BinaryInfC{$f(A):\tau\Zeq f(B):\tau$}
    \DisplayProof
    \AxiomC{$A:\omega\Zeq C:\omega$}
    \AxiomC{$B:\omega\Zeq D:\omega$}
    \RightLabel{\;}
    \BinaryInfC{$AB:\omega\Zeq CD:\omega$}
    \DisplayProof
    \AxiomC{$T\Zeq S$}
    \AxiomC{$U\Zeq K$}
    \RightLabel{.}
    \BinaryInfC{$T,U\Zeq S,K$}
  \end{bprooftree}
\end{center}
We call the pairs of types $\Type_{\mathcal L}\times\Type_{\mathcal L}$ {\em sorts}, and denote a sort by $(T \mid S)$.

The next definition introduces the class of terms that has to be contained in every layered theory. The important notion is that of {\em internal terms}: informally, these are the terms that can be drawn inside a single layer or ``tube'' without partitioning or splitting the area; formally, they are defined by mutual recursion with the basic terms below.
\begin{definition}[Basic terms]\label{def:terms-layered-signature}
Let $(\Omega,\mathcal F,\M_{\omega})$ be a layered signature. The {\em basic terms} are generated by the recursive rules in Figure~\ref{fig:layered-terms}, with the side condition that the rules~\ref{term:int-box} and \ref{term:int-tensor} only apply to {\em internal} terms, defined as follows:
\begin{itemize}
\item the terms generated by the rules~\ref{term:int-unit},~\ref{term:int-id},~\ref{term:int-gen},~\ref{term:int-box} or~\ref{term:int-tensor} are internal,
\item if the terms $x : (T \mid S)$ and $y : (S \mid U)$ are internal, then so is the term $x;y : (T \mid U)$ obtained by~\ref{term:comp}.
\end{itemize}
\end{definition}

\begin{figure}
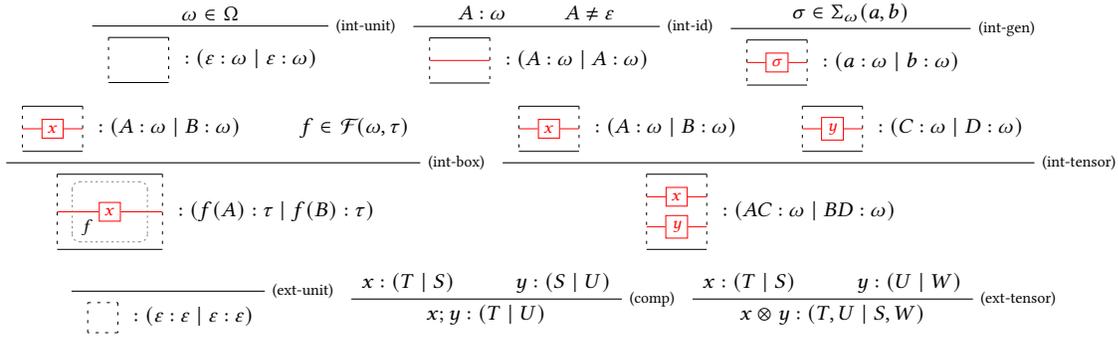

  \centering\small\noindent
  \begin{prooftree}
    \AxiomC{$\omega\in\Omega$}
    \RightLabel{\scriptsize\customlabel{term:int-unit}{(int-unit)}\;\;}
    \UnaryInfC{$\scalebox{.8}{\input{journal-figures/emptydiag-sheet.tikz}} : (\varepsilon : \omega\mid\varepsilon : \omega)$}
    \DisplayProof
    \AxiomC{$A:\omega$}
    \AxiomC{$A\neq\varepsilon$}
    \RightLabel{\scriptsize\customlabel{term:int-id}{(int-id)}\;\;}
    \BinaryInfC{$\scalebox{.8}{\input{journal-figures/iddiag-sheet.tikz}} : (A:\omega \mid A:\omega)$}
    \DisplayProof
    \AxiomC{$\sigma\in\Sigma_{\omega}(a,b)$}
    \RightLabel{\scriptsize\customlabel{term:int-gen}{(int-gen)}}
    \UnaryInfC{$\scalebox{.8}{\input{journal-figures/internalsigmadiag.tikz}} : (a:\omega\mid b:\omega)$}
  \end{prooftree}
  \begin{bprooftree}
    \AxiomC{$\scalebox{.8}{\input{journal-figures/internalxdiag.tikz}} : (A:\omega\mid B:\omega)$}
    \AxiomC{$f\in\mathcal F(\omega,\tau)$}
    \RightLabel{\scriptsize\customlabel{term:int-box}{(int-box)}\;\;}
    \BinaryInfC{$\scalebox{.8}{\input{journal-figures/f-box.tikz}} : (f(A):\tau \mid f(B):\tau)$}
    \DisplayProof
    \AxiomC{$\scalebox{.8}{\input{journal-figures/internalxdiag.tikz}} : (A:\omega\mid B:\omega)$}
    \AxiomC{$\scalebox{.8}{\input{journal-figures/internalydiag.tikz}} : (C:\omega\mid D:\omega)$}
    \RightLabel{\scriptsize\customlabel{term:int-tensor}{(int-tensor)}}
    \BinaryInfC{$\scalebox{.8}{\input{journal-figures/internalxydiag.tikz}} : (AC:\omega \mid BD:\omega)$}
  \end{bprooftree}
  \begin{prooftree}
    \AxiomC{\phantom{$\omega\in\Omega$}}
    \RightLabel{\scriptsize\customlabel{term:ext-unit}{(ext-unit)}\;\;}
    \UnaryInfC{$\scalebox{.8}{\input{journal-figures/emptydiag.tikz}} : (\varepsilon : \varepsilon\mid\varepsilon : \varepsilon)$}
    \DisplayProof
    \AxiomC{$x : (T\mid S)$}
    \AxiomC{$y : (S\mid U)$}
    \RightLabel{\scriptsize\customlabel{term:comp}{(comp)}\;\;}
    \BinaryInfC{$x;y : (T\mid U)$}
    \DisplayProof
    \AxiomC{$x : (T\mid S)$}
    \AxiomC{$y : (U \mid W)$}
    \RightLabel{\scriptsize\customlabel{term:ext-tensor}{(ext-tensor)}}
    \BinaryInfC{$x\otimes y : (T,U \mid S,W)$}
  \end{prooftree}
  \caption{Rules for generating the basic terms of a layered signature.\label{fig:layered-terms}}
\end{figure}

\begin{remark}\label{rem:external-generators-sliding}
In the definition of a layered signature (Definition~\ref{def:layered-signature}), the generators between the layers are all of the sort 1-1, i.e.~the arity and coarity both consist of a single layer -- this is in contrast to the generators within the layers, whose arity and coarity are lists of colours. This restriction is imposed as there is no obvious generalisation of the rule~\ref{term:int-box} to generators with multiple layers in the (co)arity that would make use of the diagrammatic notation. For example, if there was a generator $f\in\mathcal F(\omega\rho,\tau\mu)$, the rule would need to assume two terms $x : (A:\omega\mid B:\omega)$ and $y : (C:\rho\mid D:\rho)$, and output two terms with the sorts $(f(A,C):\tau\mid f(B,D):\tau)$ and $(f(A,C):\mu\mid f(B,D):\mu)$ each of which would need to be indexed with the terms $x$ and $y$; it is not clear how to draw this. Graphically producing such terms is important, since the crucial aspect of reasoning with layered theories amounts to ``sliding'' the internal terms through the diagram: see e.g.~the equations in Figures~\ref{fig:structural-twocells-monoidal} and~\ref{fig:structural-twocells-functors-ext}.
\end{remark}

Despite Remark~\ref{rem:external-generators-sliding}, the layers are allowed to interact via terms with non-singleton (co)arities: such terms are introduced at the level of rules that define terms rather than as generators. We thus define a layered theory (Definition~\ref{def:layered-theory}) relative to further recursive rules extending the basic ones in Figure~\ref{fig:layered-terms}. The additional rules we shall consider are the {\em symmetry terms} (Definition~\ref{def:symmetry-terms}), {\em opfibrational terms} (Figure~\ref{fig:opfibrational-terms}) and {\em fibrational terms} (Figure~\ref{fig:fibrational-terms}). In the next section, we will also add the rule for the {\em cobox} (Definition~\ref{def:cobox}), although the terms generated by this rule can always be eliminated. Any choice of these additional rules is referred to as a {\em recursive sorting procedure}. We remark that any choice of a recursive sorting procedure does not change the set of internal terms.

For the remainder of this subsection, we assume a fixed recursive sorting procedure. Given a layered signature $\mathcal L$, the set of resulting terms is denoted by $\Term^1_{\mathcal L}$, and we write $t:(T\mid S)$ for a term $t$ with the sort $(T\mid S)$. We emphasise that the precise terms in $\Term^1_{\mathcal L}$ depend on the chosen recursive sorting procedure, although it will always contain the basic terms of Definition~\ref{def:terms-layered-signature}. We note that any morphism of layered signatures $F:\mathcal L\rightarrow\mathcal K$ uniquely extends to a function $F:\Term^1_{\mathcal L}\rightarrow\Term^1_{\mathcal K}$ by translating the types and terms in the assumption of each recursive rule, and then applying the corresponding rule in $\Term^1_{\mathcal K}$. The induced function is moreover consistent with the induced function on types: $t:(T\mid S)\mapsto F(t) : (F(T)\mid F(S))$.

We say that two terms are {\em parallel} when they have the same sort, up to 0-equations, if there are any. Formally, given a set of 0-equations $E^0$, we define the binary relation $P^1_{\mathcal L}\sse\Term^1_{\mathcal L}\times\Term^1_{\mathcal L}$ by $\left(t:(T_t\mid S_t),s:(T_s\mid S_s)\right)\in P^1_{\mathcal L}$ if and only if both $T_t\Zeq T_s$ and $S_t\Zeq S_s$, and call the pairs of terms in $P^1_{\mathcal L}$ {\em parallel} with respect to $E^0$. Note that if $E^0$ is empty (i.e.~there are no equations), then terms are parallel if and only if they have the same sort.

\begin{proposition}\label{prop:parallel-terms-preserved}
Let $F:(\mathcal L,E^0_{\mathcal L})\rightarrow (\mathcal K,E^0_{\mathcal K})$ be a morphism of layered signatures such that the induced function on types $F:\Type_{\mathcal L}\rightarrow\Type_{\mathcal K}$ preserves the chosen 0-equations. Then the induced function on terms $F:\Term^1_{\mathcal L}\rightarrow\Term^1_{\mathcal K}$ preserves parallel terms: if $(s,t)\in P^1_{\mathcal L}$, then $(F(s),F(t))\in P^1_{\mathcal K}$.
\end{proposition}
\begin{proof}
By the fact that terms with the same sort are mapped to the terms with the same sort, and by induction on the construction of the type congruence.
\end{proof}

\begin{definition}[1-equations]\label{def:1equations}
Given a layered signature $\mathcal L$ together with a set of 0-equations $E^0$, a set of {\em 1-equations} with respect to $E^0$ is a subset $E^1\sse P^1_{\mathcal L}$.
\end{definition}
Given a set of 1-equations $E^1$, we extend it to the {\em term congruence} $\ \Oeq\ $, i.e.~the smallest equivalence relation on $P^1_{\mathcal L}$ containing $E^1$ that is preserved by the recursive rules in Figure~\ref{fig:layered-terms}.

In many cases, equality between terms is too strong a notion. Often, two 1-cells are not equal but merely related by a 2-cell, or there is a clear direction to a transformation that preserves semantics, as in the graph rewrites for MBQC-graphs (Subsection~\ref{subsec:zx-extraction}). To capture this, we need to extend the layered signature to include generating 2-cells (in addition to 1-equations).
\begin{definition}[Choice of 2-cells]\label{def:choice-2cells}
Given a layered signature $\mathcal L$ together with a set of 0-equations $E^0$, a {\em choice of 2-cells} with respect to $E^0$ is a function $\eta : P^1_{\mathcal L}\rightarrow\Set$, assigning to each parallel pair of terms a set of {\em generating 2-cells}.
\end{definition}
Note that any layered signature in the sense of Definition~\ref{def:layered-signature} can be seen as having a choice of 2-cells by setting $\eta$ to be the constant function returning the empty set. A morphism between layered signatures with a choice of 2-cells $(F,F^2) : (\mathcal L, E^0_{\mathcal L}, \eta_{\mathcal L})\rightarrow (\mathcal K, E^0_{\mathcal K}, \eta_{\mathcal K})$ is given by a morphism of layered signatures $F:\mathcal L\rightarrow\mathcal K$ such that the induced function $F:\Type_{\mathcal L}\rightarrow\Type_{\mathcal K}$ preserves the 0-equations, and for each pair $(t,s)\in P^1_{\mathcal L}$, a function $F^2_{t,s}:\eta_{\mathcal L}(t,s)\rightarrow\eta_{\mathcal K}(F(t),F(s))$.

\begin{definition}[2-terms]\label{def:2terms}
Given a layered signature $\mathcal L$ with a choice of 2-cells $\eta$, a {\em 2-term} is an expression of the form $\alpha : (s,t)$, where $(s,t)\in P^1_{\mathcal L}$ is its {\em sort}, generated by the following recursive procedure:
\begin{center}
  \small
  \begin{bprooftree}
    \AxiomC{$a\in\eta(t,s)$}
    \RightLabel{\;\;}
    \UnaryInfC{$a : (t,s)$}
    \DisplayProof
    \AxiomC{$t\in\Term^1_{\mathcal L}$}
    \RightLabel{\;\;}
    \UnaryInfC{$\id_t : (t,t)$}
    \DisplayProof
    \AxiomC{$\alpha : (t,s)$}
    \AxiomC{$\beta : (s,k)$}
    \RightLabel{\;\;}
    \BinaryInfC{$\alpha;\beta : (t,k)$}
    \DisplayProof
    \AxiomC{$\alpha : (t_1,s_1)$}
    \AxiomC{$\beta : (t_2,s_2)$}
    \RightLabel{\;\;}
    \BinaryInfC{$\alpha\otimes\beta : (t_1\otimes t_2,s_1\otimes s_2)$}
  \end{bprooftree}
  \begin{prooftree}
    \AxiomC{$\alpha : (t_1,s_1)$}
    \AxiomC{$\beta : (t_2,s_2)$}
    \AxiomC{$t_1 : (T\mid S)$}
    \AxiomC{$t_2 : (S\mid K)$}
    \RightLabel{.}
    \QuaternaryInfC{$\alpha *\beta : (t_1;t_2, s_1;s_2)$}
  \end{prooftree}
  \end{center}
Akin to terms, we denote the set of 2-terms by $\Term^2_{\mathcal L}$.
\end{definition}
\begin{remark}
We have chosen not to extend the 2-terms to the recursively defined internal terms (i.e.~the ones generated by the rules~\ref{term:int-box} and~\ref{term:int-tensor}). There is no inherent reason for omitting these terms: upon adding them, one would just need to add the corresponding structural 2-equations (see Definition~\ref{def:str-2eqns}). However, in the models and the examples we will consider the 2-categorical structure only arises at the level of external terms, or does not need to be propagated between layers. Thus, for the sake of slightly reducing the complexity of the current presentation, we omit these terms. Moreover, as we shall observe in Corollary~\ref{cor:defl-twocells-extend}, in deflational theories, such terms are automatically induced by the external ones.
\end{remark}
Let $(F,F^2) : (\mathcal L, E^0_{\mathcal L}, \eta_{\mathcal L})\rightarrow (\mathcal K, E^0_{\mathcal K}, \eta_{\mathcal K})$ be a morphism between layered signatures with a choice of 2-cells. The functions between the generating 2-cells $F^2_{t,s}$ then yield a function $F^2:\Term^2_{\mathcal L}\rightarrow\Term^2_{\mathcal K}$ recursively defined as follows:
\begin{align*}
a : (t,s) &\mapsto F^2_{t,s}(a) : (F(t),F(s)) &
\id_t : (t,t) &\mapsto \id_{F(t)} : (F(t),F(t)) \\
\alpha\otimes\beta : (t_1\otimes t_2,s_1\otimes s_2) &\mapsto F^2(\alpha)\otimes F^2(\beta) : (F(t_1)\otimes F(t_2),F(s_1)\otimes F(s_2)) &
\alpha;\beta : (t,k) &\mapsto F^2(\alpha); F^2(\beta) : (F(t),F(k)) \\
\alpha *\beta : (t_1;t_2, s_1;s_2) &\mapsto F^2(\alpha) * F^2(\beta) : (F(t_1);F(t_2), F(s_1);F(s_2)).
\end{align*}

As for the terms, given a set of 1-equations $E^1$, a pair of 2-terms is in the {\em parallel 2-term relation} $P^2_{\mathcal L}$ with respect to $E^1$ if and only if both 2-terms have the same sort up to the 1-equations: $\left(\alpha:(t_{\alpha},s_{\alpha}), \beta:(t_{\beta},s_{\beta})\right)\in P^2_{\mathcal L}$ if and only if both $t_{\alpha}\Oeq t_{\beta}$ and $s_{\alpha}\Oeq s_{\beta}$. We draw a 2-term $\alpha : (t,s)$ either as \scalebox{.8}{\input{journal-figures/2term.tikz}}, or simply as an arrow $t\xrightarrow{\alpha} s$ (see e.g.~Figure~\ref{fig:structural-twocells-adjoints} on page~\pageref{fig:structural-twocells-adjoints}).

\begin{proposition}\label{prop:parallel-2terms-preserved}
Let $(F,F^2) : (\mathcal L, E^0_{\mathcal L}, \eta_{\mathcal L}, E^1_{\mathcal L})\rightarrow (\mathcal K, E^0_{\mathcal K}, \eta_{\mathcal K}, E^1_{\mathcal K})$
be a morphism between layered signatures with a choice of 2-cells such that the induced function on terms $F:\Term^1_{\mathcal L}\rightarrow\Term^1_{\mathcal K}$ preserves the specified 1-equations. Then the induced function on 2-terms $F^2:\Term^2_{\mathcal L}\rightarrow\Term^2_{\mathcal K}$ preserves parallel terms: if $(\alpha,\beta)\in P^2_{\mathcal L}$, then $(F^2(\alpha),F^2(\beta))\in P^2_{\mathcal K}$.
\end{proposition}
\begin{proof}
By the fact that 2-terms with the same sort are mapped to 2-terms with the same sort, and by induction on the construction of the term congruence.
\end{proof}

\begin{definition}[2-equations]\label{def:2equations}
Let $\mathcal L$ be layered signature with a set of 0-equations $E^0$. Given a set of 1-equations $E^1$ and a choice of 2-cells $\eta$ with respect to $E^0$, a set of {\em 2-equations} with respect to $E^1$ and $\eta$ is a subset $E^2\sse P^2_{\mathcal L}$.
\end{definition}
Given a set of 2-equations $E^2$, we extend it to the {\em 2-term congruence} $\ \Teq\ $, i.e.~the smallest equivalence relation on $P^2_{\mathcal L}$ containing $E^2$ that is preserved by the recursive rules in Definition~\ref{def:2terms}.

\subsection{Theories and structural equations}\label{subsec:equations-theories}

We now have all the ingredients to define a layered theory.

\begin{definition}[Layered theory]\label{def:layered-theory}
Let us fix a recursive sorting procedure. A {\em layered theory} $(\mathcal L,E^0,E^1,\eta,E^2)$ consists of the following:
\begin{itemize}
\item a layered signature $\mathcal L$ (Definition~\ref{def:layered-signature}),
\item a set of 0-equations $E^0$ (Definition~\ref{def:0equations}),
\item a set of 1-equations $E^1$ with respect to $E^0$ (Definition~\ref{def:1equations}),
\item a choice of 2-cells $\eta$ with respect to $E^0$ (Definition~\ref{def:choice-2cells}),
\item a set of 2-equations $E^2$ with respect to $E^1$ and $\eta$ (Definition~\ref{def:2equations}).
\end{itemize}
\end{definition}
For a fixed recursive sorting procedure, a {\em morphism of layered theories}
$$(F,F^2) : (\mathcal L,E^0_{\mathcal L},E^1_{\mathcal L},\eta_{\mathcal L},E^2_{\mathcal L})\rightarrow (\mathcal K,E^0_{\mathcal K},E^1_{\mathcal K},\eta_{\mathcal K},E^2_{\mathcal K})$$
is given by a morphism of layered signatures with a choice of 2-cells $(F,F^2)$ (thus, in particular, the induced function $F:\Type_{\mathcal L}\rightarrow\Type_{\mathcal K}$ preserves the 0-equations), such that the induced functions $F:\Term^1_{\mathcal L}\rightarrow\Term^1_{\mathcal K}$ and $F^2:\Term^2_{\mathcal L}\rightarrow\Term^2_{\mathcal K}$ preserve the 1-equations and the 2-equations, respectively. We denote any category of layered theories by $\LTh$: note that each recursive sorting procedure results in a different category.
\begin{proposition}\label{prop:lth-lsgn-fibration}
Let us fix a recursive sorting procedure. The forgetful functor $\LTh\rightarrow\LSgn$ is a fibration. Moreover, the objects in the fibre $\LTh(\mathcal L)$ are precisely the layered theories with the signature $\mathcal L$.
\end{proposition}
\begin{proof}
The cartesian maps are given by the pairs $(F,F^2)$ such that $F^2_{t,s}=\id$ is the identity map for each pair of parallel terms, and the induced functions on types, terms and 2-terms not only preserve but also reflect the equations.

Let $(\mathcal L,E^0_{\mathcal L},E^1_{\mathcal L},\eta_{\mathcal L},E^2_{\mathcal L})$ be a layered theory and let $F:\mathcal K\rightarrow\mathcal L$ be a morphism of layered signatures. We define the layered theory $(\mathcal K,E^0_{\mathcal K},E^1_{\mathcal K},\eta_{\mathcal K},E^2_{\mathcal K})$ as follows:
\begin{itemize}
\item $(S,T)\in E^0_{\mathcal K}$ if and only if $\left(F(S),F(T)\right)\in E^0_{\mathcal L}$,
\item $(s,t)\in E^1_{\mathcal K}$ if and only if $\left(F(s),F(t)\right)\in E^1_{\mathcal L}$,
\item $\eta_{\mathcal K}(s,t) \coloneq \eta_{\mathcal L}\left(F(s),F(t)\right)$,
\item $(\alpha,\beta)\in E^2_{\mathcal K}$ if and only if $(\alpha,\beta)\in E^2_{\mathcal L}$.
\end{itemize}
Now the pair $(F,\id)$ gives the cartesian lifting of $F$, where $\id$ returns the identity map for each parallel pair of terms.
\end{proof}

Next, we define the structural equations that will hold in nearly all the models we will encounter, with the exception the ``conditional box'' notation defined in Subsection~\ref{subsec:prob-channels}, where we have to drop the functoriality requirement.

\begin{definition}[Structural 0-equations]\label{def:str-0eqns}
Given a layered signature, the {\em structural 0-equations} $S^0$ are given by:
\begin{align*}
(AB)C : \omega &= A(BC) : \omega, &
\varepsilon A : \omega &= A : \omega = A\varepsilon : \omega, &
f(AB):\tau &= f(A)f(B):\tau, &
f(\varepsilon) : \tau &= \varepsilon : \tau.
\end{align*}
\end{definition}

\begin{definition}[Structural 1-equations]\label{def:str-1eqns}
Given a layered signature with the structural 0-equations (Definition~\ref{def:str-0eqns}), the {\em structural 1-equations} $S^1$ are given by the following:
\begin{itemize}
\item the structural identities for monoidal theories (Definition~\ref{def:str-id}), where sequential composition is given by the rule~\ref{term:comp}, the identities are given by the terms~\ref{term:int-id}, parallel composition is given by the rule~\ref{term:ext-tensor}, and the monoidal unit is given by~\ref{term:ext-unit},
\item the structural identities for the monoidal product (the bottom half of Definition~\ref{def:str-id}) hold for the internal terms, where parallel composition is given by the rule~\ref{term:int-tensor}, the monoidal unit is given by~\ref{term:int-unit}, and the sequential composition is given by~\ref{term:comp}\footnote{We note that the structural identities for internal terms which only involve sequential composition and identities are already implied by requiring them for all terms.},
\item the following identity for the identity terms with sort $(AB:\omega\mid AB:\omega)$ generated by the rules~\ref{term:int-id} and~\ref{term:int-tensor}: $\scalebox{.8}{\input{journal-figures/int-id-equals-int-tensor.tikz}} : (AB:\omega\mid AB:\omega)$, where the term on the left-hand side is generated by the rule~\ref{term:int-id}, while the term on the right-hand side is obtained by applying~\ref{term:int-tensor} to the identity terms with sorts $(A:\omega\mid A:\omega)$ and $(B:\omega\mid B:\omega)$, which are, in turn, obtained from~\ref{term:int-id},
\item the identities in Figure~\ref{fig:structural-twocells-functors-int} involving the~\ref{term:int-box} terms, making each $f$ a strict monoidal functor.
\end{itemize}
\end{definition}

\begin{figure}
  \centering\small\noindent
  \scalebox{.8}{\input{journal-figures/structural-twocells-functors-int.tikz}}
  \caption{1-equations defining monoidal functors inside the layers.\label{fig:structural-twocells-functors-int}}
\end{figure}

\begin{definition}[Structural 2-equations]\label{def:str-2eqns}
Given a layered signature and the structural 1-equations of Definition~\ref{def:str-1eqns}, the {\em structural 2-equations} $S^2$ are given in Figure~\ref{fig:str-2eqns} by the expressions on the left, as long as both sides are defined (on the right, we draw the situation in which each equation applies).
\end{definition}

\begin{figure}
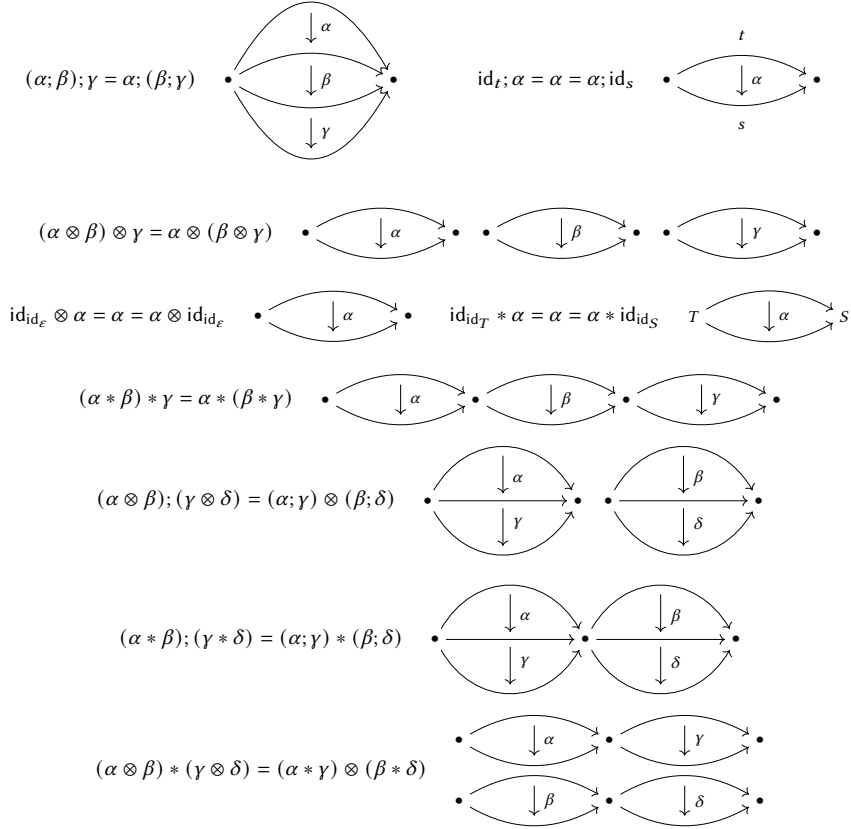

\centering\small
\renewcommand{\arraystretch}{3}
\begin{tabular}{c c}
$(\alpha;\beta);\gamma = \alpha;(\beta;\gamma)$\quad\scalebox{.8}{\input{journal-figures/2cells-compose-assoc.tikz}} & $\id_t;\alpha = \alpha = \alpha;\id_s$\quad\scalebox{.8}{\input{journal-figures/2cells-compose-id.tikz}} \\
\multicolumn{2}{c}{$(\alpha\otimes\beta)\otimes\gamma = \alpha\otimes(\beta\otimes\gamma)$\quad\scalebox{.8}{\input{journal-figures/2cells-tensor-assoc.tikz}}} \\
$\id_{\id_{\varepsilon}}\otimes\alpha = \alpha = \alpha\otimes\id_{\id_{\varepsilon}}$\quad\scalebox{.8}{\input{journal-figures/2cells-tensor-id.tikz}} & $\id_{\id_T} *\alpha = \alpha = \alpha *\id_{\id_S}$\quad\scalebox{.8}{\input{journal-figures/2cells-horizontal-id.tikz}} \\
\multicolumn{2}{c}{$(\alpha *\beta) *\gamma = \alpha *(\beta *\gamma)$\quad\scalebox{.8}{\input{journal-figures/2cells-horizontal-assoc.tikz}}} \\
\multicolumn{2}{c}{$(\alpha\otimes\beta);(\gamma\otimes\delta) = (\alpha;\gamma)\otimes (\beta;\delta)$\quad\scalebox{.8}{\input{journal-figures/2cells-tensor-compose.tikz}}} \\
\multicolumn{2}{c}{$(\alpha *\beta);(\gamma *\delta) = (\alpha;\gamma) * (\beta;\delta)$\quad\scalebox{.8}{\input{journal-figures/2cells-horizontal-compose.tikz}}} \\
\multicolumn{2}{c}{$(\alpha\otimes\beta) * (\gamma\otimes\delta) = (\alpha *\gamma)\otimes (\beta *\delta)$\quad\scalebox{.8}{\input{journal-figures/2cells-tensor-horizontal.tikz}}}
\end{tabular}
\renewcommand{\arraystretch}{1}
\caption{The structural 2-equations. Here $\id_T$ and $\id_S$ are the identity terms on the types $T$ and $S$, while $\id_{\varepsilon}$ is the identity term on $\varepsilon : \varepsilon$ obtained by the rule~\ref{term:ext-unit}.\label{fig:str-2eqns}}
\end{figure}

We say that a layered theory {\em has structural equations} when its sets of equations contain all the structural equations of Definitions~\ref{def:str-0eqns}, \ref{def:str-1eqns} and~\ref{def:str-2eqns}.

We finish this subsection by covering some examples of layered theories.

\noindent\begin{minipage}{0.4\textwidth}
\begin{definition}[Symmetry terms]\label{def:symmetry-terms}
Given a layered signature, the {\em symmetry terms} are generated by the rule on the right:
\end{definition}\end{minipage}
\begin{minipage}{0.6\textwidth}
\begin{center}
  \begin{prooftree}
    \AxiomC{$A:\omega$}
    \AxiomC{$B:\tau$}
    \RightLabel{\customlabel{term:swap}{(swap)}.\;\;}
    \BinaryInfC{$\scalebox{.8}{\input{journal-figures/symdiag-sheet1.tikz}} : (A:\omega, B:\tau \mid B:\tau, A:\omega)$}
  \end{prooftree}
\end{center}
\end{minipage}

\begin{definition}[Externally symmetric layered theory]\label{def:symmetry-1equations}
A layered theory is {\em externally symmetric} if it has the symmetry terms, and the structural 0-equations (Definition~\ref{def:str-0eqns}) and 1-equations (Definition~\ref{def:str-1eqns}) hold, and the 1-equations contain the equations for symmetric monoidal theories in Definition~\ref{def:symm-mon-thy} for the symmetry terms.
\end{definition}

\begin{example}
Consider a layered theory $(\bullet,S^0,E^1,\eset,\eset)$ with no generating 2-cells or 2-equations, where $\bullet$ is a layered signature from Example~\ref{ex:layered-sgn} (one layer and no generators). The internal terms contain the terms of a monoidal signature (Definition~\ref{def:terms-monoidal}). If further $S^1\sse E^1$, then there is a one-to-one correspondence between the equivalence classes of internal terms and the equivalence classes of monoidal terms.
\end{example}

\begin{example}
Consider a layered theory $(\omega\rightarrow\tau,S^0,S^1,\eset,\eset)$ with no generating 2-cells or 2-equations, where $\omega\rightarrow\tau$ is a layered signature from Example~\ref{ex:layered-sgn} (two layers with one generator between them). Then the internal terms with type $\omega$ give the free monoidal category generated by the signature indexed by $\omega$, while the internal terms with type $\tau$ give the free monoidal category generated by the disjoint union of the signature indexed by $\tau$ and the generators obtained by the rule~\ref{term:int-box}. The 1-equations of Figure~\ref{fig:structural-twocells-functors-int} then make $f:\omega\rightarrow\tau$ a strict monoidal functor between the two free models. The functor is free in the sense that all the terms obtained by applying $f$ to the internal terms with type $\omega$ (the \ref{term:int-box} rule) remain uninterpreted, and are simply added as generators to the monoidal category defined by the internal terms with type $\tau$.
\end{example}

\subsection{Opfibrational, fibrational and deflational theories}\label{subsec:opfib-defl-theories}

We finish our abstract discussion of layered theories by defining three classes of layered theories that we will focus on: opfibrational, fibrational and deflational theories. In each case, we specify the recursive sorting procedure as well as the corresponding structural equations. In the case of deflational theories, we also define the structural 2-cells. The three classes of theories give the free models of {\em opfibrations with indexed monoids}, {\em fibrations with indexed comonoids} and {\em monoidal deflations}, respectively -- this is where the names {\em opfibrational} and {\em fibrational} come from. These notions of models are defined and studied in detail in the sequel~\cite{lmt-part2}. Suffice it to say that the models correspond to monoids in the category of opfibrations (with a fixed base), comonoids in the category of fibrations (with a fixed base) and collages of certain diagrams of profunctors. All three kinds of theories ultimately contain the same information -- that of a covariant functor into the category of strict monoidal categories and strict monoidal functors, i.e.~an (op)indexed monoidal category~\cite{hofstra-demarchi2006,shulman2008,ponto-shulman12,monoidal-gro} -- packaged in different ways. One may think of the resulting diagrammatic terms as morphisms of the indexing category ``filled'' or ``inflated'' with the morphisms of the monoidal categories in the image; ``flattening'' or ``deflating'' the diagrams recovers the indexing category -- hence the name {\em deflational} for our main theories of interest.

\begin{figure}
  \centering\small\noindent
  \begin{prooftree}
    \AxiomC{$A:\omega$}
    \AxiomC{$f\in\Omega(\omega,\tau)$}
    \RightLabel{\customlabel{term:ext-gen}{(ext-gen)}\;\;}
    \BinaryInfC{$\scalebox{.8}{\input{journal-figures/refine-sheet.tikz}} : (A:\omega \mid f(A):\tau)$}
    \DisplayProof
    \AxiomC{$\omega\in\Omega$}
    \RightLabel{\customlabel{term:monoid-unit}{(monoid-unit)}\;\;}
    \UnaryInfC{$\scalebox{.8}{\input{journal-figures/cup.tikz}} : (\varepsilon : \varepsilon\mid\varepsilon : \omega)$}
    \DisplayProof
    \AxiomC{$A:\omega$}
    \RightLabel{\customlabel{term:diag-counit}{(diag-counit)}\;\;}
    \UnaryInfC{$\scalebox{.8}{\input{journal-figures/a-cap.tikz}} : (A:\omega \mid \varepsilon : \varepsilon)$}
  \end{prooftree}
  \begin{prooftree}
    \AxiomC{$A:\omega$}
    \AxiomC{$B:\omega$}
    \RightLabel{\customlabel{term:monoid}{(monoid)}\;\;}
    \BinaryInfC{$\scalebox{.8}{\input{journal-figures/pants.tikz}} : (A:\omega, B:\omega \mid AB:\omega)$}
    \DisplayProof
    \AxiomC{$A:\omega$}
    \RightLabel{\customlabel{term:diag}{(diag)}\;\;}
    \UnaryInfC{$\scalebox{.8}{\input{journal-figures/copants-copy.tikz}} : (A:\omega \mid A:\omega, A:\omega)$}
  \end{prooftree}
  \caption{Rules for generating the opfibrational terms.\label{fig:opfibrational-terms}}
\end{figure}

All three classes of theories restrict what kinds of 1-equations are allowed. To capture this restriction, we need the notion of {\em global equations}, which correspond to equalities between functors that do not reference any objects or morphisms.
\begin{definition}[Global terms]\label{def:global-terms}
Let $(\Omega,\mathcal F,\M_{\omega})$ be a layered signature. The {\em global terms} for $(\Omega,\mathcal F)$ are generated as follows:
\begin{center}
  \centering\small\noindent
  \begin{bprooftree}
    \AxiomC{$\omega\in\Omega$}
    \RightLabel{\;}
    \UnaryInfC{$\id_{\omega} : \omega\rightarrow\omega$}
    \DisplayProof
    \AxiomC{$f\in\mathcal F(\omega,\tau)$}
    \RightLabel{\;}
    \UnaryInfC{$f:\omega\rightarrow\tau$}
    \DisplayProof
    \AxiomC{$t:\omega\rightarrow\tau$}
    \AxiomC{$s:\tau\rightarrow\rho$}
    \RightLabel{.}
    \BinaryInfC{$t;s : \omega\rightarrow\rho$}
  \end{bprooftree}
\end{center}
\end{definition}
In other words, global terms are the terms of the free category generated by the quiver induced by $(\Omega,\mathcal F)$ (see Definition~\ref{def:shape-layered-signature}). In the presence of the recursive sorting procedures for basic (Figure~\ref{fig:layered-terms}) and opfibrational (Figure~\ref{fig:opfibrational-terms}) terms, a global term $t:\omega\rightarrow\tau$ together with an internal type $A:\omega$ induce the term $t(A) : (A:\omega \mid t(A)_{cod}:\tau)$ of the layered signature $(\Omega,\mathcal F,\M_{\omega})$ as follows:
\begin{center}
  \centering\small\noindent
  \begin{bprooftree}
    \AxiomC{$\id_{\omega} : \omega\rightarrow\omega$}
    \AxiomC{$A:\omega$}
    \RightLabel{\;}
    \BinaryInfC{$\scalebox{.8}{\input{journal-figures/iddiag-sheet.tikz}} : (A:\omega \mid A:\omega)$}
    \DisplayProof
    \AxiomC{$f:\omega\rightarrow\tau$}
    \AxiomC{$A:\omega$}
    \RightLabel{\;}
    \BinaryInfC{$\scalebox{.8}{\input{journal-figures/refine-sheet.tikz}} : (A:\omega \mid f(A):\tau)$}
    \DisplayProof
    \AxiomC{$t;s : \omega\rightarrow\rho$}
    \AxiomC{$A:\omega$}
    \RightLabel{.}
    \BinaryInfC{$t(A);s\left(t(A)_{cod}\right) : \left(A:\omega \mid s\left(t(A)_{cod}\right)_{cod}:\rho\right)$}
  \end{bprooftree}
\end{center}
Similarly, if we replace the sorting procedure for opfibrational terms with the sorting procedure for fibrational terms (Figure~\ref{fig:fibrational-terms}), each global term $t:\omega\rightarrow\tau$ together with an internal type $A:\omega$ induces the term $t(A) : (t(A)_{dom}:\tau \mid A:\omega)$ -- the defining clauses are the same as above but reflected horizontally (and the filled triangle is replaced with the unfilled one).

A {\em global equation} $t=s$ is simply a pair of parallel global terms $t:\omega\rightarrow\tau$ and $s:\omega\rightarrow\tau$. Any global equation $s=t$ induces the class of 1-equations $s(A)=t(A)$ by ranging over all the internal types $A:\omega$. Given a set of global equations $G$, we denote the set of induced 1-equations by $G^1_{op}$ under the sorting procedure in Figure~\ref{fig:opfibrational-terms}, and by $G^1_{fib}$ under the sorting procedure in Figure~\ref{fig:opfibrational-terms}.

Opfibrational theories are defined with respect to the recursive procedure in Figure~\ref{fig:opfibrational-terms}: they are required to have no generating 2-cells, and come with a restriction on the 1-equation that guarantees that the data of an opfibrational theory corresponds to a diagram of monoidal categories and functors.

\begin{definition}[Opfibrational layered theory]\label{def:opfib-layered-theory}
We say that a layered theory $(\mathcal L,E^0,E^1,\eta,E^2)$ is {\em opfibrational} if
\begin{itemize}
\item its recursive sorting procedure consists of the rules in Figure~\ref{fig:opfibrational-terms} and the symmetry terms of Definition~\ref{def:symmetry-terms} (in addition to the basic terms in Figure~\ref{fig:layered-terms}),
\item it has no generating 2-cells: $\eta$ is the constantly empty function and $E^2=\eset$,
\item there is a set of global equations $G$ such that (1) $G^1_{op}\sse E^1$, and (2) for all ${(t=s)}\in E^1$, either ${(t=s)}\in G^1_{op}$ or both $t$ and $s$ are internal terms.
\end{itemize}
\end{definition}

The intuition behind the opfibrational terms is that we take an ``external'' view of the monoidal categories and monoidal functors between them. In more detail:
\begin{itemize}
\item the terms~\ref{term:monoid} and~\ref{term:monoid-unit} capture the monoidal product and unit in the layer $\omega$: e.g.~sliding the internal terms through the term~\ref{term:monoid} defines the monoidal product on morphisms (see Figure~\ref{fig:structural-twocells-monoidal}),
\item the terms~\ref{term:ext-gen} capture the monoidal functors $f:\omega\rightarrow\tau$: any internal term appearing on the left-hand side of the boundary is in the domain $\omega$, and can be pushed through the boundary into the codomain $\tau$ (see Figure~\ref{fig:structural-twocells-functors-ext}),
\item the terms~\ref{term:diag-counit} and~\ref{term:diag} capture the cartesian monoidal structure of the category of monoidal categories; note that they do not correspond to any deleting or copying {\em inside} a layer, rather, it might be instructive to think of them as {\em branching} or {\em nondeterminism}: \ref{term:diag} represents branching into two potential histories of a system, while \ref{term:diag-counit} corresponds to discarding one of the branches or histories.
\end{itemize}

We denote the resulting category of opfibrational layered theories by $\OpFTh$. In this case the morphisms are simply equation preserving morphisms of signatures rather than pairs of morphisms, as there are no generating 2-cells.

\begin{definition}[Structural opfibrational equations]\label{def:str-opfib-eqns}
By the {\em structural opfibrational equations} we mean the structural equations, the external symmetry equations, as well as the following sets of 1-equations, denoted by $S_{\mathsf{opf}}^1$: (1) the equations of uniform comonoids (Definition~\ref{def:thy-univ-comonoids}) for the terms~\ref{term:diag} and~\ref{term:diag-counit}, (2) the defining identities of monoidal categories in Figure~\ref{fig:structural-twocells-monoidal}, and (3) the defining identities of monoidal functors in Figure~\ref{fig:structural-twocells-functors-ext}.
\end{definition}
Note that the structural opfibrational equations do not themselves form an opfibrational theory -- the 1-equations in Figures~\ref{fig:structural-twocells-monoidal} and~\ref{fig:structural-twocells-functors-ext} are not of the required shape. Rather, they capture the identities that must always be satisfied by any opfibrational theory, whereas Definition~\ref{def:opfib-layered-theory} singles out what data of the theory one is free to choose.

\begin{figure}
  \centering\small\noindent
  \scalebox{.8}{\input{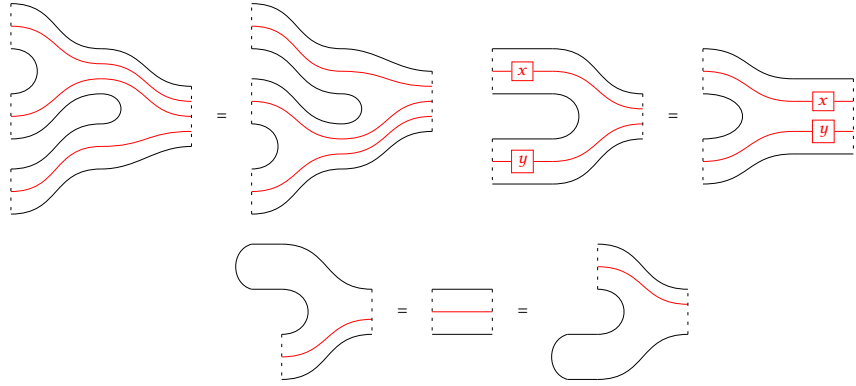}}
  \caption{1-equations defining monoidal categories.\label{fig:structural-twocells-monoidal}}
\end{figure}

\begin{figure}
  \centering\small\noindent
  \scalebox{.8}{\input{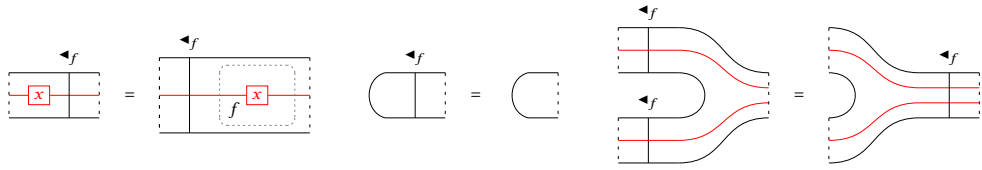}}
  \caption{1-equations defining monoidal functors.\label{fig:structural-twocells-functors-ext}}
\end{figure}

Dualising the construction of opfibrational theories, we obtain {\em fibrational theories}, where the composition of internal and external terms go in opposite directions. Formally, we obtain opfibrational theories by replacing the opfibrational terms (Figure~\ref{fig:opfibrational-terms}) with the {\em fibrational terms} (Figure~\ref{fig:fibrational-terms}) and $G^1_{op}$ with $G^1_{fib}$ in Definition~\ref{def:opfib-layered-theory}. Note that each fibrational term is a horizontal reflection of an opfibrational term, and that in the rule~\ref{term:ext-gen-op} we use an unfilled triangle for emphasis -- the lack of filling plays no formal role.

\begin{figure}
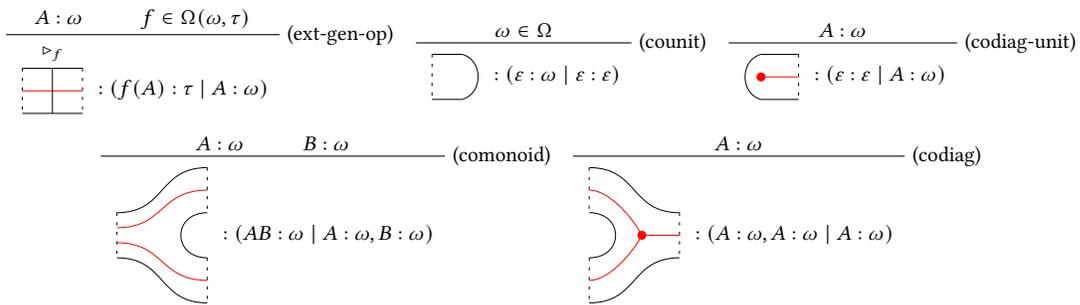

  \centering\small\noindent
  \begin{prooftree}
    \AxiomC{$A:\omega$}
    \AxiomC{$f\in\Omega(\omega,\tau)$}
    \RightLabel{\customlabel{term:ext-gen-op}{(ext-gen-op)}\;\;}
    \BinaryInfC{$\scalebox{.8}{\input{journal-figures/coarsen-sheet.tikz}} : (f(A):\tau \mid A:\omega)$}
    \DisplayProof
    \AxiomC{$\omega\in\Omega$}
    \RightLabel{\customlabel{term:counit}{(counit)}\;\;}
    \UnaryInfC{$\scalebox{.8}{\input{journal-figures/cap.tikz}} : (\varepsilon : \omega \mid \varepsilon : \varepsilon)$}
    \DisplayProof
    \AxiomC{$A:\omega$}
    \RightLabel{\customlabel{term:codiag-unit}{(codiag-unit)}\;\;}
    \UnaryInfC{$\scalebox{.8}{\input{journal-figures/a-cup.tikz}} : (\varepsilon : \varepsilon \mid A:\omega)$}
  \end{prooftree}
  \begin{prooftree}
    \AxiomC{$A:\omega$}
    \AxiomC{$B:\omega$}
    \RightLabel{\customlabel{term:comonoid}{(comonoid)}\;\;}
    \BinaryInfC{$\scalebox{.8}{\input{journal-figures/copants.tikz}} : (AB:\omega \mid A:\omega, B:\omega)$}
    \DisplayProof
    \AxiomC{$A:\omega$}
    \RightLabel{\customlabel{term:codiag}{(codiag)}\;\;}
    \UnaryInfC{$\scalebox{.8}{\input{journal-figures/pants-copy.tikz}} : (A:\omega, A:\omega \mid A:\omega)$}
  \end{prooftree}
  \caption{Rules for generating the fibrational terms.\label{fig:fibrational-terms}}
\end{figure}

The structural fibrational equations are likewise obtained by dualising the structural opfibrational equations (Definition~\ref{def:str-opfib-eqns}). For the sake of completeness, and since they will be used for defining the structural deflational equations, we state the definition explicitly.
\begin{definition}[Structural fibrational equations]\label{def:str-fib-eqns}
By the {\em structural fibrational equations} we mean the structural equations, the external symmetry equations, as well as the following sets of 1-equations, where ``dual'' and ``dualising'' means reflecting the terms horizontally: (1) the equations of uniform monoids (the dual of Definition~\ref{def:thy-univ-comonoids}) for the terms~\ref{term:codiag} and~\ref{term:codiag-unit}, (2) the defining identities of monoidal categories obtained by dualising Figure~\ref{fig:structural-twocells-monoidal}, and (3) the defining identities of monoidal functors obtained by dualising Figure~\ref{fig:structural-twocells-functors-ext}.
\end{definition}
With these modifications, any properties of fibrational theories can be obtained by dualising the corresponding properties of opfibrational theories: all the diagrams are simply written in reverse direction.

We now arrive to our last -- and most important -- class of theories. Deflational theories can be thought of as glueing a pair of opfibrational and fibrational theories together along the internal terms.

\begin{definition}[Deflational layered theory]\label{def:deflational-theory}
A layered theory $(\mathcal L,E^0,E^1,\eta,E^2)$ is {\em deflational} if
\begin{itemize}
\item its recursive sorting procedure consists of the rules for the symmetry terms of Definition~\ref{def:symmetry-terms}, as well as of the rules for opfibrational and fibrational terms in Figures~\ref{fig:opfibrational-terms} and~\ref{fig:fibrational-terms} (in addition to the basic terms in Figure~\ref{fig:layered-terms}),
\item there is a set of global equations $G$ such that (1) $G^1_{op}\sse E^1$, (2) $G^1_{fib}\sse E^1$, and (3) for all ${(t=s)}\in E^1$, exactly one of the following holds: ${(t=s)}\in G^1_{op}$, or ${(t=s)}\in G^1_{fib}$, or both $t$ and $s$ are internal terms.
\end{itemize}
\end{definition}
Note that, unlike for (op)fibrational theories, the deflational theories may contain generating 2-cells and 2-equations. We denote the resulting category of deflational layered theories by $\DeflTh$.

In a deflational theory, every opfibrational term $x:(T\mid S)$ generated by a rule in Figure~\ref{fig:opfibrational-terms} has a corresponding fibrational term $\bar x:(S\mid T)$ generated by the corresponding rule in Figure~\ref{fig:fibrational-terms}. We use this to define the {\em structural 2-cells}, expressing $\bar x$ as the right adjoint to $x$.
\begin{definition}[Structural 2-cells]\label{def:str-twocells}
Given a deflational layered theory with layered signature $\mathcal L$, the {\em structural 2-cells} are given by the choice of 2-cells $\eta_{\mathsf{str}}$ defined as follows: for every opfibrational term $x:(T\mid S)$ generated by a rule in Figure~\ref{fig:opfibrational-terms} one has $\eta_{\mathsf{str}}\left(\id_{T},x;\bar x\right) \coloneq\{\eta_x\}$, $\eta_{\mathsf{str}}\left(\bar x;x,\id_{S}\right) \coloneq\{\varepsilon_x\}$, and $\eta_{\mathsf{str}}$ returns the empty set otherwise.
\end{definition}
We display the 10 structural 2-cells explicitly in Figure~\ref{fig:structural-twocells-adjoints}: the left-hand side contains the unit 2-cells ($\eta_x$), while the right-hand side contains the counit 2-cells ($\varepsilon_x$) drawn as solid arrows going left to right. We consistently omit the labels of the structural 2-cells, as the domain and the codomain determine them uniquely.

\begin{figure}
  \centering\small\noindent
  \scalebox{.75}{\input{journal-figures/structural-twocells-adjoints.tikz}}
  \caption{Structural 2-cells for deflational theories. Left: unit 2-cells. Right: counit 2-cells (drawn as solid arrows); the sections (drawn as dashed gray arrows) may be added without imposing restrictions on the functors.\label{fig:structural-twocells-adjoints}}
\end{figure}

\begin{definition}[Structural deflational equations]\label{def:str-defl-eqns}
By the {\em structural deflational equations} we mean: the structural opfibrational equations (Definition~\ref{def:str-opfib-eqns}), the structural fibrational equations (Definition~\ref{def:str-fib-eqns}) and the following 2-equations for each opfibrational term $x:(T\mid S)$ and all terms $y:(T\mid T)$ and $z:(S\mid S)$ that are either internal or a product (obtained by the~\ref{term:ext-tensor}) of internal terms:
\begin{align*}
\left(\eta_x*\id_x\right);\left(\id_x*\varepsilon_x\right) &=\id_x &
\left(\id_{\bar x}*\eta_x\right);\left(\varepsilon_x*\id_{\bar x}\right) &=\id_{\bar x} &
\id_y*\eta_x &= \eta_x*\id_y &
\id_z*\varepsilon_x &= \varepsilon_x*\id_z.
\end{align*}
We denote the structural deflational {1-} and 2-equations by $S_{\mathsf{defl}}^1$ and $S_{\mathsf{defl}}^2$.
\end{definition}
Note that the structural deflational equations contain, in particular, the structural {0-}, {1-} and 2-equations of Definitions~\ref{def:str-0eqns}, \ref{def:str-1eqns} and~\ref{def:str-2eqns}. The additional 2-equations that are required to hold are the usual zigzag equations defining an adjunction, as well as the equations stating the compatibility of the structural 2-cells with the ``sliding through'' 1-equations for internal terms in Figures~\ref{fig:structural-twocells-monoidal} and~\ref{fig:structural-twocells-functors-ext}.

\begin{definition}[Counit sections and equations]
Given a deflational layered theory with structural 2-cells, its {\em counit sections} are given by the choice of 2-cells $\eta_{\varepsilon s}$ defined as $\eta_{\varepsilon s}\left(\id_{S},\bar x;x\right) \coloneq\{\kappa_x\}$ for every opfibrational term $x:(T\mid S)$ generated by a rule in Figure~\ref{fig:opfibrational-terms}. The corresponding {\em counit section 2-equations} are given by $\kappa_x;\varepsilon_x=\id_x$.
\end{definition}
Given a deflational theory $\mathcal T$ with structural equations and 2-cells, let us write $\mathcal T_{\varepsilon s}$ for the theory obtained from $\mathcal T$ by adding the counit sections and equations. The sections are drawn as dashed gray arrows in Figure~\ref{fig:structural-twocells-adjoints}. While $\mathcal T$ and $\mathcal T_{\varepsilon s}$ are, in general, not equivalent as deflational theories, in~\cite{lmt-part2} we show that they have the same (op)indexed monoidal categories as models: the class of functors is not restricted by introducing the sections. Since the interpretation as a diagram of monoidal categories and functors is our main model of interest, in Sections~\ref{sec:functor-boxes} and~\ref{sec:layered-examples} we will, therefore, make the assumption that the counits in Figure~\ref{fig:structural-twocells-adjoints} have sections $\kappa$, as this will allow us to derive more properties within the formalism of deflational theories.

Opfibrational, fibrational and deflational theories are all externally symmetric (Definition~\ref{def:symmetry-1equations}). The monoidal theories within each layer may additionally happen to be symmetric. When the internal and the external symmetries interact in a coherent way, we call the deflational theory {\em symmetric}.
\begin{definition}[Symmetric deflational theory]\label{def:symm-def-thy}
A deflational theory is {\em symmetric} if it contains the 1-equations for a symmetric monoidal theory (Definition~\ref{def:symm-mon-thy}) for all internal terms for all layers, the generating 2-cells between the following terms: \scalebox{.8}{\input{journal-figures/symmetric-layered-theory.tikz}}, the 2-equations making this an isomorphism, as well as the dual 2-cells obtained by horizontally reflecting the terms.
\end{definition}

%% file: tex/tocl-article/functor-boxes.tex
{\em Functor boxes} (or {\em functorial boxes}) extend the graphical syntax for monoidal categories (i.e.~string diagrams) to monoidal functors. Given a monoidal functor $F:\cat C\rightarrow\cat D$, the idea is to draw a morphism $f:A\rightarrow B$ in $\cat C$ inside a box: \scalebox{.8}{\input{journal-figures/functor-box.tikz}}, which turns it into a morphism of type $FA\rightarrow FB$ in $\cat D$. Thus, one may treat such boxes as terms composable with morphisms in $\cat D$. This idea extends to (co)lax monoidal functors. We refer the reader to the tutorial by Melli\` es~\cite{functorial-boxes} for the details.

In this section, we show that functor boxes arise naturally (and formally) within deflational theories. We will further observe that a box decomposes into the functor boundaries. Within our formalism, a functor box is, therefore, a derived notion rather than a primitive one. This flexibility allows for defining the dual notion, which we dub a {\em cobox} (Definition~\ref{def:cobox}), by composing the refinement and coarsening maps in the opposite order. While a box can be thought of as acting on its interior, a cobox can be thought of as acting on its exterior. Another intuition is to think of a cobox as a {\em window} (Definition~\ref{def:window}), allowing one to peek at the semantics (or at a finer level of granularity) of a part of a diagram.

This section provides the first examples of reasoning {\em within} layered theories. This is in contrast with the previous section as well as the sequel~\cite{lmt-part2}, which focus on proving properties {\em of} layered theories. A feature of such results that might appear strange at first sight is that terms are related by 2-cells that go in both directions, but there are no equalities between the 2-cells. The 2-cells are not required to be isomorphisms or related in any other way. In such a situation, we say that the terms are {\em equal up to bidirectional 2-cells} (note, however, that there are no actual equations involved). We think of this as operational equality: a term of certain shape can always be transformed into another one (and vice versa), or as {\em implications} between the terms. The lack of equations is an artefact of the semantics of deflational theories as {\em monoidal deflations}: imposing additional equations would restrict the class of functors that can be described. An interesting future development would be to search for semantics where the bidirectional 2-cells become isomorphisms without restrictions on the functions. Despite this, there are situations in which we are able to work with strict equality rather than with the bidirectional 2-cells: this is often the case when defining a deflational theory syntactically. An important example of this is given by Proposition~\ref{prop:cobox-iff-faithful}, where one direction of the equivalence also holds for strict equalities: in Subsection~\ref{subsec:elec-circuits}, we use this feature to ``quotient'' the terms of the syntax (electrical circuits) by the equations under translation to the semantics (graphical affine algebra).

\subsection{Functor boxes}\label{subsec:functor-boxes}

In any (op)fibrational or deflational theory, boxes are readily provided by the basic terms generated by the rule~\ref{term:int-box}. In deflational theories, boxes decompose into an internal term sandwiched between a coarsening and a refinement map, so that the term can slide out through either side of the box. We make this observation precise in Proposition~\ref{prop:cowindow-box-equality}. We call a decomposed box a {\em cowindow}.

\begin{definition}[Cowindow]\label{def:cowindow}
In a deflational theory, a 1-cell of the form \scalebox{.8}{\input{journal-figures/cowindow.tikz}} is called a {\em cowindow}, where $x$ is any internal term.
\end{definition}

We now make explicit the connection between cowindows and the usual notation for functor boxes.
\begin{proposition}\label{prop:cowindow-box-equality}
In a deflational theory, we have 2-cells \scalebox{.8}{\input{journal-figures/cowindow-box-equality.tikz}} with $\kappa;\varepsilon=\id$.
\end{proposition}
\begin{proof}
The internal term $x$ slides through $\refine_f$ (or $\coarsen_f$) by the first equation in Figure~\ref{fig:structural-twocells-functors-ext}, after which we remove the fibrational-opfibrational generator pair by applying the counit $\varepsilon$ (Figure~\ref{fig:structural-twocells-adjoints}), which has a section $\kappa$.
\end{proof}
Note that in Proposition~\ref{prop:cowindow-box-equality}, the term on the right-hand side is obtained by a rule for basic terms (Figure~\ref{fig:layered-terms}), while the left-hand side is a composition of three terms in a deflational theory. Thus, functor boxes correspond directly to the terms on the right-hand side, while the proposition demonstrates that we can take one step further and decompose a functor box into three constituent parts. This implies that the standard functor equalities (Figure~\ref{fig:structural-twocells-functors-int}) that had to be imposed in fibrational and opfibrational models are, in fact, provable in deflational theories:
\begin{corollary}
In a deflational theory, the internal identities for monoidal functors (Figure~\ref{fig:structural-twocells-functors-int}) are derivable from the external identities for monoidal functors (Figure~\ref{fig:structural-twocells-functors-ext}), up to bidirectional 2-cells.
\end{corollary}
\begin{proof}
We give the derivation of preservation of monoidal products below; other identities follow similarly:
\begin{equation*}
\scalebox{.8}{\input{journal-figures/ext-to-int-monoidal-product.tikz}}.
\end{equation*}
\end{proof}

\begin{corollary}\label{cor:defl-twocells-extend}
In a deflational theory, any 2-cell $\alpha:x\rightarrow y$ between internal terms $x$ and $y$ extends to all the terms obtained from $x$ and $y$ by applying a rule from Figure~\ref{fig:layered-terms}.
\end{corollary}
\begin{proof}
For the rules~\ref{term:comp} and~\ref{term:ext-tensor}, this is already the case for all terms by the construction of 2-terms (Definition~\ref{def:2terms}).

For the rule~\ref{term:int-box}, a 2-cell between internal terms $\alpha:x\rightarrow y$ with sort $(A:\omega\mid B:\omega)$ extends to
\begin{equation*}
\scalebox{.8}{\input{journal-figures/twocells-extend-internal-terms.tikz}}
\end{equation*}
for any $f\in\mathcal F(\omega,\tau)$. The case of the rule~\ref{term:int-tensor} is similar, using the counit and its section.
\end{proof}

\subsection{Functor coboxes}\label{subsec:functor-coboxes}

In the definition of a cowindow (Definition~\ref{def:cowindow}), there was no inherent reason to start with the backwards boundary $\coarsen$ and end with the forward boundary $\refine$. The reason we chose this order of composition is to obtain the usual functor box: the internal term $x$ is in the domain of $f$, while the whole term is in the codomain. Reversing the order of composition leads to the dual notion, a {\em cobox}, which, to the best knowledge of the authors, has hitherto not been studied abstractly, and has only appeared as a notational shorthand in string diagrammatic electrical circuit theory~\cite{electrical-circuits,boisseau-thesis}, which we make explicit as one of our case studies in Subsection~\ref{subsec:elec-circuits}.

Unlike a box, a cobox cannot be represented as an internal term, that is, there is no analogue of Proposition~\ref{prop:cowindow-box-equality} for coboxes. Nonetheless, in many respects a cobox behaves as if it was an internal term, as we demonstrate in Propositions~\ref{prop:cobox-slides} and~\ref{prop:cobox-composition}. To do that, we need to separate the notion of a window (Definition~\ref{def:window}) from that of a cobox (Definition~\ref{def:cobox}), which allows some wires to bypass the window.

\begin{definition}[Window]\label{def:window}
In a deflational theory, a 1-cell of the form \scalebox{.8}{\input{journal-figures/window.tikz}} is called a {\em window}, where $x$ is any internal term.
\end{definition}

At the level of syntax, coboxes behave, in some ways, like a dual of the box terms obtained by rule~\ref{term:int-box}. We then impose a 1-equation expressing the cobox as an existing term.
\begin{definition}[Cobox]\label{def:cobox}
Given a deflational theory, a {\em cobox} is a term generated by the rule
\begin{center}
  \begin{prooftree}
    \AxiomC{$f\in\mathcal F(\omega,\tau)$}
    \AxiomC{$A,B,C,D:\omega$}
    \AxiomC{$\scalebox{.8}{\input{journal-figures/internalxdiag.tikz}} : (f(A):\tau\mid f(B):\tau)$}
    \RightLabel{\customlabel{term:cobox}{(cobox)}\;\;}
    \TrinaryInfC{$\scalebox{.8}{\input{journal-figures/f-cobox.tikz}} : (CAD:\omega \mid CBD:\omega)$}
  \end{prooftree}
\end{center}
subject to the following 1-equations: \scalebox{.8}{\input{journal-figures/cobox-def.tikz}}.
\end{definition}
Note that the equations say that the terms introduced by the rule~\ref{term:cobox} can always be eliminated. Thus, adding coboxes does not increase the expressive power of a deflational theory. Adding such terms is rather of conceptual value, as it allows us to state equations that would otherwise be difficult to come about or express. The situation should be compared to that of a box, which, up to bidirectional 2-cells, decomposes into two functor boundaries and an internal term (Proposition~\ref{prop:cowindow-box-equality}). Importantly, the cobox reduces the size of the diagrams, making the notation more compact. From now on, we shall assume that all deflational theories come with coboxes and the corresponding 1-equations.

We reiterate that a cobox is {\em not} an internal term. In particular, this means that we cannot apply the rules \ref{term:int-box} and~\ref{term:int-tensor} to the cobox. Some properties of the internal terms and boxes do carry over, as the propositions below demonstrate.
\begin{proposition}\label{prop:cobox-slides}
In a symmetric deflational theory (Definition~\ref{def:symm-def-thy}), we have the following bidirectional 2-cells:
\begin{equation*}
\scalebox{.8}{\input{journal-figures/cobox-naturality.tikz}}.
\end{equation*}
\end{proposition}
\begin{proof}
We give the proof for the left-hand side in Figure~\ref{fig:cobox-naturality-proof}; the argument for the right-hand side is symmetric.
\end{proof}

\begin{figure}
  \centering
  \scalebox{.8}{\input{journal-figures/cobox-naturality-proof.tikz}}
  \caption{Proof of Proposition~\ref{prop:cobox-slides}.\label{fig:cobox-naturality-proof}}
\end{figure}

Likewise, despite the rules \ref{term:int-box} and~\ref{term:int-tensor} not being available, we may still ``tensor'' the cobox with internal terms, as well as ``apply the box to the cobox'' via the following identifications:
\begin{equation*}
\scalebox{.8}{\input{journal-figures/cobox-products-box.tikz}}.
\end{equation*}
With Proposition~\ref{prop:cobox-slides} and the above identifications, the cobox may, in most cases, be treated as an internal term. One should, however, be aware that in reality it is not: this becomes apparent from the equations that hold (or fail to hold) for the cobox. One should also beware that while some identities that hold for boxes also hold for coboxes, in general, the box identities will not hold. We give some examples of this below.
\begin{proposition}\label{prop:cobox-derivable-twocells}
The following 2-cells are derivable in any deflational theory (there is, in general, no reverse 2-cell):
\begin{equation*}
\scalebox{.8}{\input{journal-figures/cobox-derivable-twocells.tikz}}.
\end{equation*}
\end{proposition}
\begin{proof}
We construct the third 2-cell in Figure~\ref{fig:cobox-derivable-twocells-proof} -- the other two cases are similar. We note that in each case the non-reversibility comes from applying units.
\end{proof}

\begin{figure}
  \centering
  \scalebox{.8}{\input{journal-figures/cobox-derivable-twocells-proof.tikz}}
  \caption{Proof of the third case of Proposition~\ref{prop:cobox-derivable-twocells}.\label{fig:cobox-derivable-twocells-proof}}
\end{figure}

\noindent\begin{minipage}{0.5\textwidth}
\begin{proposition}\label{prop:cobox-composition}
The bidirectional 2-cells on the right are derivable in any deflational theory:
\end{proposition}
\end{minipage}
\begin{minipage}{0.5\textwidth}
\begin{equation*}
\scalebox{.8}{\input{journal-figures/cobox-composition.tikz}}.
\end{equation*}
\end{minipage}
\begin{proof}
Since the cobox whose bypassing contextual wires are monoidal units reduces to a single non-branching layer (cf.~derivations in Figures~\ref{fig:cobox-naturality-proof} and~\ref{fig:cobox-derivable-twocells-proof}), this follows by the counit and its section for the $\coarsen\refine$-pair.
\end{proof}
Note that since the cobox is not an internal term, the 2-cells in Proposition~\ref{prop:cobox-composition} do not extend to the general case of sequential composition of coboxes (with the contextual wires present), as indicated by the first case in Proposition~\ref{prop:cobox-derivable-twocells}.

The following proposition establishes the most important property of coboxes that will be used in the next section: it allows detecting when the functor the cobox is representing is faithful (up to bidirectional 2-cells).
\begin{proposition}\label{prop:cobox-iff-faithful}
In a deflational theory, there are bidirectional 2-cells~\eqref{eq:cobox-pres-id} if and only if for all terms $x$ and $y$ with the sort $(A:\omega\mid B:\omega)$, the existence of 2-cells on the left of~\eqref{eq:box-term-faithfulness} implies the existence of 2-cells on the right of~\eqref{eq:box-term-faithfulness}:

\noindent\begin{minipage}{0.35\textwidth}
\begin{equation}\label{eq:cobox-pres-id}
\scalebox{.8}{\input{journal-figures/cobox-pres-id.tikz}}
\end{equation}
\end{minipage}%
\begin{minipage}{0.65\textwidth}
\begin{equation}\label{eq:box-term-faithfulness}
\scalebox{.8}{\input{journal-figures/box-term-faithfulness.tikz}}.
\end{equation}
\end{minipage}

Moreover, if~\eqref{eq:cobox-pres-id} is an equality, then the implication~\eqref{eq:box-term-faithfulness} holds also when the bidirectional 2-cells on both sides are replaced with equalities (the converse does not hold in general).
\end{proposition}
\begin{proof}
Suppose the 2-cells~\eqref{eq:cobox-pres-id} exist, and let $x$ and $y$ be terms such that $fx\rightleftarrows fy$. We then compute as follows:
\begin{equation*}
\scalebox{.8}{\input{journal-figures/box-term-faithfulness-proof-1.tikz}}.
\end{equation*}
It is then clear that the same argument shows the equational case, upon replacing the bidirectional 2-cells with equalities in~\eqref{eq:cobox-pres-id} and in the antecedent of~\eqref{eq:box-term-faithfulness}.

Conversely, suppose that the implication~\eqref{eq:box-term-faithfulness} (as stated with the bidirectional 2-cells) holds for all terms. We observe that \scalebox{.8}{\input{journal-figures/box-term-faithfulness-proof-2.tikz}}, whence existence of 2-cells~\eqref{eq:cobox-pres-id} follows.
\end{proof}

%% file: tex/tocl-article/layered-examples.tex
Here we give extended examples of layered monoidal theories, boxes and coboxes. The examples with digital circuits (Section~\ref{subsec:digital-circuits}), electrical circuits (Section~\ref{subsec:elec-circuits}), ZX-calculus (Section~\ref{subsec:zx-extraction}), and probabilistic channels (Section~\ref{subsec:prob-channels}) all build on monoidal theories from existing literature. We discuss how several levels of monoidal structure are used within existing work, and show how layered monoidal theories explicate and formalise such levels. In the remaining two examples, the calculus of communicating systems (Section~\ref{subsec:ccs}) and chemical reactions (Section~\ref{subsec:glucose}), the monoidal theories are constructed here for the first time in order to allow for a layered perspective.

The case studies are intended to be self-contained and independent of each other. At the beginning of each case study, we give a minimal introduction required to construct the layered theory at hand, and refer the reader to the literature for more details. It is certainly not necessary to read the case studies in order, nor are the details in any one of them required to understand the other examples.

While obtaining novel results in each of the areas is not our aim here -- rather, we aim to identify layered theories where they already appear implicitly -- it is worth highlighting where the formalism of layered theories contributes to the area in question:
\begin{itemize}
\item for digital circuits (Subsection~\ref{subsec:digital-circuits}), we recover the notation for {\em bundlers} from~\cite{kaye-thesis} not as additional generators but as the bidirectional functor boundaries;
\item for electrical circuits (Subsection~\ref{subsec:elec-circuits}), we likewise recover an existing notion: namely, the {\em impedance box} introduced in~\cite{electrical-circuits,boisseau-thesis}, where it is used as syntactic sugar, whereas within the layered theory it is definable as a combination of a box with a cobox (Definition~\ref{def:impedance-box});
\item for probabilistic channels (Subsection~\ref{subsec:prob-channels}), we clarify the usage of `boxes' for conditionals and normalisation (see e.g.~\cite{jacobs-spr,lorenz-tull-causal,jacobs-szeles-stein25}): they resemble the functor boxes in the sense that the corresponding terms are obtained using the rule~\ref{term:int-box}, however, they are not exactly the functor boxes, as the functoriality equations (Figures~\ref{fig:structural-twocells-functors-int} and~\ref{fig:structural-twocells-functors-ext}) do not hold in general.
\end{itemize}
We emphasise that in each case mentioned above, as well as in the other case studies, the terms of a deflational theory provide expressivity beyond the typical situation of a functor interpreting syntax in a semantics. In the case of digital circuits (Subsection~\ref{subsec:digital-circuits}), the mixed terms of Figure~\ref{fig:ALU} cannot be syntactically represented as a functor applied to term, as internal terms from two distinct layers are mixed within a single term. The impedance box (Subsection~\ref{subsec:elec-circuits}) is a paradigmatic example of how a cobox can be used to obtain {\em partial translations} of a term -- without the need to translate the entire diagram. In similar vein, a cobox can be thought of as a {\em hole} in the coarser layer into which a diagram in the finer layer can be plugged. An extreme example of this is given by the term representing a chemical reaction in Figure~\ref{fig:phosphorylation-layers} (Subsection~\ref{subsec:glucose}): despite the coarser layer only having the identity morphisms, the layered theory contains a (highly non-trivial) term between objects in the coarser layer.

\subsection{Digital circuits}\label{subsec:digital-circuits}
\input{./tex/tocl-article/examples/digital-circuits}

\subsection{Electrical circuits}\label{subsec:elec-circuits}
\input{./tex/tocl-article/examples/electrical-circuits}

\subsection{ZX-calculus and quantum circuit extraction}\label{subsec:zx-extraction}
\input{./tex/tocl-article/examples/zx}

\subsection{Calculus of communicating systems}\label{subsec:ccs}
\input{./tex/tocl-article/examples/ccs}

\subsection{Chemical reactions}\label{subsec:glucose}
\input{./tex/tocl-article/examples/glucose}

\subsection{Probabilistic channels}\label{subsec:prob-channels}
\input{./tex/tocl-article/examples/prob-channels}

%% file: tex/tocl-article/examples/digital-circuits.tex
We give a very simple example of an {\em arithmetic logic unit} ({ALU}) that is able to perform two operations on tuples of bits, the choice between which is controlled by one bit. To this end, we define the layered monoidal theory of {\em simple arithmetic circuits} with one layer for each $n\in\N_{+}$: the interpretation of logic and arithmetic gates within a layer is that they operate on $n$-bit wires. One benefit of our approach is that the monoidal signatures within each layer are nearly identical (a total of eight generators with arities and coarities in $\{0,1,2\}$), resulting in a rather compact definition. While the example {ALU} is very simple, the theory we define is expressive enough to represent an {ALU} for any computing circuit with arbitrary bitwise logic and arithmetic operations, such as a central processing unit ({CPU}).

In Kaye~\cite[Example~3.20]{kaye-thesis}, the simple arithmetic circuits are introduced as the {\em generalised circuit signature for simple arithmetic circuits}, which is defined as an ordinary monoidal signature with an infinite number of colours (one for each number of bits $n\in\N_{+}$), and an infinite number of generators. The generators represent the same (finite) operations on bits as we define below in~\eqref{eq:sac-generators}: an infinite number of generators is needed as their arities and coarities keep track of the number of incoming and outgoing wires, as well as the number of bits within each wire. In contrast, the layered theory presented here has a finite number of generators (eight) and colours (just a single one) in each layer, but there are now infinitely many layers (again, one for each number of bits).

Consider the following layered signature: $(\N_{+}, \mathcal F, \Sigma_n)$, where for each $n\in\N_{+}$, the set $\mathcal F(n,1)$ has exactly one element, and $\mathcal F$ returns the empty set otherwise. For $n\neq 1$, the monoidal signature $\Sigma_n$ is defined to have a single colour $n$, and the following monoidal generators:
\begin{equation}\label{eq:sac-generators}
\scalebox{.6}{\input{journal-figures/wireunit.tikz}}, \scalebox{.6}{\input{journal-figures/cofork.tikz}}, \scalebox{.6}{\input{journal-figures/wirecounit.tikz}}, \scalebox{.6}{\input{journal-figures/fork.tikz}}, \scalebox{.6}{\input{journal-figures/and.tikz}}, \scalebox{.6}{\input{journal-figures/or.tikz}}, \scalebox{.6}{\input{journal-figures/plus.tikz}}, \scalebox{.6}{\input{journal-figures/not.tikz}},
\end{equation}
while for $n=1$, the monoidal signature $\Sigma_1$ is as above, except that we replace the adder (generator labelled with `+') with \scalebox{.6}{\input{journal-figures/plus-one.tikz}}. Note that the label $n$ is added for emphasis: there is just one colour within each layer, and hence the generators are the same for all layers, except for when $n=1$. We refer to this layered signature as {\em simple arithmetic circuits}. The intended interpretation is that the wires in layer $n$ carry an $n$-bit signal, while the above generators modify the signal as follows: introduce a wire with no signal, combine the signal from two wires into a single wire, delete a wire, copy the signal, bitwise {AND}, bitwise {OR}, binary sum, bitwise {NOT}.

Next, we define the deflational theory for the simple arithmetic circuits. For the bitwise operators {AND}, {OR} and {NOT}, the usual expected equalities hold in layer $1$. Similarly, in layer $1$, the (co)commutative (co)monoid equations (Definitions~\ref{def:thy-mon} and~\ref{def:thy-comon}) hold for the first four generators. The remaining equations of the layered theory are introduced in Figure~\ref{fig:dig-circ-eqns}, defined by recursion on $n$, so that each $f\in\mathcal F(n,1)$ can be thought of as a functor expressing the logical and arithmetic operations on $n$-bit wires in terms of $1$-bit wires. Here \scalebox{.6}{\input{journal-figures/varbox-circ.tikz}} ranges over \scalebox{.6}{\input{journal-figures/cofork-nolabel.tikz}}, \scalebox{.6}{\input{journal-figures/and-nolabel.tikz}} and \scalebox{.6}{\input{journal-figures/or-nolabel.tikz}}.
\begin{figure}
    \centering
    \scalebox{.6}{%
        \input{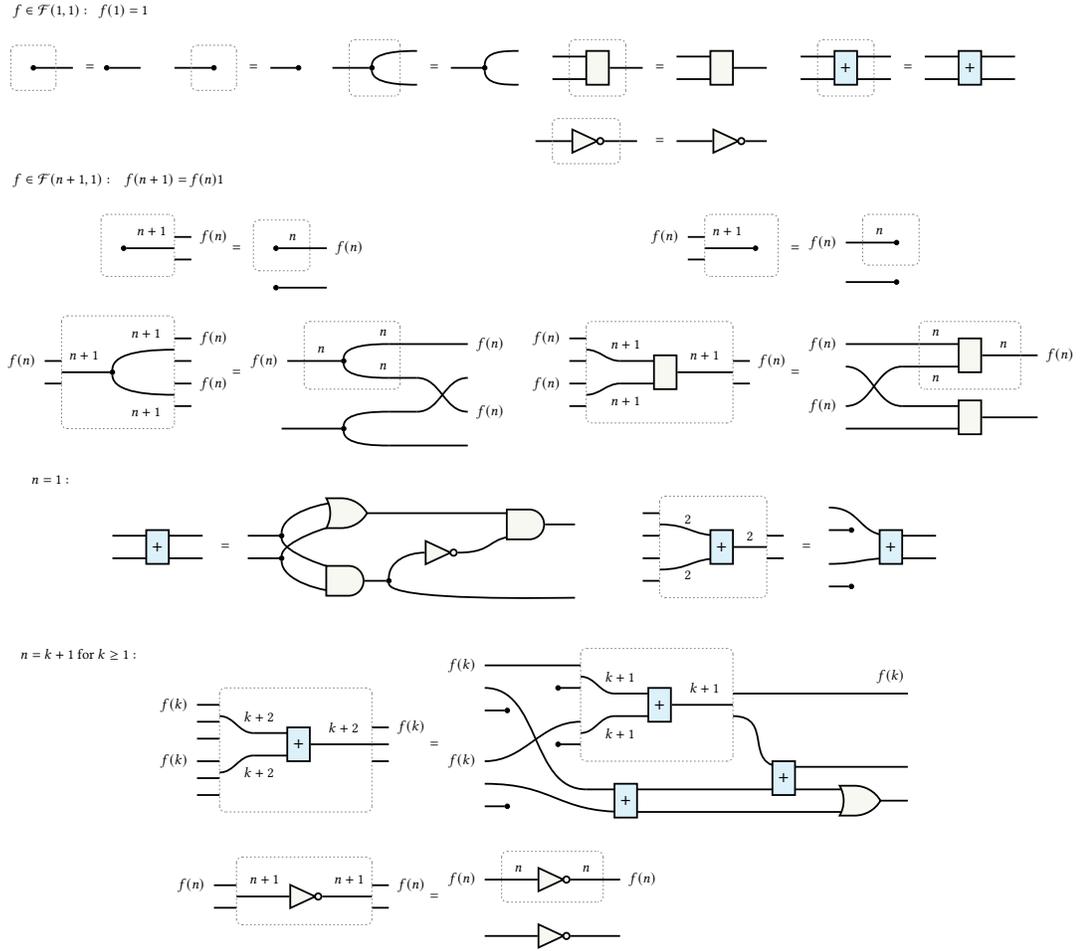}
    }
    \caption{Layered theory defining the simple arithmetic circuits.\label{fig:dig-circ-eqns}}
\end{figure}

We demonstrate the organisational power of layered theories by recasting Example~3.22 from Kaye~\cite{kaye-thesis} as a term in the theory of simple arithmetic circuits: we draw it on the left-hand side of Figure~\ref{fig:ALU}. The term represents a simple {ALU} (arithmetic logic unit) operating on four bit wires with a one bit control wire: when the control bit is false, it performs bitwise {AND}, while when the control bit is true, it adds the numbers represented in binary by the first three bits. We note that the labels $1$ and $4$ refer to the whole enclosed region, corresponding to the layers $n=1$ and $n=4$. The layers are separated by a functor boundary labelled with $\coarsen$, corresponding to what is called a {\em bundler} in Kaye~\cite{kaye-thesis}. The translation from $4$-bit wires to $1$-bit wires can, therefore, be thought of as a bracketing, demarcating the two layers from each other. The formalism, however, allows for more than just separating different layers: using the equations within this layered theory, we observe that the term is equal to the one shown on the right-hand side. Note that while the functor boundary $\coarsen$ ``acts'' from right to left, the overall logical flow of the diagram is from left to right.
\begin{figure}
    \centering
    \scalebox{.6}{%
        \input{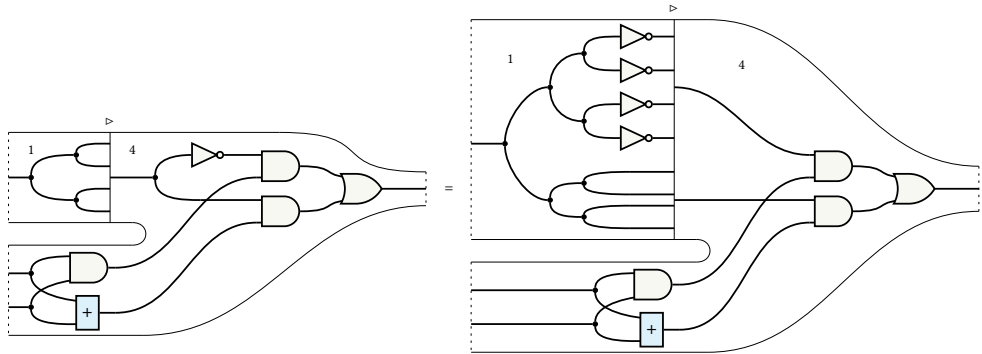}
    }
    \caption{A simple arithmetic logic unit.\label{fig:ALU}}
\end{figure}

%% file: figures/journal-figures/wireunit.tikz
\begin{tikzpicture}
	\begin{pgfonlayer}{nodelayer}
		\node [style=none] (1) at (0.5, 0) {};
		\node [style=none] (4) at (1, 0) {$n$};
		\node [style=vertex] (5) at (-1, 0) {};
	\end{pgfonlayer}
	\begin{pgfonlayer}{edgelayer}
		\draw [style=wire] (1.center) to (5);
	\end{pgfonlayer}
\end{tikzpicture}

%% file: figures/journal-figures/wirecounit.tikz
\begin{tikzpicture}
	\begin{pgfonlayer}{nodelayer}
		\node [style=none] (1) at (-0.5, 0) {};
		\node [style=none] (4) at (-1, 0) {$n$};
		\node [style=vertex] (5) at (1, 0) {};
	\end{pgfonlayer}
	\begin{pgfonlayer}{edgelayer}
		\draw [style=wire] (1.center) to (5);
	\end{pgfonlayer}
\end{tikzpicture}

%% file: tex/tocl-article/examples/electrical-circuits.tex
We formalise the notion of an {\em impedance box}, introduced in the study of compositional electrical circuit theory by Boisseau \& Sobociński~\cite{electrical-circuits} and Boisseau~\cite{boisseau-thesis}. We find that it corresponds to a composition of a box with a cobox in a layered monoidal theory generated by translating electrical circuits to graphical affine algebra (Definition~\ref{def:impedance-box}). Within this deflational theory, we are able to replicate what is called the {\em impedance calculus} in~\cite{electrical-circuits}. To illustrate this, we derive the rule governing the sequential composition of resistors using our formalism. We summarise the layers we define in the table below:
\begin{center}
\begin{tabular}{ c | c }
$\ECirc$ & The electrical circuits \\
\hline
$\Bip$ & Bipole electrical circuits
\end{tabular}\quad%
\begin{tabular}{ c | c }
$\GAA$ & Graphical affine algebra: axiomatises affine relations \\
\hline
$\Imp$ & All terms of $\GAA$ with single input and output
\end{tabular}
\end{center}
Compositional electrical circuit theory~\cite{electrical-circuits,boisseau-thesis,compositional-networks,graphical-affine-algebra} treats the components of electrical circuits as generators in certain monoidal theory. The terms (string diagrams) in the theory bear very close resemblance to the electrical circuit diagrams in classical electrical circuit theory, and hence to the physical circuit. As formal mathematics, the diagrams allow for equational reasoning and a semantic interpretation as affine relations, bridging the gap between the physical wiring of the diagrams and proving their properties by the means of calculations. Our starting point is the work of Boisseau and Sobociński~\cite{electrical-circuits}, which introduced impedance boxes as a notational device for simplifying and clarifying proofs: the idea is to allow parts of the electrical circuit diagram to be translated to the semantics (graphical affine algebra).

\noindent\begin{minipage}{0.65\textwidth}
Let $\R(x)$ denote the field of fractions of the polynomial ring over the real numbers. We define the layered monoidal theory with the shape~\eqref{ecirc-layers}:
\end{minipage}
\begin{minipage}{0.35\textwidth}
\begin{equation}\label{ecirc-layers}
\scalebox{1}{\input{journal-figures/ecirc-layers.tikz}}.
\end{equation}
\end{minipage}
Here $\Bip$ is the layer of {\em bipoles}, $\ECirc$ the layer of {\em electrical circuits}, $\GAA_{\R(x)}$ the layer of {\em graphical affine algebra} over $\R(x)$, and $\Imp$ is the {\em impedance layer}, consisting of all terms that may appear inside an impedance box (Definition~\ref{def:impedance-box}). We drop the subscript $\R(x)$ for brevity, implicitly assuming that all the parameters in $\GAA$ come from the field $\R(x)$.

Each of the four layers has exactly one colour. Following~\cite{electrical-circuits}, the wires and the generators in $\ECirc$ and $\Bip$ are coloured blue, while in $\GAA$ and $\Imp$ they are coloured black, however, this plays no formal role. The generators in $\GAA$ are given by the generators of {\em graphical affine algebra}:
\begin{equation}\label{gaa-generators}
\scalebox{.65}{\input{journal-figures/gaa-generators.tikz}},
\end{equation}
where $k\in\R(x)$. The generators in $\Bip$ are given by the {\em bipoles}~\eqref{ecirc-bipoles} below, where $R,L,C\in\R_{+}$ and $V,I\in\R$. The generators in $\ECirc$ are given by the bipoles~\eqref{ecirc-bipoles}, together with the generators~\eqref{ecirc-non-bipoles} below:

\noindent\begin{minipage}{.6\textwidth}
\begin{equation}\label{ecirc-bipoles}
\scalebox{.65}{\input{journal-figures/ecirc-bipoles.tikz}}
\end{equation}
\end{minipage}%
\begin{minipage}{.4\textwidth}
\begin{equation}\label{ecirc-non-bipoles}
\scalebox{.65}{\input{journal-figures/ecirc-non-bipoles.tikz}}
\end{equation}
\end{minipage}

Finally, the generators in $\Imp$ are given by all the terms in $\GAA$ with sort $(1,1)$.

The equations of the layered monoidal theory are given by the following:
\begin{itemize}
\item in $\GAA$, the equations of the graphical affine algebra hold: for the generators on the first two lines in~\eqref{gaa-generators}, the equations of {\em interacting Hopf algebras}~\cite{survey-signal-flow,zanasi-thesis} hold, together with three additional equations that make the interpretation of the last generator the relation relating the unique element in the zero dimensional vector space to the unit vector in the one dimensional vector space~\cite{graphical-affine-algebra,electrical-circuits},
\item in $\Imp$, the equations of $\GAA$ hold inside the generators, together with the following additional equations, where the identity on the left-hand side of the right equation is the identity term in $\Imp$, {\em not} the identity in $\GAA$:
\begin{equation*}
\scalebox{.65}{\input{journal-figures/imp-equations.tikz}},
\end{equation*}
\item on the colours, the functors $\mathcal I:\ECirc\rightarrow\GAA$ and $W:\Imp\rightarrow\GAA$ are defined by $\bullet\mapsto\bullet\bullet$, so that the resulting map on objects is $n\mapsto 2n$,
\item the functor $B:\Bip\rightarrow\Imp$ is identity on colours,
\item the functor $\Bip\hookrightarrow\ECirc$ is identity on both colours and the generators, so that the resulting map is the inclusion,
\item the functor $\mathcal I:\ECirc\rightarrow\GAA$ is defined by equations in Figure~\ref{fig:ecirc-layered-theory},
\item the functor $B:\Bip\rightarrow\Imp$ is defined by equations in Figure~\ref{fig:ecirc-layered-theory-2},
\item the functor $W:\Imp\rightarrow\GAA$ is defined by the equation \scalebox{.65}{\input{journal-figures/wrapping-equation.tikz}},
\item finally, we add the equation \scalebox{.65}{\input{journal-figures/i-cobox-id-eqn.tikz}}.
\end{itemize}
Recall that by Proposition~\ref{prop:cobox-iff-faithful}, adding the last equation ``quotients'' the internal terms in $\ECirc$ by equality under the functor $\mathcal I$. This equation can thus be thought of as transporting the equalities in the semantics to terms in the syntax.

\begin{figure}
    \centering
    \scalebox{0.65}{
        \input{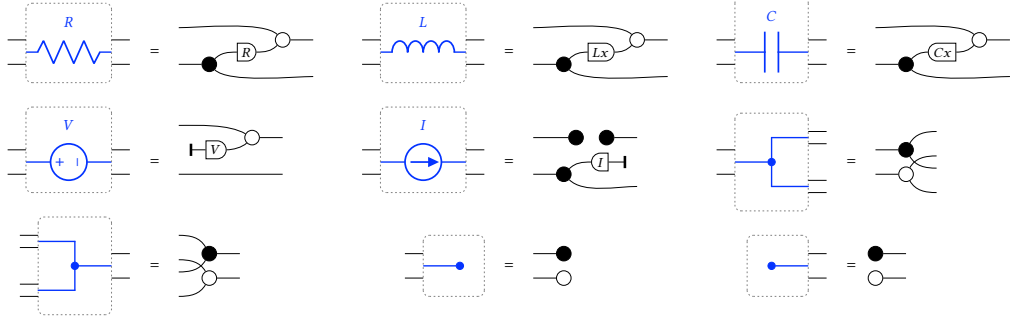}
    }
    \caption{Equations between terms defining the functor $\mathcal I:\ECirc\rightarrow\GAA$.\label{fig:ecirc-layered-theory}}
\end{figure}

\begin{figure}
    \centering
    \scalebox{0.65}{
        \input{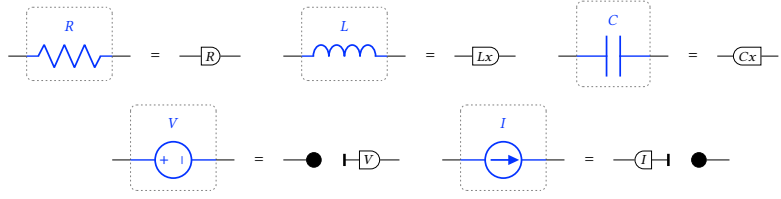}
    }
    \caption{Equations between terms defining the functor $B:\Bip\rightarrow\Imp$.\label{fig:ecirc-layered-theory-2}}
\end{figure}

\begin{remark}
The definition of the layer $\Imp$ might appear somewhat strange: it seems that all the terms disconnect due to the identity term disconnecting. This is, however, not the case, as only the identity generator of $\Imp$ disconnects: the wires inside the terms of $\GAA$ remain connected. Note that the composition is defined by the monoidal product of terms that is then connected on both sides to obtain a term with a single input and output. The composition and the identity are well-defined due to the black and white generators of $\GAA$ (and hence ultimately those of interacting Hopf algebras) being (co)associative and (co)unital. For example, the identity indeed works as expected:
\begin{equation*}
\scalebox{0.65}{\input{journal-figures/imp-identity-example.tikz}}.
\end{equation*}
\end{remark}

Note that the equations of the layered monoidal theories imply that~\eqref{ecirc-layers} is a commutative square. We use this to define the {\em impedance box} of~\cite{electrical-circuits}.
\begin{definition}[Impedance box]\label{def:impedance-box}
Let $C$ be a generator of $\Imp$ (i.e.~a term with exactly one input and output in $\GAA$). Define the {\em impedance box} as the following box-cobox combination: \scalebox{0.65}{\input{journal-figures/impedance-box-definition.tikz}}.
\end{definition}

The impedance box can now be treated like in~\cite{electrical-circuits}, where it is introduced as an additional set of generators for $\ECirc$. However, viewing the impedance box as an emergent feature of the layered monoidal theory reveals some subtleties, as we will see in the following proposition. Note that the scalar $-1$ is abbreviated as the black square in~\eqref{eq:imp-box-iii}.
\begin{proposition}[Lemma~1 in~\cite{electrical-circuits}]\label{prop:impedance-box-derived}
The following equations are derivable:

\noindent\begin{minipage}{.45\textwidth}
\begin{equation}\label{eq:imp-box-i}
\scalebox{.65}{\input{journal-figures/imp-box-compose-seq.tikz}}\tag{i}
\end{equation}
\begin{equation}\label{eq:imp-box-iii}
\scalebox{.65}{\input{journal-figures/imp-box-snake.tikz}}\tag{iii}
\end{equation}
\end{minipage}%
\hspace{0.05\textwidth}%
\begin{minipage}{.45\textwidth}
\begin{equation}\label{eq:imp-box-ii}
\scalebox{.65}{\input{journal-figures/imp-box-compose-par.tikz}}\tag{ii}
\end{equation}
\begin{equation}\label{eq:imp-box-iv}
\scalebox{.65}{\input{journal-figures/imp-box-units.tikz}}\tag{iv},
\end{equation}
\end{minipage}
\end{proposition}
The first difference between Proposition~\ref{prop:impedance-box-derived} and Lemma~1 in~\cite{electrical-circuits} is that we can, in fact, derive equation~\eqref{eq:imp-box-iv} without quotienting $\ECirc$ by the equality under the translation functor to $\GAA$, i.e.~without assuming the last equality of the theory. Up to bidirectional 2-cells, we can even derive equation~\eqref{eq:imp-box-i} without this last equality. The interpretation is that equations~\eqref{eq:imp-box-i}, \eqref{eq:imp-box-ii} and~\eqref{eq:imp-box-iii} require more than existence of a functorial translation $\ECirc\rightarrow\GAA$: they transfer information from the lower level ($\GAA$) to the higher one ($\ECirc$). The second difference is the explicit appearance of a cobox in equation~\eqref{eq:imp-box-iv}, which can be viewed as ensuring that the types on both side of the equality remain in $\ECirc$.

We conclude our discussion of electrical circuits by showing how the impedance box can be used to derive the rule for composing two resistors in $\ECirc$: \scalebox{.65}{\input{journal-figures/resistors-high.tikz}}, replicating part~(i) of Proposition~3 of~\cite{electrical-circuits}. The derivation is given below:
\begin{equation*}
\scalebox{.65}{\input{journal-figures/resistors-add.tikz}}.
\end{equation*}

%% file: tex/tocl-article/examples/zx.tex
We show how using rewriting of graphs representing measurement based quantum computations (MBQC) to extract a quantum circuit can be seen as a procedure inside a layered monoidal theory, whose layers are quantum circuits, graphs representing MBQC-computations and the ZX-calculus. We demonstrate the advantage of the layered approach by clearly separating the layers of graph rewriting -- i.e.~formal manipulations of labelled graphs -- from the semantics in terms of the ZX-calculus. The procedure of circuit extraction can then be seen as an interaction between these layers. We summarise the four layers at play in the following table:
\begin{center}
\begin{tabular}{ c | c }
$\ZX$ & The ZX-calculus: describes all linear maps between qubits \\
\hline
$\QCirc$ & Quantum circuits defined as compositions of gates \\
\hline
$\MBQC$ & Graphs representing a measurement based quantum computation \\
\hline
$\MBQCLC$ & An extension of MBQC-graphs convenient for technical reasons
\end{tabular}
\end{center}

The {\em ZX-calculus} is a prop (monoidal theory with a single colour) with the following generators:
\begin{equation}\label{eq:zx-generators}
\scalebox{1}{\input{journal-figures/zx-generators.tikz}},
\end{equation}
called the {\em Z-spider}, the {\em X-spider} and the {\em Hadamard gate}, subject to the following monoidal theory:
\begin{equation}\label{eq:zx-rules}
\scalebox{1}{\input{journal-figures/ZX-rules.tikz}},
\end{equation}
where $\alpha, \beta \in [0, 2 \pi)$ is called the {\em phase} and the addition is modulo $2\pi$. Note that the ellipsis is to be read as `zero or more wires', hence e.g.~on the left-hand side of the spider fusion equation \SpiderRule the spiders are connected by one or more wires. In addition to these generators and equations, the Hadamard gate is often abbreviated to dashed blue edge:
\begin{equation*}
\scalebox{1}{\input{journal-figures/hadamard-shorthand.tikz}}.
\end{equation*}

The ZX-calculus is able to describe any linear map between any finite number of qubits: we discuss this in more detail in Appendix~\ref{sec:zx-mbqc}. While this makes the ZX-calculus highly expressive, in practice one has to restrict available linear maps to those that can be realistically implemented as physical procedures. There are two dominant paradigms to obtain such an operational restriction of quantum maps: {\em circuit based quantum computation} and {\em measurement based quantum computation} (MBQC for short). In the former, the computation is performed by applying {\em gates}, i.e.~certain subset of unitary linear maps, to qubits, and the resulting circuit diagrams look very similar to electrical and digital circuits as encountered in engineering and computer science literature (as well as in Subsections~\ref{subsec:digital-circuits} and~\ref{subsec:elec-circuits} of this work). In the latter, the computation proceeds by applying single qubit destructive measurements to a state consisting of a finite number of qubits whose every pair may be entangled. Both models of quantum computation can be represented within the ZX-calculus. For circuits, this is straightforward by restricting the ZX-diagrams to those that are circuit-like.
\begin{definition}[Quantum circuits]\label{def:quantum-circ}
The prop $\QCirc$ of {\em quantum circuits} is the sub-prop of the ZX-calculus generated by the following morphisms: \scalebox{1}{\input{journal-figures/cnot.tikz}} ({CNOT}), \scalebox{1}{\input{journal-figures/z-alpha.tikz}} ($\text Z_{\alpha}$), \scalebox{1}{\input{journal-figures/hadamard.tikz}} ({H}).
\end{definition}
Other gates commonly used in circuit based quantum computation are then definable using the generating gates in Definition~\ref{def:quantum-circ}. In particular, we have:
\begin{equation}\label{eq:circuit-generators-2}
\scalebox{1}{\input{journal-figures/circuit-derivable.tikz}}.
\end{equation}

The translation from MBQC protocols (known as {\em measurement patterns}, Definition~\ref{def:measurement-pattern}) to the ZX-calculus is more complicated: we give the details in Appendix~\ref{sec:zx-mbqc}. However, the translation of {\em deterministic} patterns (i.e.~those which result in the same linear map at each run) is given by the composition of functors $D\iota:\MBQC\rightarrow\ZX$ that we define as part of the layered theory presented here.

In MBQC, the computation starts with a {\em graph state}: a finite number of qubits, which may be pairwise entangled by applying the CZ gate. This is represented by an {\em undirected graph}, whose vertices represent qubits (Z-spiders with the zero phase) and whose edges represent CZ-gates (i.e. a Hadamard edge between two qubits). Once we add the information about how the qubits are measured, we obtain the notion of an MBQC-graph (formalised in Definition~\ref{def:mbqc-graph}). It turns out that certain graph-theoretic operations preserve the semantics of measurement patterns.

We denote a graph by $G=(V,E)$, where $V$ is a set of {\em vertices} and $E$ is a binary relation on $V$ that specifies the {\em edges}. Given a vertex $u\in V$, we denote the set of {\em neighbours} of $u$ by $N_G(u)$. The graphs we consider are undirected, in which case there is no distinction between forward and backward neighbours.
\begin{definition}[Local complementation]\label{def:local-comp}
Let $G=(V,E)$ be an undirected irreflexive graph, and let $u\in V$ be a vertex. The {\em local complementation} about $u$ is the graph $G\star u$ defined by $G\star u\coloneq (V, E \Delta \{(a,b) : a,b\in N_G(u) \text{ and } a\neq b\})$, where $\Delta$ denotes the symmetric difference.
\end{definition}
In other words, $G\star u$ has the same vertices as $G$, and any two neighbours $a$ and $b$ of $u$ are connected in $G\star u$ if and only if they are not connected in $G$. All other edges are the same as in $G$.
\begin{definition}[Pivot]\label{def:pivot}
Let $G=(V,E)$ be an undirected irreflexive graph, and let $(u,v)\in E$. The {\em pivot} about $(u,v)$ is defined as $G\wedge uv\coloneq ((G\star u)\star v)\star u$.
\end{definition}
A pivot about $(u,v)$ results in interchanging $u$ and $v$, and complementing the edges between the three sets $N_G(u)\cap N_G(v)$, $(N_G(u)\cap N_G(v)^c)\setminus\{v\}$ and $(N_G(v)\cap N_G(u)^c)\setminus\{u\}$: any vertices $a$ and $b$ from any two distinct sets are connected in $G\wedge uv$ if and only if they are not connected in $G$. We illustrate the effect of a pivot in the following picture, where $A\coloneq N_G(u)\cap N_G(v)$, $B\coloneq (N_G(u)\cap N_G(v)^c)\setminus\{v\}$, $C\coloneq (N_G(v)\cap N_G(u)^c)\setminus\{u\}$, and crossing lines between two sets indicate complementing the edges:
\begin{equation*}
\scalebox{1}{\input{journal-figures/pivot.tikz}}.
\end{equation*}

In order to define an MBQC-graph, we define the following sets and terminology. Let $\mathcal P\coloneq\{\XYplane,\XZplane,\YZplane\}$ denote the set of {\em measurement planes}. Let $\VS$ be a fixed countable set of {\em vertex names}. By an {\em ordered set} we mean a finite list with no repeated elements.
\begin{definition}[MBQC-graph]\label{def:mbqc-graph}
An {\em MBQC-graph} is a tuple $(V,E,I,O,\lambda)$ such that (1) $V\sse\VS$ is a finite set, (2) $(V,E)$ is an undirected irreflexive graph, (3) $I\sse V$ is an ordered set of {\em inputs}, (4) $O\sse V$ is an ordered set of {\em outputs}, and (5) $\lambda : V\setminus O\rightarrow\mathcal P\times [0,2\pi)$ is a {\em measurement labelling function}, assigning a measurement plane and angle to each non-output vertex.
\end{definition}
An MBQC-graph represents the ``desired'' execution of an MBQC-computation, where all the measurements have yielded the expected outcome, so that no corrections were needed. Formally, an MBQC-graph carries the same information as the branch of a measurement pattern where all signals are evaluated to $0$, vanishing all the correction commands (see the discussion in Appendix~\ref{sec:zx-mbqc}). Thus, in this scenario, the measurement angles specified by $\lambda$ translate directly to measurements in the ZX-diagram, as depicted in Table~\ref{tab:measurement-planes}.

 \begin{table}
  \centering
  \renewcommand{\arraystretch}{2}
  \begin{tabular}{ c || c | c | c}
   measurement plane & $\XYplane$ & $\XZplane$ & $\YZplane$ \\ \hline
   ZX-diagram & \input{journal-figures/XY-effect-uncorrected.tikz} & \input{journal-figures/XZ-effect-uncorrected.tikz} & \input{journal-figures/YZ-effect-uncorrected.tikz}
  \end{tabular}
  \renewcommand{\arraystretch}{1}
  \caption{Correspondence between the measurement planes and the ZX-diagrams.\label{tab:measurement-planes}}
 \end{table}

In general, different executions of a measurement pattern will yield different linear maps. In practise, one is interested in those patterns where every execution gives the same outcome linear map (up to a global scalar). Such patterns are called {\em deterministic} (Definition~\ref{def:determinism}). Thus, for a deterministic pattern, every execution is equal to the desired one with no correction commands, i.e.~the one represented by an MBQC-graph. It turns out that there is a condition on MBQC-graphs -- called {\em generalised flow}, or {\em gflow} -- that characterises deterministic measurement patterns~\cite{browne-gflow}. We omit the definition here, as it is somewhat lengthy and requires defining additional notation, and we will not need the details. We refer the reader to~\cite{extended-meas-calculus,thereandback} for the definition.

\begin{theorem}[Browne et al.~\cite{browne-gflow}]\label{thm:gflow}
A measurement pattern with no correction commands can be completed to an equivalent strongly, uniformly and stepwise deterministic pattern (by adding signals and corrections) if and only if the MBQC-graph corresponding to the pattern has a gflow.
\end{theorem}
Thus, instead of working with measurement patterns (Definition~\ref{def:measurement-pattern}) with intricate determinism conditions (Definition~\ref{def:determinism}), this characterisation allows us to work with MBQC-graphs that have a gflow.

For the purposes of rewriting, it is convenient to work with a slight generalisation of MBQC-graphs that allow local Clifford operators at each input and output.

Let us denote by $\LC$ the free monoid generated by $\left\{\frac{r}{2},r,-\frac{r}{2},\frac{g}{2},g,-\frac{g}{2}\right\}$. The monoid $\LC$ captures syntactically the {\em local Clifford operators}, i.e.~X- and Z-spiders with a single input and output whose phase is an integer multiple of $\frac{\pi}{2}$. The translation functor $\LC\rightarrow\ZX$ is defined by the action on the generators in Table~\ref{tab:local-cliffords}.

 \begin{table}
  \centering
  \renewcommand{\arraystretch}{2}
  \begin{tabular}{ c || c | c | c | c | c | c }
   $\LC$ generator & $\frac{r}{2}$ & $r$ & $-\frac{r}{2}$ & $\frac{g}{2}$ & $g$ & $-\frac{g}{2}$ \\ \hline
   ZX-diagram & \input{journal-figures/red-half-pi.tikz} & \input{journal-figures/red-pi.tikz} & \input{journal-figures/minus-red-half-pi.tikz} & \input{journal-figures/green-half-pi.tikz} & \input{journal-figures/green-pi.tikz} & \input{journal-figures/minus-green-half-pi.tikz}
  \end{tabular}
  \renewcommand{\arraystretch}{1}
  \caption{Translation from $\LC$ to ZX.\label{tab:local-cliffords}}
 \end{table}

\begin{definition}[MBQC+LC-graph]\label{def:mbqc-lc-graph}
An {\em MBQC+LC-graph} is a tuple $(V,E,I,O,\lambda,\ell^i,\ell^o)$ where $(V,E,I,O,\lambda)$ is an MBQC-graph, while $\ell^i:I\rightarrow\LC$ and $\ell^o:O\rightarrow\LC$ are {\em input} and {\em output labelling functions}.
\end{definition}
Note that every MBQC-graph can be viewed as an MBQC+LC-graph by choosing all the additional labels (i.e.~values returned by the functions $\ell^i$ and $\ell^o$) to be the empty word. While it seems that we are capturing a larger class of computations by allowing the local Clifford operators at the ends, Proposition~\ref{prop:mbqc-is-mbqc-lc} will show that, in fact, the class of linear maps that is captured by the MBQC+LC-graphs is exactly the same as the one of MBQC-graphs. We say that an MBQC+LC-graph has a gflow if the underlying MBQC-graph has a gflow.

The following definitions combine local complementation (Definition~\ref{def:local-comp}) and pivoting (Definition~\ref{def:pivot}) with the additional data of an MBQC+LC-graph. While the precise details are somewhat technical, the importance lies in the facts that all the operations (1) preserve the semantics of the graph when interpreted as a ZX-diagram (Theorem~\ref{thm:zx-soundness}), and (2) preserve the existence of gflow (Theorem~\ref{thm:pres-gfow}). We point out that the operations add local Clifford operators to inputs and outputs, which is the reason we need to work with MBQC+LC-graphs rather than MBQC-graphs.

\begin{definition}[Local complementation on MBQC+LC-graphs]\label{def:local-comp-mbqc}
Given a vertex name $u\in\VS$, we define the partial function $\star u$ on the set of MBQC+LC-graphs as follows. A graph $G$ is in the domain if $u\in V_G$, and the output graph $G_{\star}$ is defined by the following cases:
\begin{enumerate}
\item if $u\notin I_G$, then, on the vertices and edges we apply the local complementation about $u$: $G_{\star}\coloneq G\star u$, while $I_{\star}\coloneq I_G$, $O_{\star}\coloneq O_G$, and the labelling functions are defined as follows:
\begin{itemize}
\item if $u\in O_G$, let $\ell^o_{\star}(u)\coloneq \frac{r}{2}\cdot\ell^o_G(u)$,
\item if $u\notin O_G$ and $\lambda_G(u)=(\XYplane,\alpha)$, let $\lambda_{\star}(u)\coloneq \left(\XZplane,\frac{\pi}{2}-\alpha\right),$
\item if $u\notin O_G$ and $\lambda_G(u)=(\XZplane,\alpha)$, let $\lambda_{\star}(u)\coloneq \left(\XYplane,\alpha -\frac{\pi}{2}\right),$
\item if $u\notin O_G$ and $\lambda_G(u)=(\YZplane,\alpha)$, let $\lambda_{\star}(u)\coloneq \left(\YZplane,\alpha +\frac{\pi}{2}\right),$
\item if $v\in N_G(u)\cap O_G$, let $\ell^o_{\star}(v)\coloneq -\frac{g}{2}\cdot\ell^o_G(v),$
\item if $v\in N_G(u)\cap O_G^c$ and $\lambda_G(v)=(\XYplane,\alpha)$, let $\lambda_{\star}(v)\coloneq \left(\XYplane, \alpha - \frac{\pi}{2}\right),$
\item if $v\in N_G(u)\cap O_G^c$ and $\lambda_G(v)=(\XZplane,\alpha)$, let $\lambda_{\star}(v)\coloneq \left(\YZplane, \alpha \right),$
\item if $v\in N_G(u)\cap O_G^c$ and $\lambda_G(v)=(\YZplane,\alpha)$, let $\lambda_{\star}(v)\coloneq \left(\XZplane, -\alpha \right),$
\item all other labels are the same as in $G$,
\end{itemize}
\item if $u\in I_G$, we first choose a vertex name $u'\in\VS$ not appearing in $G$, and define the graph $G'$ as follows:
\begin{align*}
V'\coloneq V\cup\{u'\} &&
E'\coloneq E\cup\{(u,u')\} &&
I'\coloneq I_G[u'/u] &&
O'\coloneq O_G &&
\lambda'(u')\coloneq (\XYplane, 0) &&
(\ell^i)'(u')\coloneq \ell^i_G(u)\cdot\frac{g}{2}\cdot\frac{r}{2}\cdot\frac{g}{2},
\end{align*}
while all other labels are the same as in $G$; now, by construction, $u\notin I'$, so that Case~1 applies, whence we define the output graph $G'_{\star}$ using the construction of Case~1.
\end{enumerate}
\end{definition}
Given an MBQC+LC graph $G$ in the domain of $\star u$, we denote the resulting graph by $G\star u$ -- we will point out whether we mean a simple graph or an MBQC+LC-graph, unless this is clear from the context.

\begin{definition}[Pivot on MBQC+LC-graphs]\label{def:pivot-mbqc}
Given two vertex names $u,v\in\VS$, we define the partial function $\wedge uv$ on the set of MBQC+LC-graphs as follows. A graph $G$ is in the domain if $u,v\in V_G$ and $(u,v)\in E_G$, in which case we define $G\wedge uv\coloneq ((G\star u)\star v)\star u$ via three consecutive applications of the partial functions implementing local complementation.
\end{definition}

Certain vertices can also be removed altogether without altering the semantics of the computation. This serves as a starting point for many simplification methods for circuits and MBQC-patterns that use the ZX-calculus.
\begin{definition}[Vertex removal on MBQC+LC-graphs]\label{def:vertex-removal}
Given a vertex name $u\in\VS$, we define the partial function $\setminus u$ as follows. A graph $G$ is in the domain if the following conditions hold: (1) $u\in V_G$ and $u\notin I_G\cup O_G$, (2) $\lambda_G(u)=(P,a\pi)$ with $P\in\{\YZplane,\XZplane\}$ and $a\in\{0,1\}$, in which case we define the output graph $G\setminus u = (V_u,E_u,I_u,O_u,\lambda_u,\ell^i_u,\ell^o_u)$ by letting $V_u\coloneq V_G\setminus\{u\}$, $E_u\coloneq E\setminus\{(u,x) : x\in V_G\}$, $I_u\coloneq I_G$, $O_u\coloneq O_G$, while the labels are defined by the following cases:
\begin{itemize}
\item if $v\in N_G(u)\cap O_G$, let $\ell^o_u(v)\coloneq g^a\cdot\ell^o_G(v)$,
\item if $v\in N_G(u)\cap O_G^c$ and $\lambda_G(v)=(Q,\alpha)$, define $\lambda_u(v)\coloneq (Q, (-1)^a\alpha)$ if $Q\in\{\YZplane,\XZplane\}$, and $\lambda_u(v)\coloneq (\XYplane, \alpha + a\pi)$ if $Q=\XYplane$,
\item if $v\notin N_G(u)$, the labels are the same as in $G$.
\end{itemize}
\end{definition}

Finally, we need to be able to rename vertices.
\begin{definition}[Renaming]\label{def:zx-renaming}
Let $G$ be an MBQC+LC-graph, let $U\sse V_G$ be an ordered subset of vertices of $G$, and let $W\sse\VS$ be an ordered set of vertices such that $|U|=|W|$ and $W\cap V_G=\eset$. We define the MBQC+LC-graph $G[W/U]$ by replacing each vertex in $U$ with a vertex in $W$ in the specified order, keeping the rest of the data as in $G$.
\end{definition}

We now have all the ingredients to define the layered monoidal theory with the following shape:
\begin{equation*}
\scalebox{1}{\input{journal-figures/zx-layers.tikz}}.
\end{equation*}
All the layers are props, i.e.~generated by a single colour, while the generators and the equations within the layers are defined as follows:
\begin{itemize}
\item $\ZX$ is the ZX-calculus, with generators~\eqref{eq:zx-generators} and equations~\eqref{eq:zx-rules},
\item $\QCirc$ is generated by the ZX-diagrams corresponding to quantum circuits in Definition~\ref{def:quantum-circ} and~\eqref{eq:circuit-generators-2}, with the equations of the ZX-calculus,
\item a generator in $\MBQCLC$ with type $n\rightarrow m$ is given by an MBQC+LC-graph $G$ such that $|I_G|=n$ and $|O_G|=m$; the 2-cells and the 2-equations are defined in Figure~\ref{fig:twocells-eqns-mbqclc},
\item the layer $\MBQC$ is defined as $\MBQCLC$, with the graphs restricted to MBQC-graphs.
\end{itemize}

\begin{figure}
    \centering
    \scalebox{1}{%
        \input{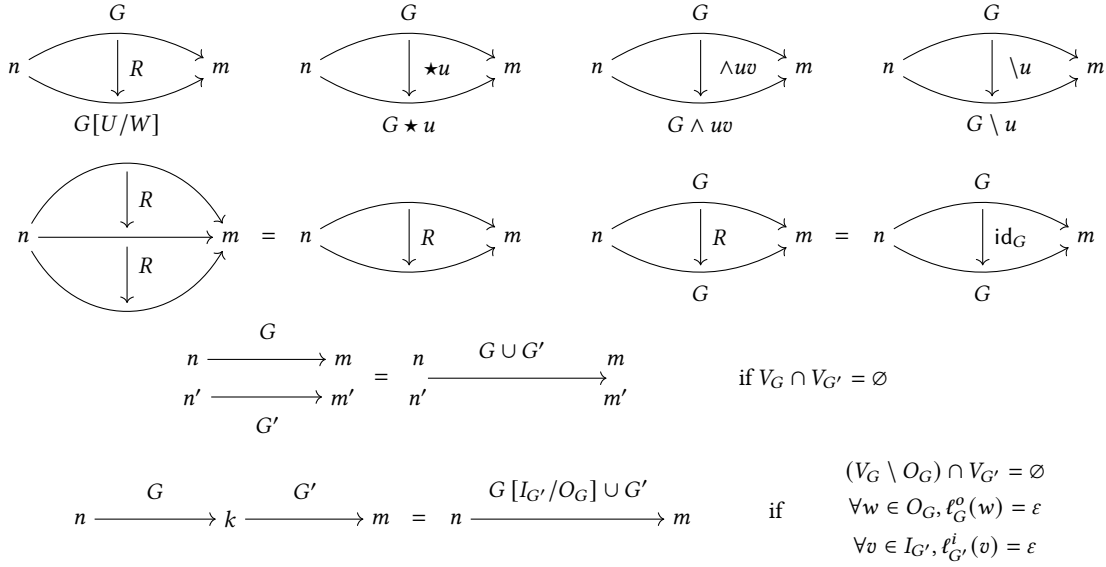}
    }
    \caption{Generating 2-cells and equations in $\MBQCLC$.\label{fig:twocells-eqns-mbqclc}}
\end{figure}

The equations defining the functors make $\QCirc\hookrightarrow\ZX$ and $\iota:\MBQC\hookrightarrow\MBQCLC$ into inclusions, while the functor $D:\MBQCLC\rightarrow\ZX$ is identity on objects and is defined on a generator $G:n\rightarrow m$ as follows:
\begin{enumerate}
\item for each $u\in V_G$, draw a Z-spider with zero phase,
\item for each $(u,v)\in E_G$, draw a Hadamard gate connecting $u$ and $v$,
\item connect all the vertices in $I_G$ to the $n$ input wires in the specified order,
\item connect all the vertices in $O_G$ to the $m$ output wires in the specified order,
\item to each vertex $u\in V_G\setminus O_G$ with $\lambda_G(P,\alpha)$, attach the measurement angle $\alpha$ in the plane $P$, as specified in Table~\ref{tab:measurement-planes},
\item precompose each input $v\in I_G$ with the translation of the monoid element $\ell^i_G(v)$, as specified in Table~\ref{tab:local-cliffords},
\item postcompose each output $w\in O_G$ with the translation of the monoid element $\ell^o_G(w)$, as specified in Table~\ref{tab:local-cliffords}.
\end{enumerate}
Graphically, the action of $D$ on a generator $G:n\rightarrow m$ is depicted by the following equation between terms (note that the inputs $I$ and the outputs $O$ need not be disjoint):
\begin{equation*}
\scalebox{.8}{\input{journal-figures/mbqclc-to-zx.tikz}}.
\end{equation*}

The layered point of view allows for a rather succinct statement of the following result.
\begin{proposition}[Lemma~4.1 in~\cite{thereandback}]\label{prop:mbqc-is-mbqc-lc}
We have $D(\MBQCLC) = D(\iota(\MBQC))$. In other words, every ZX-diagram that is a translation of an MBQC+LC-graph is equal to a ZX-diagram that is a translation of an MBQC-graph.
\end{proposition}

The main observations about the translation $D:\MBQCLC\rightarrow\ZX$ are that all the 2-cells preserve the semantics of the translation as well as the existence of gflow, which we record in the following two theorems. They can be seen as high-level summaries of the results presented in Sections~3.1, 4.2 and~4.3 of~\cite{thereandback}.

\begin{theorem}[Soundness]\label{thm:zx-soundness}
Let $G,G':n\rightarrow m$ be two terms in $\MBQCLC$ such that there is a 2-cell $\eta:G\rightarrow G'$. Then $D(G)=D(G')$ in $\ZX$.
\end{theorem}

\begin{theorem}[Preservation of gflow]\label{thm:pres-gfow}
Let $G,G':n\rightarrow m$ be two terms in $\MBQCLC$ such that there is a 2-cell $\eta:G\rightarrow G'$. Then, if $G$ has a gflow, then so does $G'$.
\end{theorem}

Compositions of the generating 2-cells already allow for some simplification of MBQC+LC-graphs (and hence of the corresponding measurement patterns). For example, by combining the vertex removal 2-cells with the local complementation and pivoting 2-cells, one can show that any non-input Clifford vertices\footnote{Vertices whose measurement angle is an integer multiple of $\frac{\pi}{2}$.} can be removed (\cite[Theorem~4.12]{thereandback}). One could, of course, add more 2-cells. For example, we could add 2-cells that transform local Clifford labels (given by the functions $\ell^i$ and $\ell^o$) into measurements, hence transforming an MBQC+LC-graph into an MBQC-graph one step at a time. This, in turn, could be used to prove Proposition~\ref{prop:mbqc-is-mbqc-lc} within the layered theory.

Theorems~\ref{thm:zx-soundness} and~\ref{thm:pres-gfow} give a general pattern on how to add new 2-cells to $\MBQCLC$. Namely, one needs to choose the 2-cells in such a way that the two theorems remain valid. In other words, one has to check that the graph is changed in such a way that both preserves the semantics under the functor $D$ and the existence of gflow.

Likewise, we could add the remaining 2-cells needed for {\em circuit extraction} -- the process of transforming any MBQC-graph with a gflow into a semantically equivalent circuit. If an MBQC-graph has gflow, there is an efficient algorithm for rewriting it step-by-step into an equivalent circuit form. This shows that the image of graphs (MBQC or MBQC+LC) with a gflow is, in fact, contained in $\QCirc$. In fact, the whole algorithm can be presented within a layered theory. For example, the following equation appears as Step~1 of the extraction procedure (\cite[p.~42]{thereandback}):
\begin{equation*}
\scalebox{1}{\input{journal-figures/example-unfuse-gates.tikz}}.
\end{equation*}
We conclude this case study by noting that the informal boundary between the ``unextracted'' and ``extracted'' parts of the diagram now becomes a formal term within the layered theory:
\begin{equation*}
\scalebox{1}{\input{journal-figures/circuit-extraction-in-layers.tikz}},
\end{equation*}
and each step of the algorithm consists of applying a 2-cell on the left-hand side ($\MBQCLC$) and moving a part of the graph to the right-hand side ($\QCirc$).

%% file: figures/journal-figures/z-alpha.tikz
\begin{tikzpicture}
	\begin{pgfonlayer}{nodelayer}
		\node [style=Z phase dot] (0) at (0, 0) {$\alpha$};
		\node [style=none] (1) at (-1, 0) {};
		\node [style=none] (2) at (1, 0) {};
	\end{pgfonlayer}
	\begin{pgfonlayer}{edgelayer}
		\draw (1.center) to (0);
		\draw (0) to (2.center);
	\end{pgfonlayer}
\end{tikzpicture}

%% file: figures/journal-figures/hadamard.tikz
\begin{tikzpicture}
	\begin{pgfonlayer}{nodelayer}
		\node [style=none] (0) at (-1, 0) {};
		\node [style=none] (1) at (1, 0) {};
		\node [style=hadamard] (2) at (0, 0) {};
	\end{pgfonlayer}
	\begin{pgfonlayer}{edgelayer}
		\draw (0.center) to (2);
		\draw (2) to (1.center);
	\end{pgfonlayer}
\end{tikzpicture}

%% file: tex/tocl-article/examples/ccs.tex
We construct a layered monoidal theory capturing two kinds of semantics for a fragment of the calculus of communicating systems (CCS)~\cite{milner} -- {\em reduction semantics} and {\em labelled transition system} (LTS) semantics -- as well as a functorial translation from a subcategory of the former to the latter. Intuitively, the LTS semantics may be seen as a lower level implementation of the concurrent processes described more coarsely by the reduction semantics. We use the layered monoidal theory to show soundness and completeness of the reduction semantics with respect to a fragment of the LTS semantics (Corollary~\ref{cor:ccs-soundness-completeness}). All the results we obtain here are standard (cf.~Corollaries~\ref{cor:lts-simulation} and~\ref{cor:ccs-soundness-completeness}); however, the treatment via monoidal theories and string diagrams is novel. The proofs of the statements here can be found in Appendix~\ref{sec:ccs-proofs}.

The calculus of communicating systems~\cite{milner} is widely used to reason about concurrent programs and processes. One of its main features is {\em synchronisation} of two processes, modelled as a {\em reduction} or a {\em handshake}. Here we consider a restricted version of CCS with only action prefixing and parallel composition, and two ways to give semantics to the CCS expressions: the reduction semantics (Definition~\ref{def:reduction-semantics}) operates at a higher level, while the LTS semantics (Definition~\ref{def:lts-semantics}) is somewhat more fine-grained. We loosely follow the diagrams from the talk by Krivine~\cite{krivine-talk}.

Let us fix a set $A$, whose elements are called the {\em action names}. We define the sets $\bar A\coloneqq\{\bar a : a\in A\}$ and $\Act\coloneqq A\cup\bar A\cup\{\tau\}$, the latter called the set of {\em actions}, and $\tau$ called the {\em silent action}. The set of {\em processes} is defined recursively as follows, where $x$ ranges over $\Act$:\label{ref-ccs-actions-processes}
\begin{center}
\begin{tabular}{ c c | c | c}
  $P\Coloneqq$ & $0$ & $x.P$ & $P\parallel P$.
\end{tabular}
\end{center}

\begin{definition}[Congruence]\label{def:ccs-congruence}
The {\em congruence} is the equivalence relation $\sim$ on the set of processes generated by:
\begin{align*}
P\parallel Q &\sim Q\parallel P, & (P\parallel Q)\parallel R &\sim P\parallel (Q\parallel R), & 0\parallel P &\sim P, & \text{if } P\sim P' \text{ and } Q\sim Q', &\text{ then } P\parallel Q \sim P'\parallel Q'.
\end{align*}
\end{definition}

\begin{definition}[Reduction semantics]\label{def:reduction-semantics}
A {\em rewrite rule} in {\em reduction semantics} is an ordered pair of processes, which we write as $P\rightarrow Q$, generated by the following three production rules, where $a\in A$:
\begin{center}
  \begin{prooftree}
    \AxiomC{\phantom{$P\rightarrow Q$}}
    \RightLabel{\customlabel{red-sem:tau}{\texttt s}}
    \UnaryInfC{$\tau.P\rightarrow P$}
    \DisplayProof
    \AxiomC{$P\rightarrow Q$}
    \RightLabel{\customlabel{red-sem:p}{\texttt p}}
    \UnaryInfC{$P\parallel R\rightarrow Q\parallel R$}
    \DisplayProof
    \AxiomC{\phantom{$P\rightarrow Q$}}
    \RightLabel{\customlabel{red-sem:r}{\texttt r}}
    \UnaryInfC{$a.P\parallel\bar a.Q\rightarrow P\parallel Q$}
    \DisplayProof
    \AxiomC{$P\rightarrow Q$}
    \AxiomC{$P\sim P'$}
    \AxiomC{$Q\sim Q'$}
    \RightLabel{\customlabel{red-sem:c}{\texttt c}}
    \TrinaryInfC{$P'\rightarrow Q'$}
  \end{prooftree}
\end{center}
\end{definition}
In other words, the rewrite rules are parallel compositions~\ref{red-sem:p} of either the silent action~\ref{red-sem:tau} or the {\em reduction}~\ref{red-sem:r}, up to the congruence~\ref{red-sem:c}. For instance, we can derive the following rewrite rule, where $a\in A$ and $x\in\Act$:
\begin{equation}\label{eqn:CCS-derivation}
a.0\parallel (x.0\parallel\bar a.0)\rightarrow 0\parallel (x.0\parallel 0).
\end{equation}

\begin{definition}[LTS semantics]\label{def:lts-semantics}
A {\em labelled transition} is a triple $(P,x,Q)$, written as $P\xrightarrow x Q$, where $P$ and $Q$ are processes and $x\in\Act$, generated by the production rules below, where $a\in A$:
\begin{center}
  \begin{prooftree}
    \AxiomC{\phantom{$P\xrightarrow x Q$}}
    \RightLabel{\customlabel{lts-sem:a}{\texttt{a}}}
    \UnaryInfC{$x.P\xrightarrow x P$}
    \DisplayProof
    \AxiomC{$P\xrightarrow x Q$}
    \RightLabel{\customlabel{lts-sem:pl}{\texttt{pl}}}
    \UnaryInfC{$R\parallel P\xrightarrow x R\parallel Q$}
    \DisplayProof
    \AxiomC{$P\xrightarrow x Q$}
    \RightLabel{\customlabel{lts-sem:pr}{\texttt{pr}}}
    \UnaryInfC{$P\parallel R\xrightarrow x Q\parallel R$}
    \DisplayProof
    \AxiomC{$P\xrightarrow a P'$}
    \AxiomC{$Q\xrightarrow{\bar a} Q'$}
    \RightLabel{\customlabel{lts-sem:r1}{\texttt{r1}}}
    \BinaryInfC{$P\parallel Q\xrightarrow{\tau} P'\parallel Q'$}
    \DisplayProof
    \AxiomC{$P\xrightarrow{\bar a} P'$}
    \AxiomC{$Q\xrightarrow a Q'$}
    \RightLabel{\customlabel{lts-sem:r2}{\texttt{r2}}}
    \BinaryInfC{$P\parallel Q\xrightarrow{\tau} P'\parallel Q'$}
  \end{prooftree}
\end{center}
\end{definition}

We are now ready to define the layered monoidal theory for the two semantics of CCS defined above, which we refer to as CCS. The definition of the colours of the theory refers to the sets of {\em actions} and {\em processes} defined on page~\pageref{ref-ccs-actions-processes}, while congruence (Definition~\ref{def:ccs-congruence}), reduction semantics (Definition~\ref{def:reduction-semantics}) and LTS semantics (Definition~\ref{def:lts-semantics}) become derived notions. In the subsequent discussion, we prove the precise connection between all three and the layered theory CCS.

CCS has three layers: $\Red$ capturing reduction semantics, $\LTS$ capturing LTS semantics, and $\Comp(\Red)$ capturing those reduction semantics rules where some {\em computation} occurs, that is, the rules without the congruence (i.e.~only the production rules~\ref{red-sem:tau}, \ref{red-sem:r} and~\ref{red-sem:p}). CCS has the following shape:
\begin{equation*}
\scalebox{1}{\input{journal-figures/ccs-layers.tikz}}.
\end{equation*}

The colours of the monoidal signature $\Red$ are given by the set containing three symbols for every process: $P$, $P\uparrow$ and $P\quest$. We refer to the plain process symbols as the {\em active fragment}, to the process symbols with $\uparrow$ as the {\em silent fragment}, and to the processes with $\quest$  as the {\em subsilent fragment}. The intuition behind the fragments is as follows:
\begin{itemize}
\item for the processes $P$ in the active fragment, no computation has yet occurred -- the only transformations that keep a process inside the active fragment correspond to the congruence,
\item for the processes $P\uparrow$ in the silent fragment, a computation (silent action, reduction or parallel composition with a process in the silent fragment) has occurred, and no further reductions or silent actions are possible
\item the processes $P\quest$ in the subsilent fragment are needed in order to capture the congruence in the silent fragment -- each process in the subsilent fragment needs to be eventually merged with a process in the silent fragment in order to obtain a valid rewrite rule (see Proposition~\ref{prop:red-string-corresp}).
\end{itemize}

\noindent\begin{minipage}{0.41\textwidth}
The generators of $\Red$ are given by: (1) the generators~\eqref{generators-red}, where $a\in A$, capturing the production rules~\ref{red-sem:tau}, \ref{red-sem:r} and~\ref{red-sem:p} of Definition~\ref{def:reduction-semantics}, and (2) the {\em structural generators}, consisting of the symmetry and the generators in Figure~\ref{fig:ccs-structural-generators}.
\vspace{3pt}
\end{minipage}
\begin{minipage}{0.59\textwidth}
\vspace{-10pt}
\begin{equation}\label{generators-red}
\scalebox{.8}{\input{journal-figures/generators-CSS.tikz}}
\end{equation}
\begin{equation}\label{generator-ccs-splitting}
\scalebox{.8}{\input{journal-figures/generator-ccs-splitting.tikz}}
\end{equation}
\end{minipage}

\begin{figure}
    \centering
    \scalebox{0.9}{
        \input{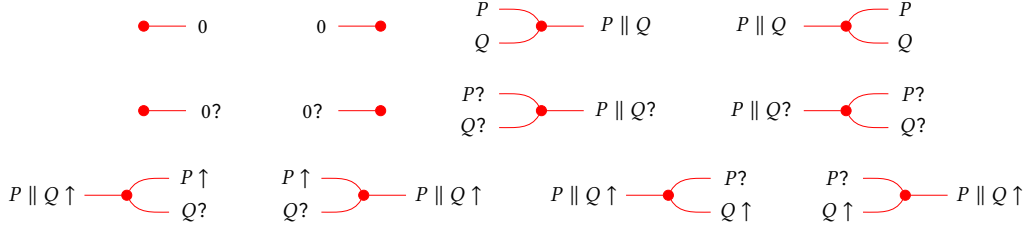}
    }
    \caption{The structural generators of $\Red$.\label{fig:ccs-structural-generators}}
\end{figure}

The monoidal signature $\Comp(\Red)$ is given by restricting the signature $\Red$ to the active and the silent fragments on colours, and the only generators are given by~\eqref{generators-red} and the ``splitting'' structural generator~\eqref{generator-ccs-splitting}. Note that, in particular, $\Comp(\Red)$ does not have the symmetry generators.

The monoidal signature $\LTS$ has as colours the symbols $P$ and $P\uparrow x$ for every process $P$ and an action $x\in\Act$. Similarly to $\Red$, we refer to the plain process symbols as the {\em active fragment}. The process symbols with $\uparrow x$ are referred to as the {\em pending fragment}: one may think of $P\uparrow x$ as a process with a ``pending'' action $x$. In this case, we call the {\em silent fragment} the colours of the form $P\uparrow\tau$, i.e.~a pair of a process and the silent action. The generators of $\LTS$ consist of the ``splitting'' structural generator~\eqref{generator-ccs-splitting}, together with the following generators capturing the production rules in Definition~\ref{def:lts-semantics}, where $a\in A$ and $x\in\Act$:
\begin{equation}\label{generators-LTS}
\scalebox{.8}{\input{journal-figures/generators-LTS.tikz}}.
\end{equation}

The equations of the layered theory for CCS are given in Figure~\ref{fig:ccs-layered-theory}, in addition to which the usual identities of a symmetric monoidal theory (Definition~\ref{def:symm-mon-thy}) hold in $\Red$. The reason for calling the generators in Figure~\ref{fig:ccs-structural-generators} {\em structural} is that they capture the congruence by casting it as (non-strict) symmetric monoidal structure. The first set of equations in Figure~\ref{fig:ccs-layered-theory} says that the structural generators in $\Red$ are isomorphisms. One then obtains the congruence relation of Definition~\ref{def:ccs-congruence} by considering isomorphisms which have exactly one input and one output wire such that both are in the same fragment (active, silent or subsilent). We make this precise in the following lemma.

\begin{figure}
    \centering
    \scalebox{0.9}{
        \input{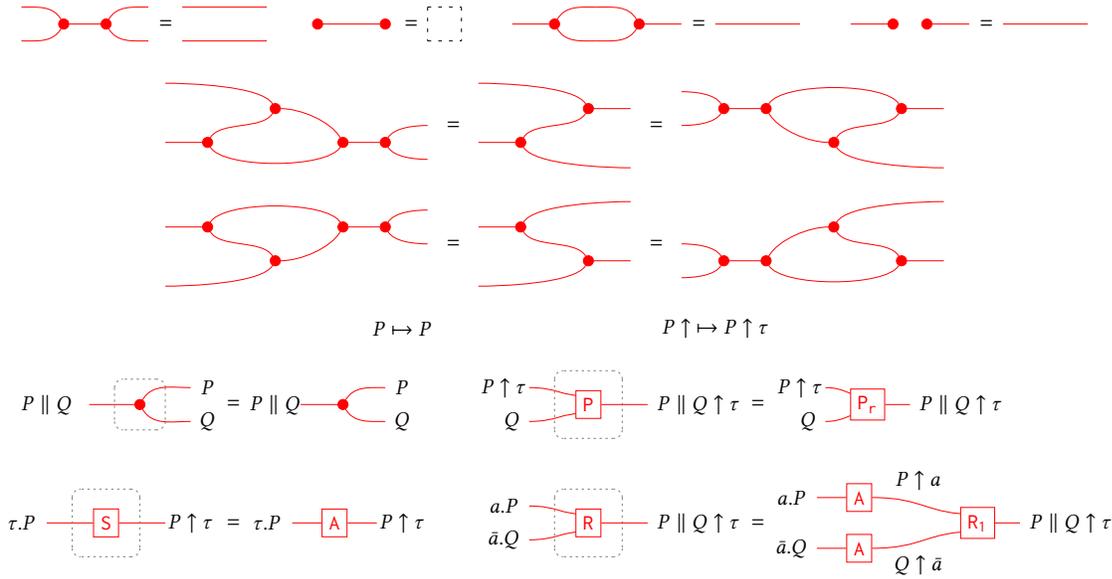}
    }
    \caption{Layered theory for CCS. The equations in the first three rows hold for any structural generators in $\Red$ for which the composition is defined. The remaining equations define the functor $\Comp(\Red)\rightarrow\LTS$, where the arrows $\mapsto$ denote the 0-equations. The functor $\Comp(\Red)\hookrightarrow\Red$ is defined as the inclusion, whose defining equations we have, therefore, omitted from the figure.\label{fig:ccs-layered-theory}}
\end{figure}

\begin{lemma}\label{lma:iso-congruence}
Let $P$ and $Q$ be processes. The following are equivalent, where all the isomorphisms are in $\Red$:
\begin{align*}
\text{(1) } P &\sim Q, & \text{(2) } P &\xrightarrow{\sim}Q, & \text{(3) } {P\quest} &\xrightarrow{\sim}Q\quest, & \text{(4) } {P\uparrow} &\xrightarrow{\sim}Q\uparrow.
\end{align*}
\end{lemma}

Certain terms in $\Red$ and $\LTS$ correspond to the derivable rewrite rules and transitions in reduction semantics (Definition~\ref{def:reduction-semantics}) and LTS semantics (Definition~\ref{def:lts-semantics}). We begin with the correspondence with reduction semantics.
\begin{proposition}\label{prop:red-string-corresp}
Let $P$ and $Q$ be processes. There is a term in $\Red(P,{Q\uparrow})$ if and only if the rewrite rule $P\rightarrow Q$ is derivable using the production rules of reduction semantics in Definition~\ref{def:reduction-semantics}.
\end{proposition}

For LTS semantics, the connection between derivable labelled transitions and string diagrammatic terms requires a further restriction on the latter, as in this case there are no isomorphisms. For example, one can construct the term~\eqref{lts-non-derivable-term}, but the corresponding labelled transition $x.P\parallel (Q\parallel R)\xrightarrow x (P\parallel Q)\parallel R$ is not derivable in LTS semantics. This is because the string diagrammatic representation imports some associativity equations into the terms. The correct term corresponding to the labelled transition $x.P\parallel (Q\parallel R)\xrightarrow x P\parallel (Q\parallel R)$ would be~\eqref{lts-correct-derivable-term}.

\noindent\begin{minipage}{0.52\textwidth}
\begin{equation}\label{lts-non-derivable-term}
\scalebox{.8}{\input{journal-figures/lts-non-derivable-term.tikz}}
\end{equation}
\end{minipage}
\begin{minipage}{0.48\textwidth}
\begin{equation}\label{lts-correct-derivable-term}
\scalebox{.8}{\input{journal-figures/lts-correct-derivable-term.tikz}}
\end{equation}
\end{minipage}

We, therefore, need to restrict to the terms of the following form.
\begin{definition}[Standard terms]
The set of {\em standard} terms in $\LTS$ is recursively generated as follows:
\noindent\paragraph{Base case:} for every process $P$, the following terms are standard: \scalebox{.8}{\input{journal-figures/ccs-normal-form-base-case.tikz}},
\noindent\paragraph{Recursive case:} if the terms $\mathtt t$ and $\mathtt s$ are standard, then so is \scalebox{.8}{\input{journal-figures/ccs-normal-form-recursive-case.tikz}}, where $\mathtt G\in\Gen\coloneq\{\Pgenl,\Pgenr,\Rgenone,\Rgentwo\}$, whenever the composition is defined.
\end{definition}
As expected, the term~\eqref{lts-correct-derivable-term} is standard, whereas the term~\eqref{lts-non-derivable-term} is not.

\begin{proposition}\label{prop:lts-string-corresp}
Let $P$ and $Q$ be processes, and let $x\in\Act$ be an action. There is a standard term in $\LTS(P,{Q\uparrow x})$ if and only if the labelled transition $P\xrightarrow x Q$ is derivable using the production rules of the LTS semantics in Definition~\ref{def:lts-semantics}.
\end{proposition}
\begin{proof}
By straightforward induction on labelled transitions / standard terms, after observing that there is a one-to-one correspondence between the production rules in Definition~\ref{def:lts-semantics} and the generators~\ref{generators-LTS}.
\end{proof}

A term in $\Comp(\Red)$ is {\em standard} if its translation to $\LTS$ is standard. The following proposition shows the connection between the layer $\Red$ and the standard terms in $\Comp(\Red)$.
\begin{lemma}\label{lma:red-term-to-comp}
Let $\mathtt t\in\Red(P,{Q\uparrow})$ be a term. Then there is a term in the layered theory CCS
\begin{equation}\label{red-iso-standard}
\scalebox{.8}{\input{journal-figures/red-iso-standard.tikz}},
\end{equation}
such that the term $t'\in\Comp(\Red)(P',{Q'\uparrow})$ is standard, and $\mathtt i$ and $\mathtt j$ are isomorphisms.
\end{lemma}

The following proposition shows that the congruence is a simulation relation for the standard terms in $\LTS$.
\begin{proposition}\label{prop:lts-bisimulation}
Let $\mathtt t\in\LTS(P,{Q\uparrow x})$ be a standard term. If there is a process $P'$ with $P\sim P'$, then there is a standard term $\mathtt t'\in\LTS(P',{Q'\uparrow x})$ such that $Q\sim Q'$.
\end{proposition}

Translating the above proposition into standard notation (using Proposition~\ref{prop:lts-string-corresp}), we obtain the following.

\noindent\begin{minipage}{0.65\textwidth}
\begin{corollary}\label{cor:lts-simulation}
The following rule is admissible in the LTS semantics:
\end{corollary}
\end{minipage}
\begin{minipage}{0.35\textwidth}
  \noindent\begin{bprooftree}
    \AxiomC{$P'\sim P$}
    \AxiomC{$P\xrightarrow x Q$}
    \BinaryInfC{$\exists Q'.~{Q'\sim Q}.~{P'\xrightarrow x Q'}$}
  \end{bprooftree}.
\vspace{3pt}
\end{minipage}

The following lemma connects the silent fragments of $\LTS$ and $\Red$. Note that the conclusion of the lemma is exactly the same as the conclusion of Lemma~\ref{lma:red-term-to-comp}, but this time we start from a term in $\LTS$.
\begin{lemma}\label{lma:lts-term-to-comp}
Let $\mathtt t\in\LTS(P,{Q\uparrow\tau})$ be a standard term. Then there is a term~\eqref{red-iso-standard} in the layered theory CCS such that the term $t'\in\Comp(\Red)(P',{Q'\uparrow})$ is standard, and $\mathtt i$ and $\mathtt j$ are isomorphisms.
\end{lemma}

\begin{theorem}\label{thm:red-lts-term-iff}
There is a term in $\Red(P,{Q\uparrow})$ if and only if there is a process $Q'$ with $Q\sim Q'$ and a standard term in $\LTS(P,{Q'\uparrow\tau})$.
\end{theorem}

In light of Theorem~\ref{thm:red-lts-term-iff}, one can think of the layered theory CCS as extending the (rather restricted) derivations in $\Comp(\Red)$ in two different but equivalent ways: $\Red$ closes all the terms under isomorphism (congruence), while $\LTS$ adds the symmetric counterparts of the parallel composition and reduction generators, as well as permits pending actions. The theorem then shows that $\Red$ and $\LTS$ have the same transitions from the active to the silent fragment. Using Propositions~\ref{prop:red-string-corresp} and~\ref{prop:lts-string-corresp}, we translate this to the usual terminology as follows.
\begin{corollary}[Soundness and completeness]\label{cor:ccs-soundness-completeness}
The reduction semantics is sound and complete with respect to the silent transitions in the LTS semantics: a rewrite rule $P\rightarrow Q$ is derivable if and only if there is a process $Q'$ such that $Q'\sim Q$ and the labelled transition $P\xrightarrow{\tau}Q'$ is derivable.
\end{corollary}

The terms in $\Red(P,{Q\uparrow})$ and $\LTS(P,{Q\uparrow x})$ correspond to the {\em derivations} of $P\rightarrow Q$ and $P\xrightarrow x Q$. Note, however, that the string diagrammatic terms carry slightly more information than the derivation trees, as they record the precise way in which the congruence is constructed.

We conclude this case study by using the layered monoidal theory to study different derivations of the rewrite rule~\eqref{eqn:CCS-derivation}: we give its derivation in $\Red$, and two derivations in $\LTS$. The first derivation in $\LTS$ is obtained by first decomposing the term as in Lemma~\ref{lma:red-term-to-comp} and then translating from $\Comp(\Red)$, while the second one is constructed directly and does not corresponding to any derivation in $\Red$. We begin with a derivation in $\Red$ (term on the top):
\begin{equation*}
\scalebox{.8}{\input{journal-figures/CSS-red.tikz}}.
\end{equation*}
In the term on the bottom, we have applied the translation $\Comp(\Red)\rightarrow\LTS$, drawn as the cobox. The isomorphisms at each side can be seen as ``preparations'' that ensure that the computations in $\Comp(\Red)$ can occur.

Next, we give the direct derivation as a standard term in $\LTS$:
\begin{equation*}
\scalebox{.8}{\input{journal-figures/example-term-lts.tikz}}.
\end{equation*}
Note that this derivation is not in in the image of the translation $\Comp(\Red)\rightarrow\LTS$. However, if we apply to this term the procedure of Lemma~\ref{lma:lts-term-to-comp}, we obtain precisely the term in $\Red$ we started with.

From the point of view of communication, one can think of the LTS semantics as a {\em refinement} of the reduction semantics: while the same synchronisations can be achieved in both, the LTS semantics has more redundancy in their implementation: there are multiple places where the passage from the active to the pending fragment can occur.

%% file: figures/journal-figures/ccs-layers.tikz
\begin{tikzpicture}
	\begin{pgfonlayer}{nodelayer}
		\node [style=objectbox] (6) at (-4, 0) {$\Red$};
		\node [style=objectbox] (7) at (0, 0) {$\Comp(\Red)$};
		\node [style=objectbox] (8) at (4, 0) {$\LTS$};
	\end{pgfonlayer}
	\begin{pgfonlayer}{edgelayer}
		\draw [style=morphism] (7) to (8);
		\draw [style=hookarrowmirror] (7) to (6);
	\end{pgfonlayer}
\end{tikzpicture}

%% file: tex/tocl-article/examples/glucose.tex
Here we construct an example that illustrates how in a layered theory of the shape $A\rightarrow C\leftarrow B$ the terms internal to $B$ and $C$ can be used to construct terms whose types are internal to $A$. We think of $A$ as a {\em high-level} language, of $C$ as a {\em low-level} description with few constraints, and of $B$ as an {\em explanatory} level describing those low-level processes which will actually occur. The fact that we can construct a term between the types in $A$ using only the terms from $B$ and $C$ can be seen, in a sense, as an instance of an emergent phenomenon: there is seemingly a process between the entities in $A$, while in reality all the processes live in $B$ and $C$.

We work inside a model for organic chemistry that has just enough structure to talk about an important biochemical process known as {\em phosphorylation of glucose}. The discussion here is inspired by a talk by Krivine~\cite{krivine-talk}, which, in turn, is motivated by the problem of systematising a vast amount of experimental data in systems biology in a way that is easy for humans to both understand and use~\cite{krivine-siglog}. Our strategy is to define a layered monoidal theory with three layers, each corresponding to viewing chemical change at a different abstraction level:
\begin{center}
\begin{tabular}{ c | c }
$\Name$ & English names of the relevant molecules \\
\hline
$\Scheme$ & Reaction schemes that represent reaction mechanisms \\
\hline
$\Disc$ & Local rewrite rules that capture all possible chemical change
\end{tabular}
\end{center}

Let us fix the following notions needed for Definition~\ref{def:chemlabgraph-layered}: a countable set of {\em vertex names} $\VS$, a finite set of {\em vertex labels} $\Atset$, which contains the special symbol $\alpha$, and finally the set of {\em edge labels} $\Lab\coloneqq\{0,1,2,3,4,\ib\}$.

We denote the vertex names by either positive integers or lowercase Latin letters, as appropriate to the situation. While formally the set of vertex labels can contain arbitrary symbols, in what follows we shall assume that $\Atset$ contains a symbol for each main-group element of the periodic table: $\{H,C,O,P,\dots\}\sse\Atset$. For this reason, we will also refer to $\Atset$ as the {\em atom labels}. The special symbol $\alpha$ may be thought of as representing an unpaired electron or a free charge. The integers $\{0,1,2,3,4\}$ in the set of edge labels stand for covalent bonds, while $\ib$ stands for an ionic bond.

\begin{definition}[Chemically labelled graph]\label{def:chemlabgraph-layered}
A {\em chemically labelled graph} is a triple $(V,\tau,m)$, where $V\sse\VS$ is a finite set of {\em vertices}, $\tau:V\rightarrow\Atset\times\Z$ is a {\em vertex labelling function}, and $m:V\times V\rightarrow\Lab$ is an {\em edge labelling function} satisfying $m(v,v)=0$ and $m(v,w)=m(w,v)$ for all $v,w\in V$.
\end{definition}
Thus, a chemically labelled graph is irreflexive (we interpret the edge label $0$ as no edge) and symmetric, and each of its vertices is labelled with an element of $\Atset$, together with an integer indicating the charge.

When drawing chemically labelled graphs, we adopt the following conventions: (1) the vertex label from $\Atset$ is drawn at the centre of a vertex, (2) the vertex name is drawn as a superscript on the left (so within a single graph, no left superscript appears twice), (3) a non-zero charge is drawn as a superscript on the right (hence the lack of a right superscript indicates zero charge), (4) the charge $-1$ is abbreviated as $-$, and similarly the charge $1$ as $+$, (5) $n$-ary covalent bonds are drawn as $n$ parallel lines, (6) ionic bonds are drawn as dashed lines.

\noindent\begin{minipage}{0.6\textwidth}
\begin{example}\label{ex:phosphate}
The phosphate functional group with one unpaired electron is drawn as a chemically labelled graph on the right:
\end{example}
\end{minipage}
\begin{minipage}{0.4\textwidth}
\begin{equation*}
\scalebox{.8}{\input{journal-figures/example-phosphate.tikz}}.
\end{equation*}
\end{minipage}

A {\em disconnection rule} is a partial endofunction on the set of chemical graphs that only changes one atom or bond (and the atoms it connects). We define four classes of disconnection rules, all of which have a clear chemical significance: two versions of {\em electron detachment}, {\em ionic bond breaking} and {\em covalent bond breaking}. We give an informal definition in Figure~\ref{fig:disc-rules-layered}: the left-hand side of the arrow with the conditions written above define the application conditions, and the output of the function is obtained by replacing the left-hand side with the right-hand side. The full formal definition can be found in~\cite{tcs-cat-model-org-chem}.

\begin{figure}
\centering
\scalebox{.8}{\input{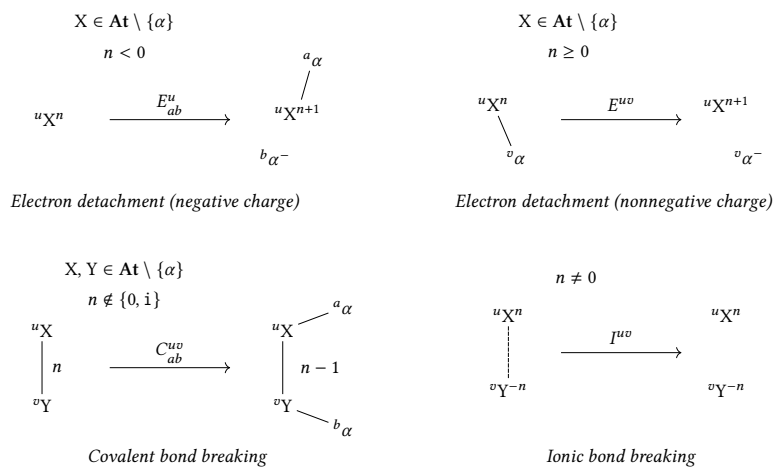}}
\caption{The four disconnection rules.\label{fig:disc-rules-layered}}
\end{figure}

We observe that each disconnection rule is injective (as a partial function), and hence has an inverse partial function. The disconnection rules or their inverses only add or remove the $\alpha$-vertices, hence they are ``matter preserving''. Moreover, valence is preserved locally at each atom, in the sense that the sum of the absolute value of the charge and the number of incident bonds remains unchanged; while the charge is preserved globally, in the sense that the net charge is the same before and after a rule is applied. In~\cite{tcs-cat-model-org-chem}, we show that the disconnection rules (together with their inverses) capture all the charge and matter preserving combinatorial rearrangements of chemical graphs, which can be thought of as formal reactions.

We define a layered monoidal theory of the shape \scalebox{1}{\input{journal-figures/glucose-layers.tikz}}. There are no generators in $\Name$, whose colours are given by $\{\glucose, \ATP, \glucosesix, \ADP, \hydrogenion\}$. Here $\ATP$ and $\ADP$ stand for {\em adenosine triphosphate} and {\em adenosine diphosphate}.

For both $\Disc$ and $\Scheme$, the colours are the chemically labelled graphs. The {\em structural generators} of $\Disc$ are the symmetry together with the following: \scalebox{.8}{\input{journal-figures/glucose-structural-generators.tikz}}, where $\eset$ is the unique chemically labelled graph on the empty set, $A$ and $B$ have disjoint sets of vertex names and $u(A,B)$ is the union graph, while $r(A)$ is any graph obtained from $A$ by only renaming some vertex names. We assume that the union of two graphs is a symmetric operation: $u(A,B)=u(B,A)$. The structural generators capture the up-to-an-isomorphism notion of the disjoint union of chemically labelled graphs as a non-strict monoidal structure. The ``union'' generators can also be given a chemical interpretation: either the molecules $A$ and $B$ have come close enough ($u(A,B)$) for a reaction to occur, or, dually, a group of molecules $u(A,B)$ splits into spatially separated groups $A$ and $B$, so that reactions can no longer occur between them.

The other generators in $\Disc$ are given by: \scalebox{.8}{\input{journal-figures/disc-generators.tikz}}, where $d$ is a disconnection rule (defined as four classes of partial function in Figure~\ref{fig:disc-rules-layered}), $A$ is in the domain of $d$, and $d(A)$ is the result of applying $d$ to $A$.

The only generators of $\Scheme$ are given by the family
\begin{equation*}
\scalebox{.8}{\input{journal-figures/scheme-generator-phosph.tikz}},
\end{equation*}
which is universally quantified over all vertex names: that is, the vertices can be renamed arbitrarily, as long as renaming is done simultaneously in the domain and the codomain graphs.

The layered monoidal theory is defined as follows. The structural generators are isomorphisms in $\Disc$: the first four equations in Figure~\ref{fig:ccs-layered-theory} hold whenever the composition is defined and the types match, and similarly, we add the following equations\footnote{Note that the first equation is well-typed since the union is symmetric.}: \scalebox{.8}{\input{journal-figures/glucose-structural-eqns.tikz}}. The functor $\Name\rightarrow\Disc$ is defined in Figure~\ref{fig:translation-name-disc}, while the functor $\Scheme\rightarrow\Disc$ is identity on colours, and is defined on the generators by the following equation: \scalebox{.8}{\input{journal-figures/phosph-scheme-as-disc-rules.tikz}}. While it might seem that the functor $\Scheme\rightarrow\Disc$ is not doing much, as it simply encodes a sequence of disconnection rules in $\Disc$, it will allow us to define where the non-trivial chemical change is happening: namely, we shall require that the only such change happens if a morphism in $\Disc$ is in the image of this functor.

\begin{figure}
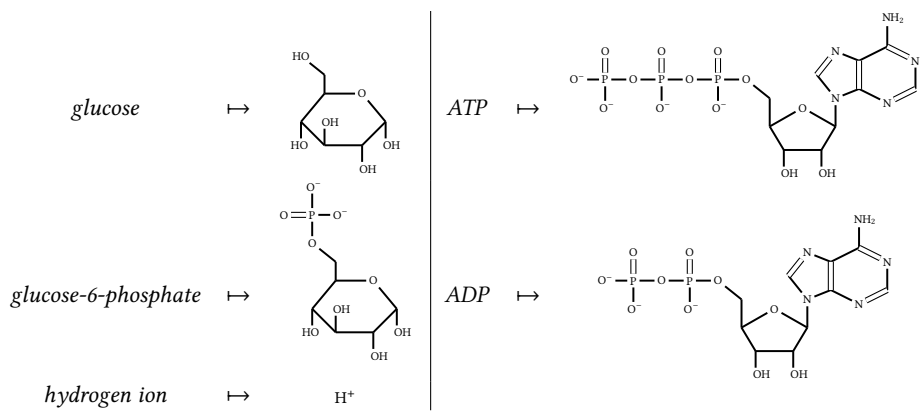

\centering
\begin{tabular}{ c c c | c c c }
$\glucose$ & $\mapsto$ & \scalebox{.5}{\input{journal-figures/glucose.tikz}} & $\ATP$ & $\mapsto$ & \scalebox{.5}{\input{journal-figures/ATP.tikz}} \\
$\glucosesix$ & $\mapsto$ & \scalebox{.5}{\input{journal-figures/glucosesix.tikz}} & $\ADP$ & $\mapsto$ & \scalebox{.5}{\input{journal-figures/ADP.tikz}} \\
$\hydrogenion$ & $\mapsto$ & \scalebox{.7}{\input{journal-figures/H.tikz}} & & &
\end{tabular}
\caption{The functor $\Name\rightarrow\Disc$. We adopt the usual organic chemistry convention that an unlabelled vertex stands for a carbon atom with an appropriate number of hydrogen atoms attached to make the valence add up to $4$. For clarity, we omit the vertex names, as their precise choice is immaterial: the only requirement we impose is that the vertex sets of all five graphs are pairwise disjoint.\label{fig:translation-name-disc}}
\end{figure}

We are now ready to show that this mathematically fairly simple setup allows us to derive the reaction for glucose phosphorylation, which in $\Name$ would be expressed as:
\begin{equation}\label{eqn:phosphorylation}
\glucose+\ATP \longrightarrow \glucosesix+\ADP+\hydrogenion,
\end{equation}
where $+$ denotes the monoidal product. Of course, such reaction cannot exist as a morphism in $\Name$, as it only contains the identity morphisms. However, it is derivable within the layered theory. Note that reaction~\eqref{eqn:phosphorylation} is very close to what one might see as a high-level explanation in a chemistry textbook. We give the derivation in Figure~\ref{fig:phosphorylation-layers}, where we omit the parts of the larger molecules that remain unchanged. Note that the term in Figure~\ref{fig:phosphorylation-layers} can be constructed such that the following properties hold: (1) inside the cobox, no other generators of $\Disc$ occur after the renaming generators, (2) the vertex names in the subscript introduced by a single disconnection rule ($ab$ or $cd$) must also be removed by a single inverse rule (the generators with a bar). Importantly, restriction (2) does not apply to the generators appearing inside a box in the translation of $\mathtt{phosph}$. The interpretation is that part of a sequence of disconnection rules is recognised as a translation of a reaction scheme in $\Scheme$. Properties (1) and (2) guarantee that the only non-trivial chemical change occurs within the reaction scheme: the other disconnections must be patched back to restore the original configuration.

Without restrictions (1) and (2), one could always derive a sequence of disconnection rules in $\Disc$ resulting in reaction~\ref{eqn:phosphorylation}, as both sides of the reaction have the same atoms and charge as the ``building material''. The layers $\Disc$ and $\Scheme$ can, therefore, be thought to carry information at distinct epistemic levels: the terms in $\Disc$ are {\em possibilistic}, expressing which chemical change is possible as far as conservation of atoms and charge is concerned, while the terms in $\Scheme$ contain {\em reaction mechanisms}, which can represent empirical knowledge or hypothetical assumptions.

\begin{figure}
\centering
\scalebox{.8}{\input{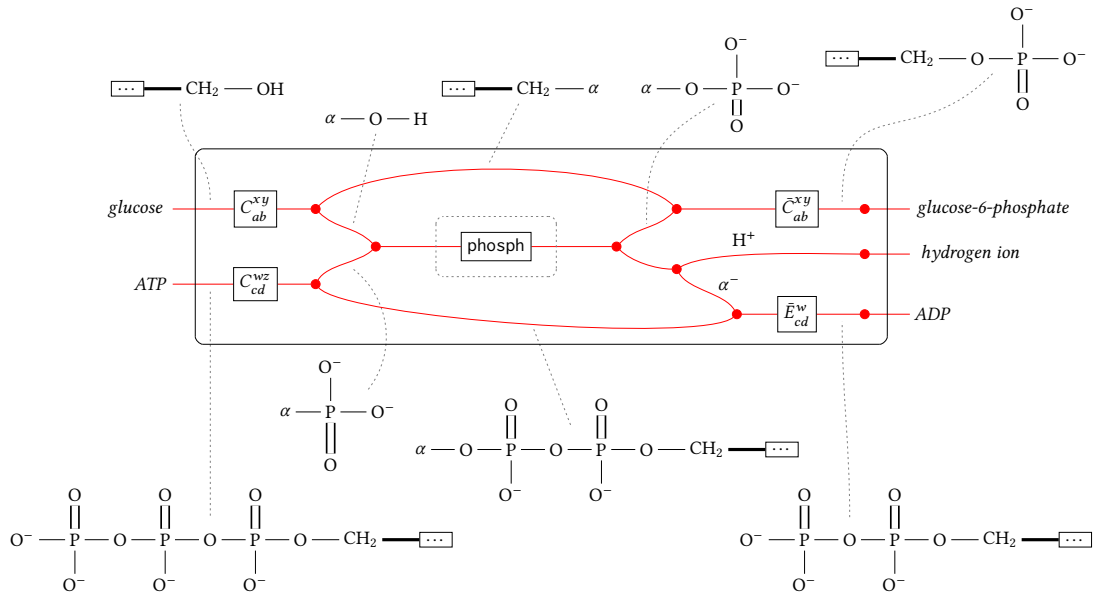}}
\caption{Glucose phosphorylation as a layered term.\label{fig:phosphorylation-layers}}
\end{figure}

%% file: figures/journal-figures/glucose-layers.tikz
\begin{tikzpicture}
	\begin{pgfonlayer}{nodelayer}
		\node [style=textbox] (3) at (-4, 0) {$\Name$};
		\node [style=textbox] (4) at (4, 0) {$\Scheme$};
		\node [style=textbox] (5) at (0, 0) {$\Disc$};
	\end{pgfonlayer}
	\begin{pgfonlayer}{edgelayer}
		\draw [style=diredge] (3) to (5);
		\draw [style=diredge] (4) to (5);
	\end{pgfonlayer}
\end{tikzpicture}

%% file: tex/tocl-article/examples/prob-channels.tex
Here our aim is to capture within a layered theory the `shaded box' notation introduced by Jacobs~\cite{jacobs-spr} in order to represent {\em conditionals} in probability theory. In particular, we define the conditional of a channel as a certain term (in fact, a box) in a three-layered monoidal theory. In our formalisation of the shaded box notation, we upgrade it in several significant ways, which are discussed in detail in Remark~\ref{rem:shaded-box}. As the first step, we use {\em (co)parametric channels}, which allow part of the output (or input) of a channel as a parameter to be conditioned over. We then define the conditional of a coparametric channel as a certain parametric channel. The original shaded box notation for conditionals fails to be functorial for two reasons: its action on objects depends on whether an object appears as a domain or a codomain of a channel, while on the level of channels, the normalisation results in non-preservation of composition. Our construction remedies the first issue by passing to (co)parametric categories. We further show that the construction is functorial on an important fragment: it preserves composition with the `trivially' coparameterised channels (Lemma~\ref{lma:disintegration-functorial}). We use this ``restricted functoriality'' to prove Proposition~\ref{prop:disintegration-depara}, which gives a diagrammatic proof of uniqueness of conditional channels with marginally full support (Definition~\ref{def:marg-full-support}).

Our presentation here is less axiomatic than in the previous examples: we do not give a monoidal theory axiomatising finite probabilistic channels $\Chan$ with the Kronecker product (Definition~\ref{def:prob-channels}), as such an axiomatisation is currently not known. However, Fritz has given an axiomatisation using the coproduct monoidal structure~\cite{fritz2009}, while Piedeleu, Torres-Ruiz, Silva and Zanasi have axiomatised the full monoidal subcategory of $\Chan$ on sets whose cardinality is a power of $2$~\cite{complete-disc-prob-prog2024}; Bonchi and Cioffo have given an alternative axiomatisation of the same subcategory using two monoidal products~\cite{tapesstochasticmatrices2026}. Hence, enough monoidal structure can be interpreted in the synthetic definition~\ref{def:prob-channels}, justifying usage of string diagrams, the lack of a complete axiomatisation notwithstanding. To make the usage of string diagrams more precise, Proposition~\ref{prop:disintegration-depara} holds for any class of morphisms in a Markov category~\cite{fritz-markov20} for which Lemmas~\ref{lma:disintegration-functorial} and~\ref{lma:fs-id} can be established.

By a {\em distribution} on a set $X$ we mean a function $\omega:X\rightarrow [0,1]$ from $X$ to the unit interval. The {\em support} of a distribution $\omega:X\rightarrow [0,1]$ is the subset $\supp(\omega)\coloneq\{x\in X : \omega(x)\neq 0\}$. We say that a distribution $\omega:X\rightarrow [0,1]$ has {\em full support} if $\supp(\omega) = X$, and write $\left|\supp(\omega)\right|<\aleph_0$ for when the support is finite.

Let $\mathcal D:\Set\rightarrow\Set$ be the finite distribution monad, i.e. $\mathcal D(X)\coloneq\left\{ \omega : X\rightarrow [0,1] : \left|\supp(\omega)\right|<\aleph_0, \sum_{x\in X} \omega(x) = 1\right\}$, and given a function $f:X\rightarrow Y$, the map $\mathcal D(f):\mathcal D(X)\rightarrow\mathcal D(Y)$ is defined by
$\mathcal D(f)(\omega)(y)\coloneq\sum_{x\in f^{-1}(y)}\omega(x)$.

\begin{definition}[Probabilistic channels]\label{def:prob-channels}
The category of {\em finite probabilistic channels} $\Chan$\footnote{The morphisms in this category are often called {\em finite stochastic maps / processes / matrices}. We stick with probabilistic channels, often simply saying `channels'.} has the finite sets as objects, while a morphism $X\circarrow Y$ is given by a function $X\rightarrow\mathcal D(Y)$. The identity $X\circarrow X$ is given by the delta distribution: $x\mapsto\delta_x$, while the composition of $f:X\circarrow Y$ and $g:Y\rightarrow Z$ is given by the formula $gf(x)(z) = \sum_{y\in Y} f(x)(y)\cdot g(y)(z)$.
\end{definition}
We say that a channel $f:X\circarrow Y$ has {\em full support} if for every $x\in X$, the distribution $f(x)$ has full support. Note that $\Chan$ is the full subcategory on the finite sets of the Kleisli category of the finite distribution monad $\mathcal D$.

The parametric channels arise via the {\em (co)para construction}, whose construction as a monoidal theory we cover in Appendix~\ref{sec:para-copara}. We refer to the morphisms in $\Para(\Chan)$ as {\em parametric channels}, and to the morphisms in $\Copara(\Chan)$ as {\em coparametric channels}. Explicitly, the parametric channel $(X,f):Z\rightarrow Y$ is a channel $f:X\times Z\rightarrow Y$, where $X$ is viewed as part of the morphism data, and the coparametric channel $(V,g):Z\rightarrow Y$ is a channel $g:Z\rightarrow V\times Y$, where $V$ is viewed as part of the morphism data. Graphically, the parametric and coparametric channels are represented as follows: \scalebox{.75}{\input{journal-figures/para-copara-channels.tikz}}. In the above situation, we call $X$ and $V$ the {\em parameter} and the {\em coparameter}, and say that $f$ is parameterised by $X$ and that $g$ is coparameterised by $V$.

\begin{definition}[Jacobs~\cite{jacobs-spr}, Definition~7.3.1]\label{def:disintegration}
A parametric channel $(X,f'):Z\rightarrow Y$ is the {\em conditional} for a coparametric channel $(X,f):Z\rightarrow Y$ if equation~\eqref{eq:disintegration} holds, and for any channel $h:Z\rightarrow X$ with full support and any parametric channel $(X,g):Z\rightarrow Y$, the equality on the left of~\eqref{eq:disintegration-unique} implies the equalities on the right of~\eqref{eq:disintegration-unique}:
\noindent\begin{minipage}{.4\textwidth}
\begin{equation}\label{eq:disintegration}
\scalebox{.75}{\input{journal-figures/disintegration-eqn.tikz}},
\end{equation}
\end{minipage}%
\begin{minipage}{.6\textwidth}
\begin{equation}\label{eq:disintegration-unique}
\scalebox{.75}{\input{journal-figures/disintegration-unique.tikz}}.
\end{equation}
\end{minipage}
\end{definition}

\begin{definition}\label{def:marg-full-support}
We say that a morphism $(X,f):Z\rightarrow Y$ in $\Copara(\Chan)$ has {\em marginally full support} if the channel \scalebox{.75}{\input{journal-figures/marg-full-support.tikz}} has full support.
\end{definition}
Note that a coparametric channel $(X,f):Z\rightarrow Y$ has marginally full support if and only if for all $z\in Z$ and $x\in X$ we have that the sum $\sum_{y\in Y}f(z)(x,y)$ is non-zero. We will solely focus on coparametric channels with marginally full support, as for them the conditional channels exist in the sense of Definition~\ref{def:disintegration} (see Proposition~\ref{prop:disintegration-depara}). Let us denote by $\Mfs$ those coparametric channels in $\Copara(\Chan)$ that have marginally full support.

Define the {\em conditional box} $B:\Mfs\rightarrow\ParaSwap(\Chan)$ by mapping $(X,f):Z\rightarrow Y$ to $(X,B(f))$, where $B(f):X\times Z\circarrow Y$ is the channel defined by $B(f)(x,z)(y)\coloneq\frac{f(z)(x,y)}{\mathcal N_f(z,x)}$, where $\mathcal N_f(z,x)\coloneq\sum_{y\in Y}f(z)(x,y)$.

\noindent\begin{minipage}{0.4\textwidth}
We thus have the diagram on the right, where the unlabelled arrows are the faithful embeddings mapping each morphism to itself (co)parameterised by the monoidal unit:
\end{minipage}
\begin{minipage}{0.6\textwidth}
\begin{equation*}
\scalebox{1}{\input{journal-figures/b-embeddings.tikz}}.
\end{equation*}
\end{minipage}
Note that $B$ is {\em not} a functor -- hence we cannot assume the functoriality equations of a layered theory (Figures~\ref{fig:structural-twocells-functors-int} and~\ref{fig:structural-twocells-functors-ext}). However, since $B$ preserves the domain and the codomain, we may still use the terms generated by treating the above diagram as a layered signature to reason about the properties of conditionals. Moreover, Lemma~\ref{lma:disintegration-functorial} shows that $B$ is functorial on (a superset of) the embedding of $\Chan$, and that the above diagram commutes.

Graphically, the action of $B$ is denoted as follows: \scalebox{.75}{\input{journal-figures/disintegration-functor.tikz}}. Note that we omit the subscript (label) $B$ as there is only one translation involved.

\begin{remark}[A note on deparameterisation]\label{rem:deparameterisation}
Since applying the conditional box to a coparametric channel yields a parametric channel, the result admits deparameterisation, as defined in Appendix~\ref{sec:para-copara}, denoted by~\eqref{eq:disintegration-channel}. This notation makes explicit the computational steps taken to obtain this channel: one starts with a {\em bona fide} channel $f:Z\rightarrow X\times Y$, one then specifies which part of the codomain is to be regarded as the coparameter, thus passing to the coparametric channel $(X,f):Z\rightarrow Y$, after which one applies the conditional box, obtaining the parametric channel $(X,Bf):Z\rightarrow Y$.

\vspace{3pt}
\noindent\begin{minipage}{0.75\textwidth}
Finally, one views the parametric channel as the honest channel $Bf:X\times Z\rightarrow Y$. Proposition~\ref{prop:disintegration-depara} that we are about to prove shows that the deparameterisation of a conditional box gives the conditional of the original channel (Definition~\ref{def:disintegration}).
\end{minipage}
\begin{minipage}{0.25\textwidth}
\begin{equation}\label{eq:disintegration-channel}
\scalebox{.75}{\input{journal-figures/disintegration-channel.tikz}}
\end{equation}
\end{minipage}
\end{remark}

\begin{remark}\label{rem:shaded-box}
The graphical depiction of the conditional box has been introduced as an informal notation by Jacobs~\cite[Section~7.3]{jacobs-spr}, where it is called the `shaded box'. However, our notation differs from the shaded box notation in three crucial ways. First, it is defined formally as a certain term in a layered monoidal theory; consequently, inside our shaded box, all the usual equalities of channels hold. Second, we keep a track of the order of the parameters. For example, the second equality in~\cite[Exercise~7.3.7]{jacobs-spr} becomes~\eqref{disintegration-nested}, making it apparent that there is a swap of order of the parameters on the right-hand side. This equality can then be used to justify ignoring the order. Third, thinking of the coparameter wire as being attached to the right side of the box gives a good visual intuition for the equalities that hold for the shaded box: any morphism that is not attached (so that there are no obstructions to dragging it out) is able to slide out of the

\vspace{3pt}
\noindent\begin{minipage}{0.55\textwidth}
box -- see Lemma~\ref{lma:disintegration-functorial} and~\cite[Exercise~7.3.7]{jacobs-spr}. Since the parameter on the left has to have the same type as the coparameter attached to the right side of the box, the intuition of ``bending" the wire backwards as in~\cite{jacobs-spr} is preserved by thinking of the boundary of the box as carrying the type information.
\end{minipage}
\begin{minipage}{0.45\textwidth}
\begin{equation}\label{disintegration-nested}
\scalebox{.75}{\input{journal-figures/disintegration-nested.tikz}}
\end{equation}
\end{minipage}
\end{remark}

\begin{lemma}\label{lma:disintegration-functorial}
Let $(1,g):Y\rightarrow W$ be a channel coparameterised by the terminal set $1$. Then $B$ acts as identity on $(1,g)$, in the sense that equality~\eqref{eq:disintegration-id} holds. Moreover, if $(X,f):Z\rightarrow Y$ is any coparametric channel in $\Mfs$ composable with $(1,g)$, then $B$ preserves the composition $(X,f);(1,g)$, in the sense that equality~\eqref{eq:disintegration-functorial-wobox} holds:

\noindent\begin{minipage}{.5\textwidth}
\begin{equation}\label{eq:disintegration-id}
\scalebox{.75}{\input{journal-figures/disintegration-id.tikz}},
\end{equation}
\end{minipage}%
\begin{minipage}{.5\textwidth}
\begin{equation}\label{eq:disintegration-functorial-wobox}
\scalebox{.75}{\input{journal-figures/disintegration-functorial-wobox.tikz}}.
\end{equation}
\end{minipage}
\end{lemma}

\begin{lemma}\label{lma:fs-id}
Let $h:Z\rightarrow X$ be any channel with full support. We then have \scalebox{.75}{\input{journal-figures/disintegration-fs-id.tikz}}.
\end{lemma}

\begin{proposition}\label{prop:disintegration-depara}
Every channel in $\Mfs$ has the conditional channel, and the unique conditional channel of $(X,f):Z\rightarrow Y$ is given by the deparameterisation of $B(X,f)$~\eqref{eq:disintegration-channel}.
\end{proposition}
\begin{proof}
Verifying that equation~\ref{eq:disintegration} holds when we take $f'$ to be the above box is a matter of computing the composite channel. For uniqueness, we give a diagrammatic argument using Lemmas~\ref{lma:disintegration-functorial} and~\ref{lma:fs-id}. Hence suppose $h$ and $g$ are as in Definition~\ref{def:disintegration}, such that their composite is equal to $f$. The first equality then follows straightforwardly:
\begin{equation*}
\scalebox{.75}{\input{journal-figures/disintegration-uniqueness-eq1.tikz}}.
\end{equation*}
The second equality relies on the aforementioned Lemmas:
\begin{equation*}
\scalebox{.75}{\input{journal-figures/disintegration-uniqueness-eq2.tikz}}.
\end{equation*}
\end{proof}

We have chosen to present the construction of the conditional box concretely in $\Chan$ for the sake of demonstrating how one can capture existing constructions using layered theories. However, it should be possible to express the conditional box more generally (and in a more principled manner) in the setting of copy-discard categories. Such a construction would follow existing axiomatisations of normalisation by Lorenz and Tull~\cite{lorenz-tull-causal} and its use for axiomatising conditionals in the discrete case by~\cite{lorenz-tull-causal} and Jacobs, Sz\'eles and Stein~\cite{jacobs-szeles-stein25}.

%% file: tex/tocl-article/layered-discussion.tex
We have introduced layered monoidal theories (Definition~\ref{def:layered-theory}) as an algebraic means for studying layers of abstraction. Our analysis led us to identifying three families of structural equations, giving rise to opfibrational, fibrational and deflational layered theories (Subsection~\ref{subsec:opfib-defl-theories}). This gives a flavour of what kinds of abstract theories are definable using the formalism. Beyond defining the theories, we have started exploring what kinds of constructions are definable {\em within} deflational theories by constructing functor boxes and coboxes in Section~\ref{sec:functor-boxes}, and showing how they can be used to detect properties about the monoidal theories and translations at hand. In Section~\ref{sec:layered-examples}, we have provided ample examples of how we envision applications of the general theory, hopefully convincing the reader that the theoretical developments have been worthwhile. We conclude with a discussion of related work and potential future applications.

\subsection{Related work}\label{subsec:related-work}

{\em Functor boxes} are a common graphical notation for reasoning about functors, as discussed for example by Melli\` es~\cite{functorial-boxes}. The {\em functor boundaries} appear as part of a graphical notation for categories and functors in Hinze \& Marsden~\cite{hinze-marsden-2023}. The decomposition of a functor box into the ``forward'' and the ``backward'' boundaries (Proposition~\ref{prop:cowindow-box-equality}), as well as the notion of a {\em functor cobox} (Definition~\ref{def:cobox}) are an evolution of these graphical notations.

The {\em internal string diagram} notation was first introduced in the study of topological quantum field theories (TQFTs) by Bartlett, Douglas, Schommer-Pries and Vicary~\cite{modular-categories}. The notation was further exploited in the study of profunctors and traced monoidal categories by Hu~\cite{hu-thesis} and Hu \& Vicary~\cite{hu-vicary21}. While these works focus more on the semantic aspects, our contribution provides a syntactic perspective via study of generators and relations for internal string diagrams.

Rom\' an~\cite{roman} has introduced so-called {\em open diagrams} that represent elements of coends. The semantics of the diagrams is given by {\em pointed profunctors}, which are closely related to the models of deflational theories.

Braithwaite and Rom\' an~\cite{braithwaite-roman23} have studied {\em collages of string diagrams}, which is closely related to the semantics of deflational theories. A collage ``flattens'' -- or creates a collage of -- a {\em bimodular category}: a category with compatible left and right monoidal actions. The terms~\ref{term:monoid} and~\ref{term:comonoid} can be thought of as the left and right monoidal actions of a monoidal category on itself. This gives a potential direction for a generalisation of this work: what kinds of diagrams account for an action of distinct monoidal categories?

One way to think about the cobox (Definition~\ref{def:cobox}) is as a ``hole'' into which any morphism in the codomain can be plugged. There are many instances of holey string diagrams, such as in the definition of context-free languages of string diagrams by Earnshaw and Rom\' an~\cite{earnshaw-roman24}. Particularly interesting are the holes in the study of quantum combs, as these have been characterised by using internal string diagrams by Hefford and Comfort~\cite{hefford-comfort23}.

The diagrams for deflational theories visually resemble {\em tape diagrams}, which also consist of thin string diagrams embedded into thick string diagrams~\cite{tape-diagrams23,diagrammatic-program-logics25,tape-monoidal-monads25}. The resemblance is, however, largely superficial -- both syntactically and semantically. While the terms for the tape diagrams are a subset of the deflational terms (namely, the terms given by~\ref{term:diag}, \ref{term:diag-counit}, \ref{term:codiag} and~\ref{term:codiag-unit}), such terms do not satisfy the bialgebra equations that are part of the defining axioms of tape diagrams; moreover, the 2-cells between these terms (Figure~\ref{fig:structural-twocells-adjoints}) go in the opposite direction to that of the inequalities of tape diagrams (see e.g.~Figures~1 \&~2 in~\cite{tape-diagrams23}). Consequently, the semantics of the terms is different: while deflational theories ultimately encode indexed monoidal categories, the tape diagrams are interpreted in {\em rig categories} -- the two different monoidal products correspond to the ``internal'' tensor inside the tapes and the ``external'' tensor of the tapes.

\subsection{Further developments}\label{subsec:potential-applications}

The immediate theoretical question arising from the syntactic constructions presented here is the semantic interpretation: what are the models characterised by layered monoidal theories? This is the subject of Part~II~\cite{lmt-part2} of this work. Here we list other potential developments and applications of layered monoidal theories.

\paragraph{Connections with linear logic}
The connection between linear logic and profunctors has been explored in the literature several times, for example by Dunn~\cite{dunn-thesis} and Comfort~\cite{comfort-profunctors-linear-logic}. We expect that the syntactic approach presented here will help in making precise the connection between internal string diagrams and proof nets~\cite{comfort-profunctors-linear-logic}.

\paragraph{Infinite number of layers}
All the examples of layered monoidal theories we have seen in Section~\ref{sec:layered-examples} -- with the exception of digital circuits in Subsection~\ref{subsec:digital-circuits} -- have had a finite number of layers. There is, of course, no reason to restrict one's attention to the finite case: the digital circuits example already gives a glimpse of why it is natural to consider an infinite number of layers in certain contexts. Other situations featuring an infinite number of layers are:
\begin{enumerate}
\item {\em infinite monoidal products}: following the construction of infinite tensor products by Fritz and Rischel~\cite{fritz-rischel20}, it would be interesting to explore the case where the layers consist of the finite subsets of an infinite set representing the infinite monoidal product; the resulting layered theory would describe the ``finite stages of construction'' of the infinite tensor product,
\item {\em distances between string diagrams}: following the quantale enriched string diagrams introduced by Lobbia, Różowski, Sarkis and Zanasi~\cite{quantitative-monoidal-algebra}, it would be interesting to see whether the distance structure can arise from indexing rather than enrichment; the layers would be upward closed sets of a quantale ordered by inclusion, and passing from one layer to another could be seen as either losing or gaining precision,
\item more generally, it would be interesting to study order or domain theoretic structure on the indexing category, and see how it interacts with the monoidal structure.
\end{enumerate}

\paragraph{Time evolution}
The examples we have seen are static, in the sense that there is no notion of change, and all layers are assumed to be simultaneously accessible. One could introduce {\em time dependence} by having an order on the layers, and some notion of dynamics, determining which future layers are accessible from the current one. This would yield a picture in which a system evolves over time, and even the rules (i.e.~equations and 2-cells) that govern the system may differ depending on the time instance.

\paragraph{Microstates versus macrostates}
The view on abstraction taken in this work seems to be particularly suited for discrete systems with relatively few parts (or those that can be usefully modelled as such). Note that this excludes systems with dynamical qualities, such as perhaps the most well-known coarse-fine theory pair: thermodynamics and statistical mechanics in classical physics. The reason these theories do not fit the view presented here is that two different microstates (e.g.~motion of particles) in statistical mechanics (finer layer) may correspond to the same macrostate (e.g.~temperature) in thermodynamics (coarser level). Thus, in this case, the functional translation seems to be going in the reverse direction from that in Figure~\ref{fig:layers-of-abstraction}. It would be interesting to construct a layered monoidal theory capturing more dynamic settings of this kind.

\paragraph{Hierarchy of causality}
Hierarchical models of causality have become important with the emergence of neural network based machine learning~\cite{nashed2023,englberger2025,causalcompositionalabstraction2026}. Especially the models of causal abstractions by Englberger \& Dhami~\cite{englberger2025} and Lorenz \& Tull~\cite{causalcompositionalabstraction2026} fit particularly well with layered monoidal theories: both use string diagrams to represent causal models, while abstraction takes the form of a natural transformation.

\paragraph{Functional versus mechanistic explanations in biology}
Given the abundance of biological data from different abstraction levels, there is a need for formal tools in systems biology that are able to rigorously separate the levels of description, most notably the {\em mechanistic} and {\em functional} rules~\cite{rosen-life-itself,modular1999,krivine-siglog}. Differentiating the levels has at least two advantages. First, it gives a way to incorporate the vast amount of experimental data (that is growing at a fast rate) into mechanistic models, where the empirically inferred rules take the role of a higher level description~\cite{krivine-siglog}. Second, it provides a formal framework for {\em counterfactual models}, which has become an important aspect when reasoning about complex systems~\cite{counterfactual, pearl-causality, blaming2015, nashed2023}. Roughly, the lower layers can be thought of as {\em causes} for the {\em effects} at higher layers, allowing to mix-and-match the processes at lower levels to answer counterfactual questions.

\paragraph{Scalable ZX-calculus}
The {\em scalable notation} for the ZX-calculus~\cite{szx-calculus19} extends the monoidal theory of the ZX-calculus (Subsection~\ref{subsec:zx-extraction} and Appendix~\ref{sec:zx-mbqc}) to allow for wires that carry multiple qubits; the scalable notation for graphical linear algebra~\cite{graphicalsymplecticalgebra24} extends the same idea to all finite-dimensional vector spaces. The generators used are highly similar to the {\em bundlers} of Kaye~\cite{kaye-thesis} -- suggesting a similar treatment as the layered theory of digital circuits in Subsection~\ref{subsec:digital-circuits} is possible. The notation further introduces matrices and their transposes that translate back and forth between the wires of different thickness: the laws obeyed by such matrices are reminiscent of the ``sliding'' equations for the bidirectional functor boundaries (Figure~\ref{fig:structural-twocells-functors-ext}) -- suggesting an even closer connection with deflational theories.

%% file: tex/tocl-article/appendix/ccs-proofs.tex
Here we give the proofs omitted in the discussion of the calculus of communicating systems.

\begin{lemma}[Lemma~\ref{lma:iso-congruence}]
Let $P$ and $Q$ be processes. The following are equivalent, where all the isomorphisms are in $\Red$:
\begin{align*}
\text{(1) } P &\sim Q, & \text{(2) } P &\xrightarrow{\sim}Q, & \text{(3) } {P\quest} &\xrightarrow{\sim}Q\quest, & \text{(4) } {P\uparrow} &\xrightarrow{\sim}Q\uparrow.
\end{align*}
\end{lemma}
\begin{proof}
We prove the equivalence between (1) and (2) -- the proofs for the other two are analogous, using the isomorphisms in the corresponding fragment.

First, note that the relations defining the congruence $\sim$ (Definition~\ref{def:ccs-congruence}) are obtained via the following isomorphisms in the active fragment:
\begin{equation*}
\scalebox{.8}{\input{journal-figures/ccs-congruence-iso.tikz}},
\end{equation*}
where in the last isomorphism the terms $\mathtt t:P\xrightarrow{\sim}P'$ and $\mathtt s:Q\xrightarrow{\sim}Q'$ are inductively assumed to be isomorphisms. Reflexivity and transitivity are then witnessed by the categorical structure, while symmetry is witnessed by inverses. This shows that (1) implies (2).

For proving the converse, let us assign the multiset $M(P)$ to each process $P$ as follows:
\begin{align*}
M(0) &\coloneq\eset & M(x.P) &\coloneq\left\{x.P\right\} & M(P\parallel Q) &\coloneq M(P)\uplus M(Q),
\end{align*}
where $\uplus$ denotes the multiset union. Observe that $P\sim Q$ if and only if $M(P)=M(Q)$: this can be shown by induction on the congruence, and by observing that $P$ is congruent to the process obtained by composing all the processes in $M(P)$. It is now easy to see that $P\xrightarrow{\sim}Q$ implies $M(P)=M(Q)$, as the isomorphisms cannot delete or create non-zero processes. This shows that (2) implies (1), thus concluding the proof.
\end{proof}

\begin{proposition}[Proposition~\ref{prop:red-string-corresp}]
Let $P$ and $Q$ be processes. There is a term in $\Red(P,{Q\uparrow})$ if and only if the rewrite rule $P\rightarrow Q$ is derivable using the production rules of reduction semantics in Definition~\ref{def:reduction-semantics}.
\end{proposition}
\begin{proof}
To show that we have a term whenever a rewrite rule is derivable (the `if' direction), we show that every deduction rule holds as a recursion scheme for terms. First, the silent action rule~\ref{red-sem:tau} holds as we have the generators $\Sgen : {\tau.P}\rightarrow {P\uparrow}$, and the reduction rule~\ref{red-sem:r} holds since we have the terms \scalebox{.8}{\input{journal-figures/red-sem-r-string.tikz}}.

For the parallel composition rule~\ref{red-sem:p}, suppose that we have a term $\mathtt t:P\rightarrow {Q\uparrow}$. One then obtains the term as required by the rule: \scalebox{.8}{\input{journal-figures/red-sem-p-string.tikz}}.

The congruence rule~\ref{red-sem:c} follows from Lemma~\ref{lma:iso-congruence} by replacing the assumed congruences with isomorphisms in the active and silent fragments.

To show the `only if' direction, i.e.~that every term corresponds to a derivable rewrite rule, we show that every term in $\Red(P,{Q\uparrow})$ is equal to a diagram that can be obtained by the recursion scheme above. Hence suppose that there is a term $\mathtt t:P\rightarrow {Q\uparrow}$. Since there is no way to combine the objects in the silent fragment, we have two cases: (1) $\mathtt t$ contains exactly one instance of the generator $\Sgen : {\tau.R}\rightarrow {R\uparrow}$ and no instances of the generator $\Rgen$, and (2) $\mathtt t$ contains exactly one instance of the generator \scalebox{.8}{\input{journal-figures/red-sem-rgen.tikz}}, and no instances of the generator $\Sgen$.

In Case~(1) we can rewrite the term $\mathtt t$ to have the following form:
\begin{equation}\label{red-sem-generic-form-1}
\scalebox{.8}{\input{journal-figures/red-sem-generic-form-1.tikz}},
\end{equation}
where the vertical dashed lines indicate any number of isomorphisms, that is, symmetries and the structural generators (Figure~\ref{fig:ccs-structural-generators}). The ellipsis stands for a sequence of parallel composition generators, of which there are assumed to be $n$: at step $i$, the process $R_i\uparrow$ in the silent fragment is composed with the $i$th identity wire $T_i$, resulting in ${R_i\parallel T_i}\uparrow$, which is isomorphic to ${R_{i+1}\uparrow}$ (taking ${R_{n+1}\uparrow}={Q\uparrow}$). We conclude Case~(1) by observing that the segments between the vertical dashed lines corresponds to an application of a reduction semantics production rule (either~\ref{red-sem:tau} or~\ref{red-sem:p}), and the dashed lines to an application of the congruence rule~\ref{red-sem:c}.

In Case~(2), the term may likewise be rewritten to have the following form:
\begin{equation}\label{red-sem-generic-form-2}
\scalebox{.8}{\input{journal-figures/red-sem-generic-form-2.tikz}},
\end{equation}
where the structure of the isomorphisms and parallel composition is depicted as in Case~(1). We then read off a derivation in reduction semantics exactly as in Case~(1), except that we replace the silent action rule~\ref{red-sem:tau} with the reduction rule~\ref{red-sem:r}.
\end{proof}

\begin{lemma}[Lemma~\ref{lma:red-term-to-comp}]
Let $\mathtt t\in\Red(P,{Q\uparrow})$ be a term. Then there is a term in the layered theory CCS
\begin{equation*}
\scalebox{.8}{\input{journal-figures/red-iso-standard.tikz}},
\end{equation*}
such that the term $t'\in\Comp(\Red)(P',{Q'\uparrow})$ is standard, and $\mathtt i$ and $\mathtt j$ are isomorphisms.
\end{lemma}
\begin{proof}
By the proof of Proposition~\ref{prop:red-string-corresp}, there are two cases: (1) $\mathtt t$ contains exactly one instance of the generator $\Sgen$ and no instances of the generator $\Rgen$, and (2) $\mathtt t$ contains exactly one instance of the generator $\Rgen$ and no instances of the generator $\Sgen$. In Case~(1), existence of the term~\eqref{red-sem-generic-form-1} implies existence of the following term:
\begin{equation*}
\scalebox{.8}{\input{journal-figures/red-sem-generic-form-1-standard.tikz}},
\end{equation*}
where, as before, the vertical dashed lines stand for some isomorphisms, and we have $n$ parallel composition generators $\Pgen$. This time, the equations $R_{i+1}=R_i\parallel T_i$ hold strictly for all $i=1,\dots,n-1$, and we still have an isomorphism ${R_n\parallel T_n\uparrow}\simeq {Q\uparrow}$. We observe that the term between the vertical dashed lines is in $\Comp(\Red)$ and is standard, so that it is indeed of the required form. This term will not, in general, be equal to the original term, as they correspond to different derivation trees, but they do derive the same rewrite rule.

Case~(2) is nearly identical, except that we replace the generator $\Sgen$ with $\Rgen$.
\end{proof}

\begin{proposition}[Proposition~\ref{prop:lts-bisimulation}]
Let $\mathtt t\in\LTS(P,{Q\uparrow x})$ be a standard term. If there is a process $P'$ with $P\sim P'$, then there is a standard term $\mathtt t'\in\LTS(P',{Q'\uparrow x})$ such that $Q\sim Q'$.
\end{proposition}
\begin{proof}
We argue by induction on construction of the congruence $P\sim P'$.

\noindent\paragraph{Base case~1:} $P={0\parallel R}$ and $P'=R$. Since the term $\mathtt t:{0\parallel R}\rightarrow {Q\uparrow x}$ is standard, it must, in fact, be of the form
\begin{equation*}
\scalebox{.8}{\input{journal-figures/lts-bisimulation-case1.tikz}},
\end{equation*}
so that the sought-after term is given by $\mathtt r:R\rightarrow {R'\uparrow x}$. The argument for the case when $P=R$ and $P'={0\parallel R}$ is symmetric.

\noindent\paragraph{Base case~2:} $P=(R\parallel S)\parallel T$ and $P'=R\parallel (S\parallel T)$. There are two cases for the form of the standard term $\mathtt t:(R\parallel S)\parallel T\rightarrow {Q\uparrow x}$. Either it is of the form
\begin{equation*}
\scalebox{.8}{\input{journal-figures/lts-bisimulation-case2-1.tikz}},
\end{equation*}
or it is of the form
\begin{equation*}
\scalebox{.8}{\input{journal-figures/lts-bisimulation-case2-2.tikz}},
\end{equation*}
where $\mathtt G,\mathtt H\in\Gen$. In the first case, the sought-after term $\mathtt t'$ is simply given by
\begin{equation*}
\scalebox{.8}{\input{journal-figures/lts-bisimulation-case2-1-new.tikz}}.
\end{equation*}
The second case has two subcases: first, when $\mathtt G=\mathtt H=\Pgenr$ (and consequently $\mathtt s$ and $\mathtt u$ are identities), and second, when either $\mathtt G\neq\Pgenr$ or $\mathtt H\neq\Pgenr$. In the first subcase, the sought-after term $\mathtt t'$ is given by
\begin{equation*}
\scalebox{.8}{\input{journal-figures/lts-bisimulation-case2-2-1-new.tikz}}.
\end{equation*}
For the second subcase, we first define the function $(-',-'):\Gen\times\Gen\rightarrow\Gen\times\Gen$ as follows: if $\mathtt G\neq\Pgenr$, then $(\mathtt G',\mathtt H')\coloneq (\mathtt H,\mathtt G)$, and $(\Pgenr',\mathtt H')\coloneq (\Pgenl,\mathtt H)$. With this notation, taking $(\mathtt G,\mathtt H)$ that appear in the term as the input for the function $(-',-')$, the sought-after term $\mathtt t'$ is given by
\begin{equation*}
\scalebox{.8}{\input{journal-figures/lts-bisimulation-case2-new.tikz}}.
\end{equation*}
The argument for the case when $P$ and $P'$ are interchanged is symmetric.

\noindent\paragraph{Base case~3:} $P=R\parallel S$ and $P'=S\parallel R$. In this case, the standard term $\mathtt t:R\parallel S\rightarrow {Q\uparrow x}$ is of the generic form
\begin{equation*}
\scalebox{.8}{\input{journal-figures/lts-bisimulation-case3.tikz}}.
\end{equation*}
Let us define $\Pgenl^*\coloneq\Pgenr$, $\Pgenr^*\coloneq\Pgenl$, $\Rgenone^*\coloneq\Rgentwo$ and $\Rgentwo^*\coloneq\Rgenone$. With this notation, the sought-after term $\mathtt t'$ is given by
\begin{equation*}
\scalebox{.8}{\input{journal-figures/lts-bisimulation-case3-new.tikz}}.
\end{equation*}

\noindent\paragraph{Inductive case:} $P=R\parallel S$ and $P'=R'\parallel S'$, such that $R\sim R'$ and $S\sim S'$, and the statement holds for these instances of the congruence. As in the previous case, the standard term $\mathtt t:R\parallel S\rightarrow {Q\uparrow x}$ is of the generic form
\begin{equation*}
\scalebox{.8}{\input{journal-figures/lts-bisimulation-case4.tikz}}.
\end{equation*}
By the induction hypothesis, there are standard terms $\mathtt r':R'\rightarrow T'$ and $\mathtt s':S'\rightarrow U'$ such that $T'\sim T$ and $U'\sim U$ (note that if the codomain of either $\mathtt r$ or $\mathtt s$ remains in the active fragment, then the term must be the identity). The sought-after term $\mathtt t'$ is thus given by
\begin{equation*}
\scalebox{.8}{\input{journal-figures/lts-bisimulation-case4-new.tikz}},
\end{equation*}
thereby completing the induction.
\end{proof}

\begin{lemma}[Lemma~\ref{lma:lts-term-to-comp}]
Let $\mathtt t\in\LTS(P,{Q\uparrow\tau})$ be a standard term. Then there is a term in the layered theory CCS
\begin{equation*}
\scalebox{.8}{\input{journal-figures/red-iso-standard.tikz}},
\end{equation*}
such that the term $t'\in\Comp(\Red)(P',{Q'\uparrow})$ is standard, and $\mathtt i$ and $\mathtt j$ are isomorphisms.
\end{lemma}
\begin{proof}
There are two cases: either (1) $\mathtt t$ contains the generator $\Agen : \tau.R\rightarrow R\uparrow\tau$ for some process $R$ and all other generators in $\mathtt t$ are instances of the parallel composition generators $\Pgenr$ and $\Pgenl$, or (2) $\mathtt t$ contains exactly one of the generators \scalebox{.8}{\input{journal-figures/lts-reduction-gens.tikz}} for some processes $R$ and $S$ and some $a\in A$, the generators $\Agen : a.R\rightarrow R\uparrow a$ and $\Agen :\bar a.S\rightarrow S\uparrow\bar a$, and all other generators in $\mathtt t$ are instances of the parallel composition generators $\Pgenr$ and $\Pgenl$.

In Case~(1), let us rearrange the processes separated by $\parallel$ in $P$ in such a way that $\tau.R$ is the leftmost process, and the parentheses are associated to the left: $P'\coloneq (\dots((\tau.R\parallel T_1)\parallel T_2)\dots)\parallel T_n$. Since $P'\sim P$, Proposition~\ref{prop:lts-bisimulation} produces a standard term $\mathtt t'\in\LTS(P', {Q'\uparrow\tau})$ such that $Q'\sim Q$. Due to the choice of bracketing, only the right parallel composition generators $\Pgenr$ can appear in $\mathtt t'$ (in addition to the generator $\Agen : \tau.R\rightarrow R\uparrow\tau$). Thus, $\mathtt t'$ is in the image of the translation $\Comp(\Red)\rightarrow\LTS$, producing the required standard term in $\Comp(\Red)$ upon replacing $\Agen$ with $\Sgen$ and each $\Pgenr$ with $\Pgen$.

In Case~(2), we similarly rearrange the processes in $P$ so as to obtain the following $P'$, with $P'\sim P$:
$$P'\coloneq (\dots(a.R\parallel\bar a.S)\parallel T_1)\dots)\parallel T_n.$$
The same argument as in Case~(1) then produces the required term.
\end{proof}

\begin{theorem}[Theorem~\ref{thm:red-lts-term-iff}]
There is a term in $\Red(P,{Q\uparrow})$ if and only if there is a process $Q'$ with $Q\sim Q'$ and a standard term $\LTS(P,{Q'\uparrow\tau})$.
\end{theorem}
\begin{proof}
First, let $\mathtt t\in\Red(P,{Q\uparrow})$ be a term. By Lemma~\ref{lma:red-term-to-comp}, we obtain a standard term $\mathtt t':P'\rightarrow {Q'\uparrow}$ in $\Comp(\Red)$, with $P\sim P'$ and $Q\sim Q'$. Therefore, the translation of $\mathtt t'$ to $\LTS$ gives the standard term $P'\rightarrow {Q'\uparrow\tau}$. By Proposition~\ref{prop:lts-bisimulation}, there is a standard term in $\LTS(P,{Q''\uparrow\tau})$ such that $Q''\sim Q'$, whence also $Q''\sim Q$, as required.

Conversely, let $\mathtt t\in\LTS(P,{Q'\uparrow\tau})$ be a standard term and let $Q\sim Q'$. By Lemma~\ref{lma:lts-term-to-comp}, there is a term
\begin{equation*}
\scalebox{.8}{\input{journal-figures/red-iso-standard-equiv.tikz}},
\end{equation*}
so that we obtain the required term in $\Red(P,{Q\uparrow})$ by composing with the isomorphism ${Q'\uparrow}\simeq {Q\uparrow}$.
\end{proof}

%% file: tex/tocl-article/appendix/para-copara.tex
Here we cover the graphical notation for the {\em (co)para construction}. The construction has appeared several times in the (applied) category theory literature, explicitly under this name in~\cite{fong2019backprop,gavranovic19,cruttwell21}, and it allows morphisms to depend (or be indexed by) objects in the monoidal category at hand.

\begin{definition}\label{def:copara}
Let $\mathcal T = (C,\Sigma,E)$ be a symmetric monoidal theory. We denote by $\Copara(\mathcal T)$ the {\em coparametric monoidal theory over $\mathcal T$}: the colours are given by $C$, the monoidal generators are
\begin{equation*}
\scalebox{.75}{\input{journal-figures/copara-generators.tikz}},
\end{equation*}
while the equations are shown in Figure~\ref{fig:copara-eqns}. Additionally, it will be convenient to omit the box parameterised by the monoidal unit, and to allow ``deparemeterisation'':
\begin{equation*}
\scalebox{.75}{\input{journal-figures/copara-shorthand.tikz}}.
\end{equation*}
\end{definition}
\begin{figure}
  \centering
  \scalebox{.75}{\input{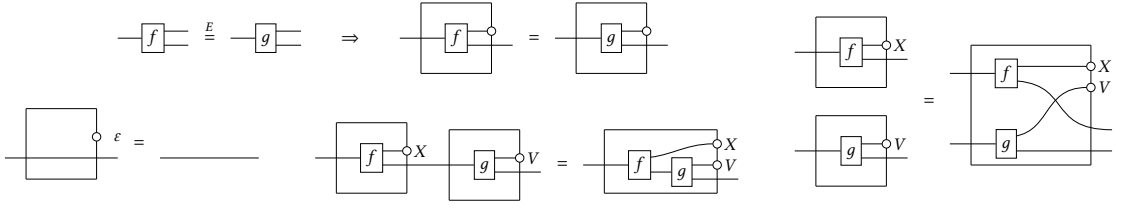}}
  \caption{Equations of the coparametric monoidal theory.\label{fig:copara-eqns}}
\end{figure}

\begin{proposition}
The identity functions $\mathcal T\rightarrow\Copara(\mathcal T)$ on colours and generators give a morphism of monoidal theories.
\end{proposition}

The dual construction, with the generators flipped, is the {\em parametric monoidal theory}, denoted by $\Para(\mathcal T)$. A variant of this construction is given by $\ParaSwap(\mathcal T)$, where the equation defining composition is replaced by
\begin{equation*}
\scalebox{.75}{\input{journal-figures/paraswap-composition.tikz}}.
\end{equation*}

%% file: tex/tocl-article/appendix/zx-mbqc.tex
The basic unit of quantum computation is a {\em qubit}, represented by the Hilbert space $\C^2$. The {\em state} of a qubit is a ray in $\C^2$, and can, therefore, be represented by a linear combination of the standard basis vectors $\ket{0}\coloneq (1,0)$ and $\ket{1}\coloneq (0,1)$, where two linear combinations represent the same state if they are equal up to multiplication by a non-zero scalar. As is standard in the quantum information literature, we represent the basis states in the tensor products of qubits by listing the elements in each component: e.g.~$\ket{00}=\ket{0}\otimes\ket{0}$ is a state in $\C^2\otimes\C^2$. Similarly, the linear maps between qubits are represented by reversing the notation for states: e.g.~$\bra{0}:\C^2\rightarrow\C$ is the linear functional defined by taking the inner product with $\ket{0}$, and $\ketbra{0}{0}:\C^2\rightarrow\C^2$ is the projection to the state $\ket{0}$ (this is known as Dirac's {\em bra-ket notation}). In the context of quantum computing, the standard orthonormal basis $\left\{\ket{0},\ket{1}\right\}$ is often referred to as the {\em computational basis}, to distinguish it from the {\em Hadamard basis} $\left\{\ket{+},\ket{-}\right\}$, defined by $\ket{+}\coloneq\frac{1}{\sqrt 2}\left(\ket{0}+\ket{1}\right)$ and $\ket{-}\coloneq\frac{1}{\sqrt 2}\left(\ket{0}-\ket{1}\right)$. The ZX-calculus can then be seen as an axiomatisation of how these two bases interact.

There is a strict monoidal functor from $\ZX$ (with generators~\eqref{eq:zx-generators} and equations~\eqref{eq:zx-rules}) to the category of finite-dimensional Hilbert spaces given by $n\mapsto\bigotimes_{i=1}^n\C^2$ (the n-fold tensor product of $\C^2$ with itself) on objects, and the Z- and X-spiders are, respectively, mapped to
$$\ketbra{0\cdots 0}{0\cdots 0} + e^{i\alpha}\ketbra{1\cdots 1}{1\cdots 1}\text{ and } \ketbra{+\cdots +}{+\cdots +} + e^{i\alpha}\ketbra{-\cdots -}{-\cdots -}.$$
Note that equation \HadamardDef in~\eqref{eq:zx-rules} then forces the interpretation of the Hadamard gate to be (up to a scalar) the transformation that maps the computational basis to the Hadamard basis: $\ket{0}\mapsto\ket{+}$ and $\ket{1}\mapsto\ket{-}$. As an example of this functorial interpretation, the computational basis states $\ket{0}$ and $\ket{1}$ are now represented by the X-spiders \input{journal-figures/comp-zero.tikz} and \input{journal-figures/comp-one.tikz}, while the Hadamard basis states $\ket{+}$ and $\ket{-}$ are represented by the Z-spiders \input{journal-figures/hadamard-plus.tikz} and \input{journal-figures/hadamard-minus.tikz}.

The above functor is full, corresponding to universality of the calculus: any linear map between finite tensor products of qubits can be expressed as a ZX-diagram. Under the equations we chose to present, the functor is not quite faithful, so the calculus is not complete. However, restricting the diagrams to the {\em Clifford fragment}, i.e.~the diagrams where the phase $\alpha$ is a multiple of $\frac{\pi}{2}$ makes the calculus complete, so long as we ignore the non-zero scalars~\cite{backens-2014}. Adding more equations, the calculus has been shown complete for the Clifford+T fragment~\cite{jeandel-univ-complete}, and finally for the full ZX-calculus~\cite{hadzihasanovic-zx18}.

Recall that $\mathcal P\coloneq\{\XYplane,\XZplane,\YZplane\}$ denotes the set of {\em measurement planes}.
\begin{definition}[Measurement pattern~\cite{measurement-calculus}]\label{def:measurement-pattern}
Let $(V,I,O)$ be a triple of a finite set $V$ with two chosen subsets $I,O\sse V$. The set $\mathcal S$ of {\em signals} consists of the formal sums of $0$, $1$ and $s_i$, where $i\in V\setminus O$. A {\em measurement pattern} over $(V,I,O)$ is an element in the free monoid generated by the set $\left\{N_j, E_{nm}, M^{\lambda,\alpha}_k, X_n^s, Z_n^s\right\}$, where $j\in V\setminus I$, $k\in V\setminus O$, $n,m\in V$ with $n\neq m$, $\lambda\in\mathcal P$, $\alpha\in [0,2\pi)$, and $t,s\in\mathcal S$, and for $C\in\{X,Z\}$ we identify $C_n\coloneq C_n^1$ and $C_n^0=\varepsilon$.
\end{definition}
A measurement pattern on $(V,I,O)$ can be thought of as a sequence of commands that takes in $|I|$ qubits, executes the commands one at a time in the order they appear, and outputs $|O|$ qubits. The intuition behind the commands is as follows (see also Table~\ref{tab:MBQC-to-ZX} for the translation of measurement patterns to the ZX-calculus):
\begin{itemize}
\item $N_j$ prepares the non-input qubit $j$ in the state $\ket{+}$,
\item $E_{nm}$ entangles two distinct qubits $n$ and $m$ by applying a CZ-gate to them,
\item $M^{\lambda,\alpha}_k$ destructively measures the non-output qubit $k$ by projecting it to the Hadamard basis in the plane $\lambda$ in the Bloch sphere shifted by the angle $\alpha$,
\item the corrections $X_n^s$ and $Z_n^s$ act as the Pauli-X and Pauli-Z operators on the qubit $n$ if the signal $s$ evaluates to $1$ in $\Z_2$, and as the identities otherwise.
\end{itemize}
Once a measurement on qubit $k$ is performed, the qubit $k$ is either in the (shifted) state $\ket{+}$, which is modelled by setting the signal $s_k$ to $0$, or in the (shifted) state $\ket{-}$, which is modelled by setting the signal $s_k$ to $1$. A correction is said to {\em depend} on the outcome of measuring a qubit $k$ if the signal $s_k$ appears in its expression.

Clearly, not all measurement patterns are physically realisable: e.g.~no correction should depend on an outcome of a qubit that has not yet been measured, and commands should only act on qubits that have been either prepared or are input qubits. Patterns satisfying such realisability conditions are called {\em runnable}~\cite{browne-gflow,thereandback}. Any runnable pattern can be written in a standard form, without changing the semantics in terms of linear maps~\cite{measurement-calculus,browne-gflow}. In the standard form, the classes of commands appear in the following order: preparations, entanglement commands, measurements, corrections. Note that this corresponds to the physical order in which the commands would be executed in the experimental implementation. We now proceed to define the semantics of patterns.

Given a runnable measurement pattern with $n$ measurement commands, a {\em branch} of the pattern is a sequence of $0$'s and $1$'s of length $n$. The interpretation of a branch is the experimental scenario in which $i$th measurement has yielded the outcome specified by the corresponding entry in the sequence. A {\em partial branch} is a sequence of $0$'s, $1$'s and symbols $x$ of length $n$. The interpretation of a partial branch is that only some of the qubits were measured, with the outcomes specified by the $0$ and $1$ entires in the sequence. Given a partial branch $b=b_1\cdots b_n$, define the {\em support} of $b$ as $\supp(b)\coloneq\{i : b_i\neq x\}$. We say that a partial branch $c$ {\em complements} a branch $b$ if $\supp(b)\cap\supp(c)=\eset$ and $\supp(b)\cup\supp(c)=\{1,\dots,n\}$. We note that a branch is a partial branch with full support, i.e.~with no symbols $x$, and that any two partial branches $b$ and $c$ that complement each other can be uniquely combined into a branch $b*c$.

\begin{definition}[Pattern coarsening]
Given a measurement pattern $P$ and a partial branch $b=b_1\cdots b_n$, we obtain a new pattern $P_b$ by modifying $P$ according to each entry $b_k$ as follows:
\begin{itemize}
\item if $b_k=x$, do nothing,
\item if $b_k=0$, substitute $s_k\mapsto 0$,
\item if $b_k=1$, substitute $M^{\lambda,\alpha}_k\mapsto M^{\lambda,\alpha+\pi}_k$ and $s_k\mapsto 1$.
\end{itemize}
We say that the pattern $P_b$ is the {\em coarsening} of the pattern $P$ by the partial branch $b$.
\end{definition}
Coarsening a runnable pattern by a branch results in a pattern with all signals replaced with formal sums of $0$'s and $1$'s.

Any branch $b$ of a runnable pattern $P$ realises a linear map from $|I|$ qubits to $|O|$ qubits~\cite{measurement-calculus} as follows:
\begin{enumerate}[label={(\arabic*)}]
\item coarsen the pattern, obtaining the pattern $P_b$ with no signal symbols $s_k$,
\item replace each signal in $P_b$ with either $0$ or $1$ by evaluating the formal sums in $\Z_2$,
\item translate the pattern to a ZX-diagram according to the translation procedure in Table~\ref{tab:MBQC-to-ZX},
\item translate the resulting diagram to a linear map.
\end{enumerate}
The conditions for a runnable pattern guarantee that the procedure defined in Table~\ref{tab:MBQC-to-ZX} is well-defined and terminating. Given a runnable measurement pattern $P=(V,I,O)$ and a branch $b$, we denote the linear map $\bigotimes_{i\in I}\C^2\rightarrow\bigotimes_{o\in O}\C^2$ realised by $b$ by $L(P_b)$. We remark that, since the ZX-calculus is universal and can be made complete, the last step of the translation can often be omitted, and we can reason directly with ZX-diagrams rather than linear maps.

 \begin{table}
  \renewcommand{\arraystretch}{2}
  \begin{tabular}{c || c | c | c | c | c | c}
   command / vertex & $i\in I$ & $N_j$ & $E_{nm}$ & $M^{\XYplane,\alpha}_k$ & $M^{\XZplane,\alpha}_k$ & $M^{\YZplane,\alpha}_k$ \\ \hline
   diagram & \input{journal-figures/id-input.tikz} & \input{journal-figures/Z-zero.tikz} & \input{journal-figures/cz.tikz} & \input{journal-figures/XY-effect.tikz} & \input{journal-figures/XZ-effect.tikz} & \input{journal-figures/YZ-effect.tikz}
  \end{tabular}
  \begin{tabular}{c | c | c}
   $X_n$ & $Z_n$ & $o\in O$ \\ \hline
   \input{journal-figures/Pauli-X.tikz} & \input{journal-figures/Pauli-Z.tikz} & \input{journal-figures/id-output.tikz}
  \end{tabular}
  \renewcommand{\arraystretch}{1}
  \caption{Translation from measurement patterns to ZX-diagrams: match the labels on the left to labels that have already been drawn, append the specified diagram, and replace the labels with the ones on the right.\label{tab:MBQC-to-ZX}}
 \end{table}

In general, different branches of a runnable pattern may result in different linear maps. For useful algorithms and computations, one is interested in patterns where the linear maps realised by different branches are related. This gives rise to different flavours of {\em determinism}.
\begin{definition}[Determinism~\cite{measurement-calculus,browne-gflow}]\label{def:determinism}
A runnable measurement pattern $P$ is
\begin{itemize}
\item {\em deterministic} if for any two branches $b$ and $c$, the linear maps they realise are pointwise proportional: for any input $q$, the vectors $L(P_b)(q)$ and $L(P_c)(q)$ are scalar multiples of each other,
\item {\em strongly deterministic} if any two branches result in the same linear map, up to a non-zero scalar,
\item {\em uniformly deterministic} if it is deterministic for any choice of the measurement angles,
\item {\em stepwise deterministic} if for any two partial branches $b$ and $c$ with the same support, there is a sequence $C$ of correction commands, which depend only on the qubits that are not in the common support of $b$ and $c$, making the coarsened patterns $P_b$ and $P_c$ equivalent in the following sense: for any partial branches $b'$ and $c'$ that complement $b$ and $c$, the resulting linear maps $L\left((C\cdot P)_{b*b'}\right)$ and $L\left(P_{c*c'}\right)$ are pointwise proportional.
\end{itemize}
\end{definition}

\begin{example}
Consider the pattern\footnote{Note that the usual order of commands in a measurement pattern is right to left, corresponding to the usual order of composition of operators (functions). We, however, choose to write the commands from left to right, so as to match the diagrammatic order of composition.} $N_2E_{12}M_2^{\XYplane,0}X_1^{s_2}$ acting on $(\{1,2\},\{1\},\{1\})$. It has two branches, corresponding to the measurement outcomes of $\ket{+}$ ($s_2=0$) and $\ket{-}$ ($s_2=1$), resulting in the following ZX-diagrams upon translating according to Table~\ref{tab:MBQC-to-ZX}:
\begin{equation*}
\scalebox{1}{\input{journal-figures/bitflip-correction-1.tikz}}.
\end{equation*}
Simplifying using the equations of the ZX-calculus, we get that the above diagrams are equal to
\begin{equation*}
\scalebox{1}{\input{journal-figures/bitflip-correction-2.tikz}},
\end{equation*}
which are $\ketbra{0}{0}$ and $\ketbra{0}{1}$, so that the outcome state is always proportional to $\ket{0}$. Thus the pattern is deterministic, but not strongly deterministic.
\end{example}

\begin{example}
Consider the pattern $N_2E_{12}M_1^{\XYplane,0}X_2^{s_1}$ acting on $(\{1,2\},\{1\},\{2\})$. The two branches result in the following diagrams:
\begin{equation*}
\scalebox{1}{\input{journal-figures/hadamard-pattern.tikz}},
\end{equation*}
both of which simplify to the Hadamard gate. Thus the pattern is strongly deterministic, and each branch implements the Hadamard gate.
\end{example}

%% file: figures/journal-figures/comp-zero.tikz
\begin{tikzpicture}
	\begin{pgfonlayer}{nodelayer}
		\node [style=none] (1) at (1, 0) {};
		\node [style=X dot] (2) at (0, 0) {};
	\end{pgfonlayer}
	\begin{pgfonlayer}{edgelayer}
		\draw (2) to (1.center);
	\end{pgfonlayer}
\end{tikzpicture}

%% file: figures/journal-figures/comp-one.tikz
\begin{tikzpicture}
	\begin{pgfonlayer}{nodelayer}
		\node [style=none] (1) at (1, 0) {};
		\node [style=X phase dot] (3) at (0, 0) {$\pi$};
	\end{pgfonlayer}
	\begin{pgfonlayer}{edgelayer}
		\draw (3) to (1.center);
	\end{pgfonlayer}
\end{tikzpicture}

%% file: figures/journal-figures/hadamard-plus.tikz
\begin{tikzpicture}
	\begin{pgfonlayer}{nodelayer}
		\node [style=Z dot] (0) at (0, 0) {};
		\node [style=none] (1) at (1, 0) {};
	\end{pgfonlayer}
	\begin{pgfonlayer}{edgelayer}
		\draw (0) to (1.center);
	\end{pgfonlayer}
\end{tikzpicture}

%% file: figures/journal-figures/hadamard-minus.tikz
\begin{tikzpicture}
	\begin{pgfonlayer}{nodelayer}
		\node [style=Z phase dot] (0) at (0, 0) {$\pi$};
		\node [style=none] (1) at (1, 0) {};
	\end{pgfonlayer}
	\begin{pgfonlayer}{edgelayer}
		\draw (0) to (1.center);
	\end{pgfonlayer}
\end{tikzpicture}

%% file: bibliography.bib
@incollection{JacobsZanasi-essentials,
  author       = {Bart Jacobs and
                  Fabio Zanasi},
  editor       = {Gilles Barthe and
                  Joost{-}Pieter Katoen and
                  Alexandra Silva},
  title        = {The Logical Essentials of Bayesian Reasoning},
  booktitle    = {Foundations of Probabilistic Programming},
  pages        = {295--332},
  publisher    = {Cambridge University Press},
  year         = {2020},
  url          = {https://doi.org/10.1017/9781108770750.010},
  doi          = {10.1017/9781108770750.010},
  address      = {Cambridge},
}

@article{kissinger19,
   title={A categorical semantics for causal structure},
   volume={Volume 15, Issue 3},
   ISSN={1860-5974},
   url={http://dx.doi.org/10.23638/LMCS-15(3:15)2019},
   DOI={10.23638/lmcs-15(3:15)2019},
   journal={Logical Methods in Computer Science},
   publisher={Centre pour la Communication Scientifique Directe (CCSD)},
   author={Kissinger, Aleks and Uijlen, Sander},
   year={2019},
   pages={15:1--15:48},
}

@article{JacobsKZ21,
  author       = {Bart Jacobs and
                  Aleks Kissinger and
                  Fabio Zanasi},
  title        = {Causal inference via string diagram surgery: {A} diagrammatic approach
                  to interventions and counterfactuals},
  journal      = {Math. Struct. Comput. Sci.},
  volume       = {31},
  number       = {5},
  pages        = {553--574},
  year         = {2021},
  url          = {https://doi.org/10.1017/S096012952100027X},
  doi          = {10.1017/S096012952100027X},
  bibsource    = {dblp computer science bibliography, https://dblp.org}
}

@misc{krivine-talk,
      title={Physical systems, composite explanations and diagrams}, 
      author={Jean Krivine},
      year={2019},
      series={SYCO 5 / STRINGS 3},
      type={Conference talk},
      NOTE = "\url{https://youtu.be/w8vTTvA0HTE} and 
              \url{https://www.cl.cam.ac.uk/events/syco/strings3-syco5/}",
}

@Inbook{survey-signal-flow,
author="Bonchi, Filippo
and Soboci{\'{n}}ski, Pawe{\l}
and Zanasi, Fabio",
title="A Survey of Compositional Signal Flow Theory",
bookTitle="Advancing Research in Information and Communication Technology: IFIP's Exciting First 60+ Years, Views from the Technical Committees and Working Groups",
year="2021",
publisher="Springer International Publishing",
address="Cham",
pages="29--56",
isbn="978-3-030-81701-5",
doi="10.1007/978-3-030-81701-5\_2",
}

@INPROCEEDINGS{graphical-affine-algebra,
  author={Bonchi, Filippo and Piedeleu, Robin and Sobociński, Pawel and Zanasi, Fabio},
  booktitle={2019 34th Annual ACM/IEEE Symposium on Logic in Computer Science (LICS)}, 
  title={Graphical Affine Algebra}, 
  year={2019},
  pages={1-12},
  doi={10.1109/LICS.2019.8785877},
  publisher={IEEE},
  address={Vancouver, BC, Canada},
}

@inproceedings{functorial-boxes,
author = {Paul-André Melliès},
address = {Berlin, Heidelberg},
booktitle = {Computer Science Logic},
copyright = {Springer-Verlag Berlin Heidelberg 2006},
isbn = {9783540454588},
issn = {0302-9743},
language = {eng},
pages = {1-30},
publisher = {Springer Berlin Heidelberg},
series = {Lecture Notes in Computer Science},
title = {Functorial Boxes in String Diagrams},
volume = {4207},
year = {2006},
doi={10.1007/11874683\_1},
}

@article{roman,
	doi = {10.4204/eptcs.333.5},
  
	year = 2021,
	month = {feb},
  
	publisher = {Open Publishing Association},
  
	volume = {333},
  
	pages = {65--78},
  
	author = {Mario Rom{\'{a}}n},
  
	title = {Open Diagrams via Coend Calculus},
  
	journal = {Electronic Proceedings in Theoretical Computer Science}
}

@article{krivine-siglog,
    author = {Krivine, Jean},
    title = {Systems Biology},
    year = {2017},
    publisher = {Association for Computing Machinery},
    address = {New York, NY, USA},
    volume = {4},
    number = {3},
    doi = {10.1145/3129173.3129182},
    journal = {ACM SIGLOG News},
    month = {jul},
    pages = {43–61},
    numpages = {19}
    }

@mastersthesis{hu-thesis,
  title = {External traced monoidal categories},
  author = {Nick Hu},
  year = {2019},
  institution = {University of Oxford},
  school = {Department of Computer Science},
  url = {https://www.cs.ox.ac.uk/files/11696/project.pdf},
}

@misc{modular-categories,
  doi = {10.48550/ARXIV.1509.06811},
  
  author = {Bartlett, Bruce and Douglas, Christopher L. and Schommer-Pries, Christopher J. and Vicary, Jamie},
  
  title = {Modular categories as representations of the 3-dimensional bordism 2-category},
  
  publisher = {arXiv},
  
  year = {2015},
  
  copyright = {arXiv.org perpetual, non-exclusive license}
}

@misc{electrical-circuits,
  doi = {10.48550/ARXIV.2106.07763},
  
  author = {Boisseau, Guillaume and Sobociński, Paweł},
  
  title = {String Diagrammatic Electrical Circuit Theory},
  
  publisher = {arXiv},
  
  year = {2021},
}

@misc{compositional-networks,
  doi = {10.48550/ARXIV.1504.05625},
  
  author = {Baez, John C. and Fong, Brendan},
  
  title = {A Compositional Framework for Passive Linear Networks},
  
  publisher = {arXiv},
  
  year = {2015},
}

@book{rosen-life-itself,
author = {Rosen, Robert},
address = {New York},
isbn = {0231075642},
publisher = {Columbia University Press},
series = {Complexity in ecological systems series},
title = {Life itself : a comprehensive inquiry into the nature, origin, and fabrication of life},
year = {1991},
}

@inproceedings{counterfactual,
author = {Laurent, Jonathan and Yang, Jean and Fontana, Walter},
title = {Counterfactual resimulation for causal analysis of rule-based models},
year = {2018},
isbn = {9780999241127},
publisher = {AAAI Press},
booktitle = {Proceedings of the 27th International Joint Conference on Artificial Intelligence},
pages = {1882–1890},
numpages = {9},
series = {IJCAI'18},
doi = {10.5555/3304889.3304920},
address = {Stockholm, Sweden},
}

@inbook{pearl-causality,
address={Cambridge},
edition={2},
title={The Logic of Structure-Based Counterfactuals},
DOI={10.1017/CBO9780511803161.009},
booktitle={Causality},
publisher={Cambridge University Press},
author={Pearl, Judea},
year={2009},
pages={201–258}
}

@book{milner,
author = {Milner, Robin},
address = {Berlin, Heidelberg},
copyright = {Springer-Verlag Berlin Heidelberg 1980},
isbn = {9783540102359},
issn = {0302-9743},
language = {eng},
organization = {SpringerLink (Online service)},
publisher = {Springer Berlin Heidelberg},
series = {Lecture Notes in Computer Science},
title = {A Calculus of Communicating Systems},
volume = {92},
year = {1980},
doi = {10.1007/3-540-10235-3},
}

@phdthesis{zanasi-thesis,
  doi = {10.48550/ARXIV.1805.03032},
  
  author = {Zanasi, Fabio},
  
  title = {Interacting Hopf Algebras: the theory of linear systems},
  
  institution = "\' Ecole Normale Sup\' erieure de Lyon",
  
  school = "Laboratoire de l’Informatique du Parall\' elisme",
  
  year = {2015},
  
}

@article{selinger,
   title={A Survey of Graphical Languages for Monoidal Categories},
   ISBN={9783642128219},
   ISSN={1616-6361},
   journal={Lecture Notes in Physics},
   publisher={Springer Berlin Heidelberg},
   author={Selinger, Peter},
   year={2010},
   doi={10.1007/978-3-642-12821-9\_4},
   volume={813},
   pages="289--355",
}

@article{lobski-zanasi,
title={String Diagrams for Layered Explanations},
volume={380},
ISSN={2075-2180},
DOI={10.4204/eptcs.380.21},
journal={Electronic Proceedings in Theoretical Computer Science},
publisher={Open Publishing Association},
author={Lobski, Leo and Zanasi, Fabio},
year={2023},
pages={362–382}
}

@article{modular1999,
author = {Hartwell, Leland H. and Hopfield, John J. and Leibler, Stanislas and Murray, Andrew W.},
address = {LONDON},
issn = {0028-0836},
journal = {Nature (London)},
language = {eng},
number = {6761},
pages = {C47-C52},
publisher = {MACMILLAN MAGAZINES LTD},
title = {From molecular to modular cell biology},
volume = {402},
year = {1999},
doi = {10.1038/35011540},
}

@article{blaming2015,
author = {Goessler, Gregor and Le Metayer, Daniel},
address = {AMSTERDAM},
issn = {0167-6423},
journal = {Science of computer programming},
language = {eng},
number = {Part 3},
pages = {223-235},
publisher = {Elsevier B.V},
title = {A general framework for blaming in component-based systems},
volume = {113},
year = {2015},
doi = {10.1016/j.scico.2015.06.010hal-01211484},
}

@misc{zx-for2020,
  eprint = {2012.13966},
  author = {van de Wetering, John},
  title = {{ZX}-calculus for the working quantum computer scientist},
  year = {2020},
}

@article{thereandback,
author = {Backens, Miriam and Miller-Bakewell, Hector and de Felice, Giovanni and Lobski, Leo and van de Wetering, John},
publisher = {{Verein zur F{\"{o}}rderung des Open Access Publizierens in den Quantenwissenschaften}},
title = {There and back again: {A} circuit extraction tale},
year = {2021},
journal = {{Quantum}},
issn = {2521-327X},
volume = {5},
pages = {421},
doi = {10.22331/q-2021-03-25-421},
}

@article{chem-trans-motifs,
author = {Andersen, Jakob L. and Flamm, Christoph and Merkle, Daniel and Stadler, Peter F.},
issn = {1545-5963},
journal = {IEEE/ACM transactions on computational biology and bioinformatics},
number = {2},
pages = {510-523},
publisher = {IEEE},
title = {Chemical Transformation Motifs -- Modelling Pathways as Integer Hyperflows},
volume = {16},
year = {2019},
doi = {10.1109/TCBB.2017.2781724}
}

@article{monoidal-gro,
  doi = {10.70930/tac/4tsjzc1o},
  author = {Moeller, Joe and Vasilakopoulou, Christina},  
  title = {Monoidal Grothendieck construction},
  journal = {Theory and Applications of Categories},
  volume = {35},
  number = {31},
  pages = {1159-1207},
  year = {2020}
}

@book{heunen-vicary-book,
    author = {Heunen, Chris and Vicary, Jamie},
    title = "{Categories for Quantum Theory: An Introduction}",
    publisher = {Oxford University Press},
    year = {2019},
    isbn = {9780198739623},
    doi = {10.1093/oso/9780198739623.001.0001},
    address = {Oxford},
}

@book{piedeleu-zanasi,
series={Elements in Applied Category Theory},
title={An Introduction to String Diagrams for Computer Scientists},
publisher={Cambridge University Press},
author={Piedeleu, Robin and Zanasi, Fabio},
year={2025},
doi={10.1017/9781009625715},
address={Cambridge},
}

@misc{fong2019backprop,
      title={Backprop as Functor: A compositional perspective on supervised learning}, 
      author={Brendan Fong and David I. Spivak and Rémy Tuyéras},
      year={2019},
      doi={10.48550/arXiv.1711.10455}, 
}

@article{tcs-cat-model-org-chem,
title = {A categorical model for organic chemistry},
journal = {Theoretical Computer Science},
volume = {1032},
pages = {115084},
year = {2025},
issn = {0304-3975},
doi = {10.1016/j.tcs.2025.115084},
author = {Ella Gale and Leo Lobski and Fabio Zanasi},
}

@misc{kaye-thesis,
      title={Foundations of Digital Circuits: Denotation, Operational, and Algebraic Semantics}, 
      author={George Kaye},
      year={2025},
      doi={10.48550/arXiv.2502.08497},
}

@misc{lorenz-tull-causal,
      title={Causal models in string diagrams}, 
      author={Robin Lorenz and Sean Tull},
      year={2023},
      eprint={2304.07638},
}

@article{fritz-markov20,
title = {A synthetic approach to Markov kernels, conditional independence and theorems on sufficient statistics},
journal = {Advances in Mathematics},
volume = {370},
pages = {107239},
year = {2020},
issn = {0001-8708},
doi = {10.1016/j.aim.2020.107239},
author = {Tobias Fritz},
}

@misc{jacobs-spr,
      title={Structured Probabilistic Reasoning}, 
      author={Bart Jacobs},
      year={2025},
      url={https://www.cs.ru.nl/B.Jacobs/PAPERS/ProbabilisticReasoning.pdf},
}

@misc{discocat,
      title={Mathematical Foundations for a Compositional Distributional Model of Meaning}, 
      author={Bob Coecke and Mehrnoosh Sadrzadeh and Stephen Clark},
      year={2010},
      eprint={1003.4394},
}

@article{hu-vicary21,
   title={Traced Monoidal Categories as Algebraic Structures in Prof},
   volume={351},
   ISSN={2075-2180},
   DOI={10.4204/eptcs.351.6},
   journal={Electronic Proceedings in Theoretical Computer Science},
   publisher={Open Publishing Association},
   author={Hu, Nick and Vicary, Jamie},
   year={2021},
   pages={84–97}
}

@article{braithwaite-roman23,
   title={Collages of String Diagrams},
   volume={397},
   ISSN={2075-2180},
   DOI={10.4204/eptcs.397.3},
   journal={Electronic Proceedings in Theoretical Computer Science},
   publisher={Open Publishing Association},
   author={Braithwaite, Dylan and Román, Mario},
   year={2023},
   pages={39–53},
}

@article{ponto-shulman12,
title = {Duality and traces for indexed monoidal categories},
journal = {Theory and Applications of Categories},
volume = {26},
number = {23},
pages = {582-659},
year = {2012},
doi = {10.70930/tac/2ckgqep9 },
author = {Kate Ponto and Michael Shulman},
}

@misc{earnshaw-roman24,
      title={Context-Free Languages of String Diagrams}, 
      author={Matt Earnshaw and Mario Román},
      year={2024},
      eprint={2404.10653},
}

@article{hefford-comfort23,
   title={Coend Optics for Quantum Combs},
   volume={380},
   ISSN={2075-2180},
   DOI={10.4204/eptcs.380.4},
   journal={Electronic Proceedings in Theoretical Computer Science},
   publisher={Open Publishing Association},
   author={Hefford, James and Comfort, Cole},
   year={2023},
   pages={63–76},
}

@article{fritz-rischel20,
    title      = {Infinite products and zero-one laws in categorical probability},
    author     = {Tobias Fritz and Eigil Fjeldgren Rischel},
    doi        = {10.32408/compositionality-2-3},
    journal    = {Compositionality},
    issn       = {2631-4444},
    volume     = {Volume 2 (2020)},
    year       = {2020},
    numpages   = {20},
}

@misc{quantitative-monoidal-algebra,
      title={Quantitative Monoidal Algebra: Axiomatising Distance with String Diagrams}, 
      author={Gabriele Lobbia and Wojciech Różowski and Ralph Sarkis and Fabio Zanasi},
      year={2025},
      eprint={2410.09229},
}

@phdthesis{boisseau-thesis,
author={Boisseau, Guillaume},
year={2023},
title={Graphical Electrical Circuit Theory},
isbn={9798302876584},
institution={University of Oxford},
school = {Department of Computer Science},
url={https://www.proquest.com/dissertations-theses/graphical-electrical-circuit-theory/docview/3165317412/se-2},
}

@mastersthesis{dunn-thesis,
  title = {Profunctor Semantics for Linear Logic},
  author = {Dunn~III, Lawrence H.},
  year = {2015},
  institution = {University of Oxford},
  school = {Department of Computer Science},
  url = {https://www.cs.ox.ac.uk/people/aleks.kissinger/theses/bob/Dunn.pdf},
}

@misc{comfort-profunctors-linear-logic,
      title={Notes on profunctors and compact multiplicative linear logic},
      author={Comfort, Cole},
      year={2024},
      url={https://colecomfort.github.io/pdfs/notes.pdf},
}

@article{browne-gflow,
doi = {10.1088/1367-2630/9/8/250},
year = {2007},
volume = {9},
number = {8},
author = {Browne, Daniel E and Kashefi, Elham and Mhalla, Mehdi and Perdrix, Simon},
title = {Generalized flow and determinism in measurement-based quantum computation},
journal = {New Journal of Physics},
pages={250},
}

@article{measurement-calculus,
author = {Danos, Vincent and Kashefi, Elham and Panangaden, Prakash},
title = {The measurement calculus},
year = {2007},
publisher = {Association for Computing Machinery},
volume = {54},
number = {2},
issn = {0004-5411},
doi = {10.1145/1219092.1219096},
journal = {J. ACM},
pages = {8–es},
}

@incollection{extended-meas-calculus,
title={Extended Measurement Calculus},
booktitle={Semantic Techniques in Quantum Computation},
publisher={Cambridge University Press},
author={Danos, Vincent and Kashefi, Elham and Panangaden, Prakash and Perdrix, Simon},
editor={Gay, Simon and Mackie, Ian},
year={2009},
pages={235-310},
doi={10.1017/CBO9781139193313.008},
address={Cambridge},
}

@article{backens-2014,
doi = {10.1088/1367-2630/16/9/093021},
year = {2014},
publisher = {IOP Publishing},
volume = {16},
number = {9},
author = {Backens, Miriam},
pages = {093021},
title = {The ZX-calculus is complete for stabilizer quantum mechanics},
journal = {New Journal of Physics},
}

@article{jeandel-univ-complete,
    title      = {Completeness of the ZX-Calculus},
    author     = {Emmanuel Jeandel and Simon Perdrix and Renaud Vilmart},
    doi        = {10.23638/LMCS-16(2:11)2020},
    journal    = {Logical Methods in Computer Science},
    issn       = {1860-5974},
    volume     = {16},
    issue      = {2},
    year       = {2020},
    pages      = {11:1--11:72},
}

@inproceedings{hadzihasanovic-zx18,
author = {Hadzihasanovic, Amar and Ng, Kang Feng and Wang, Quanlong},
title = {Two complete axiomatisations of pure-state qubit quantum computing},
year = {2018},
isbn = {9781450355834},
publisher = {Association for Computing Machinery},
doi = {10.1145/3209108.3209128},
booktitle = {Proceedings of the 33rd Annual ACM/IEEE Symposium on Logic in Computer Science},
pages = {502-511},
series = {LICS '18},
address = {New York, NY, USA},
}

@misc{gavranovic19,
      title={Compositional Deep Learning}, 
      author={Bruno Gavranović},
      year={2019},
      eprint={1907.08292},
}

@misc{cruttwell21,
      title={Categorical Foundations of Gradient-Based Learning}, 
      author={G. S. H. Cruttwell and Bruno Gavranović and Neil Ghani and Paul Wilson and Fabio Zanasi},
      year={2021},
      eprint={2103.01931},
}

@article{shulman2008,
title = {Framed bicategories and monoidal fibrations},
journal = {Theory and Applications of Categories},
volume = {20},
number = {18},
pages = {650-738},
year = {2008},
doi = {10.70930/tac/2m83wy59},
author = {Michael Shulman},
}

@article{hofstra-demarchi2006,
title = {Descent for monads},
journal = {Theory and Applications of Categories},
volume = {16},
number = {24},
pages = {668-699},
year = {2006},
doi = {10.70930/tac/nfs954c1},
author = {Pieter Hofstra and Federico De Marchi},
}

@article{mellies-cat-semantics-linear-logic,
    title      = {Categorical semantics of linear logic},
    author     = {Paul-André Melliès},
    journal    = {Panoramas \& Synth\` eses},
    volume     = {27},
    year       = {2009},
    pages      = {1-197},
    url        = {https://www.irif.fr/~mellies/mpri/mpri-ens/biblio/categorical-semantics-of-linear-logic.pdf},
}

@article{jacobs-szeles-stein25,
   title={Compositional Inference for Bayesian Networks and Causality},
   volume={Volume 5 -- Proceedings of {MFPS} {XLI}},
   ISSN={2969-2431},
   doi={10.46298/entics.17029},
   journal={Electronic Notes in Theoretical Informatics and Computer Science},
   publisher={Centre pour la Communication Scientifique Directe (CCSD)},
   author={Jacobs, Bart and Széles, Márk and Stein, Dario},
   year={2025},
   numpages={21},
}

@misc{lmt-part2,
title={Layered Monoidal Theories {II}: Fibrational Semantics}, 
      author={Leo Lobski and Fabio Zanasi},
      year={2026},
      eprint={2602.22373},
}

@book{hinze-marsden-2023,
address={Cambridge},
title={Introducing String Diagrams: The Art of Category Theory},
publisher={Cambridge University Press},
author={Hinze, Ralf and Marsden, Dan},
year={2023},
doi={10.1017/9781009317825},
}

@misc{nashed2023,
      title={Causal Explanations for Sequential Decision Making Under Uncertainty}, 
      author={Samer B. Nashed and Saaduddin Mahmud and Claudia V. Goldman and Shlomo Zilberstein},
      year={2023},
      eprint={2205.15462}, 
}

@misc{englberger2025,
      title={Causal Abstractions, Categorically Unified}, 
      author={Markus Englberger and Devendra Singh Dhami},
      year={2025},
      eprint={2510.05033},
}

@misc{fritz2009,
      title={A presentation of the category of stochastic matrices}, 
      author={Tobias Fritz},
      year={2009},
      eprint={0902.2554},
}

@misc{complete-disc-prob-prog2024,
      title={A Complete Axiomatisation of Equivalence for Discrete Probabilistic Programming}, 
      author={Robin Piedeleu and Mateo Torres-Ruiz and Alexandra Silva and Fabio Zanasi},
      year={2024},
      eprint={2408.14701}, 
}

@misc{tapesstochasticmatrices2026,
      title={Tapes as Stochastic Matrices of String Diagrams}, 
      author={Filippo Bonchi and Cipriano Junior Cioffo},
      year={2026},
      eprint={2601.01472},
}

@misc{causalcompositionalabstraction2026,
      title={Causal and Compositional Abstraction}, 
      author={Robin Lorenz and Sean Tull},
      year={2026},
      eprint={2602.16612},
}

@misc{graphicalsymplecticalgebra24,
      title={Graphical Symplectic Algebra}, 
      author={Robert I. Booth and Titouan Carette and Cole Comfort},
      year={2024},
      eprint={2401.07914},
}

@inproceedings{szx-calculus19,
  doi = {10.4230/LIPICS.MFCS.2019.55}, 
  author = {Carette, Titouan and Horsman, Dominic and Perdrix, Simon},
  title = {SZX-Calculus: Scalable Graphical Quantum Reasoning},
  booktitle = {44th International Symposium on Mathematical Foundations of Computer Science (MFCS 2019)},
  journal = {LIPIcs, Volume 138, MFCS 2019},
  volume = {138},
  pages = {55:1-55:15},
  publisher = {Schloss Dagstuhl -- Leibniz-Zentrum für Informatik},
  year = {2019},
}

@article{tape-diagrams23,
author = {Bonchi, Filippo and Di Giorgio, Alessandro and Santamaria, Alessio},
title = {Deconstructing the Calculus of Relations with Tape Diagrams},
year = {2023},
publisher = {Association for Computing Machinery},
address = {New York, NY, USA},
volume = {7},
number = {POPL},
doi = {10.1145/3571257},
journal = {Proc. ACM Program. Lang.},
numpages = {31},
}

@InProceedings{diagrammatic-program-logics25,
author="Bonchi, Filippo and Di Giorgio, Alessandro and Di Lavore, Elena",
editor="Abdulla, Parosh Aziz and Kesner, Delia",
title="A Diagrammatic Algebra for Program Logics",
booktitle="Foundations of Software Science and Computation Structures",
year="2025",
publisher="Springer Nature Switzerland",
pages="308--330",
isbn="978-3-031-90897-2",
address="Cham",
}

@InProceedings{tape-monoidal-monads25,
  author =	{Bonchi, Filippo and Cioffo, Cipriano Junior and Di Giorgio, Alessandro and Di Lavore, Elena},
  title =	{{Tape Diagrams for Monoidal Monads}},
  booktitle =	{11th Conference on Algebra and Coalgebra in Computer Science (CALCO 2025)},
  pages =	{11:1--11:24},
  series =	{Leibniz International Proceedings in Informatics (LIPIcs)},
  ISBN =	{978-3-95977-383-6},
  ISSN =	{1868-8969},
  year =	{2025},
  volume =	{342},
  editor =	{C\^{i}rstea, Corina and Knapp, Alexander},
  publisher =	{Schloss Dagstuhl -- Leibniz-Zentrum f{\"u}r Informatik},
  doi =		{10.4230/LIPIcs.CALCO.2025.11},
  address = {Dagstuhl},
}

@misc{lobski-phdthesis,
      title={Layered Monoidal Theories}, 
      author={Leo Lobski},
      year={2025},
      eprint={2512.12139},
}
